\documentclass[twoside]{article}

\usepackage[accepted]{aistats2023}
\usepackage[round]{natbib}

\usepackage{graphicx}
\usepackage{amsmath,amsthm}
\usepackage{amsfonts}
\usepackage{bm}
\usepackage{algorithmic,algorithm}
\usepackage{soul}
\usepackage{booktabs}
\usepackage{hyperref}
\usepackage{cleveref}
\usepackage{xcolor}
\hypersetup{colorlinks,linkcolor={blue},citecolor={green!50!black},urlcolor={blue}}
\usepackage{dblfloatfix}
\usepackage{paralist}
\usepackage{subfig}

\DeclareMathOperator*{\argmax}{arg\,max}
\DeclareMathOperator*{\argmin}{arg\,min}

\newtheorem{assumption}{Assumption}

\newtheorem{example}{Example}
\newtheorem{proposition}{Proposition}
\newtheorem{lemma}{Lemma}
\newtheorem{theorem}{Theorem}

\renewcommand{\cal}[1]{\mathcal{#1}}

\renewcommand{\rm}[1]{\mathrm{#1}}
\renewcommand{\sf}[1]{\mathsf{#1}}

\renewcommand{\r}{\mathbb{R}}

\newcommand{\rn}{\mathbb{R}^n}

\newcommand{\mmag}[1]{\left|#1\right|}
\newcommand{\norm}[1]{\left|\left|{#1}\right|\right|}
\newcommand{\iprod}[2]{\left\langle{#1},{#2}\right\rangle}

\newcommand{\Ebb}[1]{\mathbb{E}\left[#1\right]}

\usepackage{minitoc}

\begin{document}
\doparttoc % Tell to minitoc to generate a toc for the parts
\faketableofcontents % Run a fake tableofcontents command for the partocs



\twocolumn[

\aistatstitle{Particle algorithms for maximum likelihood training of latent variable models}

\aistatsauthor{ Juan Kuntz \And Jen Ning Lim \And Adam M.\ Johansen}

\aistatsaddress{Department of Statistics, University of Warwick.} ]

\begin{abstract}
\citet{Neal1998} recast maximum likelihood estimation of any given latent variable model as the minimization of a free energy functional $F$, and the EM algorithm as coordinate descent applied to $F$. Here, we explore  alternative ways to optimize the functional. In particular, we identify various gradient flows associated with $F$ and show that their limits coincide with $F$'s stationary points. By  discretizing the flows, we obtain practical particle-based algorithms for maximum likelihood estimation in broad classes of latent variable models. The novel algorithms scale to high-dimensional settings and perform well in numerical experiments.
\end{abstract}
\section{INTRODUCTION}\label{sec:intro}
In machine learning and statistics, we often use a probabilistic model, $p_\theta(x,y)$, defined in terms of a vector of  parameters, $\theta$, to infer some quantities, $x$, that we cannot observe experimentally from some, $y$, that we can.  A pragmatic middle ground between Bayesian and frequentist approaches to this type of problem is the \emph{empirical Bayes} (EB) paradigm~\citep{Robbins1955} wherein we
\begin{compactenum}
\item[(S1)] learn the parameters from the data: we search for parameters $\theta_*$ that explain the data $y$ well;
\item[(S2)] use  $\theta_*$ to infer, and quantify the uncertainty in, $x$.
\end{compactenum}
Because this approach does not require eliciting a prior over the parameters, it is particularly appealing for models whose parameters lack physical interpretations or meaningful prior information; e.g.\ the generator network in Sec.~\ref{sec:gen}.  Steps (S1,2) are typically reformulated technically as
\begin{compactenum}
\item[(S1)] find a $\theta_*$ maximizing the \emph{marginal likelihood},  
$$p_\theta(y):=\int p_\theta(x,y)dx;$$
\item[(S2)] obtain the corresponding \emph{posterior distribution}, 
$$p_{\theta_*}(x|y):=\frac{p_{\theta_*}(x,y)}{p_{\theta_*}(y)}.$$
\end{compactenum}
Perhaps the most well-known method for tackling (S1,2) is the \emph{expectation maximization} (EM) algorithm~\citep{Dempster1977}: starting from an initial guess $\theta_0$, alternate,% for $k \geq 0$:
\begin{compactenum}
\item[(E)] compute $q_{k}:=p_{\theta_{k}}(\cdot|y)$,
\item[(M)] solve for $\theta_{k+1}:=\argmax_{\theta\in\Theta}\int \ell(\theta,x)q_{k+1}(x)dx$,
\end{compactenum}
where $\ell(\theta,x):=\log(p_\theta(x,y))$ denotes the log-likelihood. 
Under general conditions~\citep[Chap.~3]{McLachlan2007}, $\theta_k$ converges  to a stationary point $\theta_*$ of the marginal likelihood and $q_k$ to the corresponding posterior $p_{\theta_*}(\cdot|y)$. In cases where the above steps are not analytically tractable, it is common to approximate (E) using Monte Carlo (or Markov chain Monte Carlo if $p_\theta(\cdot|y)$ cannot be sampled directly) and (M) using numerical optimization (e.g.\ with a single gradient or Newton step in  Euclidean spaces); cf.~\citet{Wei1990,Kuk1997,Delyon1999,Younes1999,Kuhn2004,Han2017,Qiu2020,Cai2010,Nijkamp2020,Debortoli2021}.

Here, we take a different approach that builds on an insightful observation made by~\citet{Neal1998} (see \citet{Csiszar1984} for a precedent): EM can be recast as a well-known optimization routine applied to a certain objective. The objective is the `free energy':
\begin{align}F(\theta,q):
&=\int \log(q(x))q(x)dx-\int \ell(\theta,x) q(x)dx\label{eq:vfe}\end{align}
for all $(\theta,q)$ in  $\Theta\times\cal{P}(\cal{X})$, where $\Theta$ denotes the parameter space and $\cal{P}(\cal{X})$ the space of probability distributions over the latent space $\cal{X}$. The optimization routine is coordinate descent: starting from an initial guess $\theta_0$, alternate,
\begin{compactenum}
\item[(E)] solve for $q_{k}:=\argmin_{q\in\cal{P}(\cal{X})}F(\theta_{k},q)$,
\item[(M)] solve for $\theta_{k+1}:=\argmin_{\theta\in\Theta}F(\theta,q_{k})$.
\end{compactenum}
The key here is the following result associating the maxima of $p_\theta(y)$  with the minima of $F$:
\begin{theorem}\label{thrm:vfe}For any $\theta$ in $\Theta$, the posterior $p_\theta(\cdot|y):=p_\theta(\cdot,y)/p_\theta(y)$ minimizes $q\mapsto F(\theta,q)$. 
Moreover, $p_\theta(y)$ has a global maximum at $\theta$ if and only if $F$ has a global minimum at $(\theta,p_\theta(\cdot|y))$.
\end{theorem}

The theorem follows easily from the same type of arguments as those used to prove \citet[Lem.~1, Thrm.~2]{Neal1998}. Similar statements can also be made for local optima, but we refrain from doing so here because it involves specifying what we mean by `local' in $\Theta\times \cal{P}(\cal{X})$. The point is that finding a maximum of $p_\theta(y)$ and computing the corresponding posterior is equivalent to finding a minimum of $F$, and this is precisely what EM does. It has the same drawback as coordinate descent: we must be able to carry out the coordinate descent steps (or, equivalently, the EM steps) exactly. Consequently, at least in its original presentation, EM is limited to relatively  simple models.

For more complex models, it is natural to ask: `Could we instead solve (S1,2) by applying a different optimization routine to $F$? What about perhaps the most basic of them all, gradient descent?'. To affirmatively answer  both questions, we need (a) a sensible notion of a `gradient' for functionals on $\Theta\times\cal{P}(\cal{X})$ and (b) practical methods implementing the gradients steps, at least approximately. At the time of \citet{Neal1998}'s publication, these obstacles had already begun to crumble: Otto and coworkers had introduced~\citep{Otto1998,Otto2001} a notion of gradients for functionals on $\cal{P}(\cal{X})$ 
(w.r.t.\ to the Wasserstein-2 geometry\footnote{Defining a gradient or `direction of maximum ascent' for a functional on $\cal{P}(\cal{X})$ requires quantifying the relative distances of neighbouring points in $\cal{P}(\cal{X})$ and, consequently, a metric. Otto et al.'s original work used the Wasserstein-$2$ metric on $\cal{P}(\cal{X})$, hence the `Wasserstein-2 geometry' jargon; cf.\ App.~\ref{app:crash} for more details.})  
and an associated calculus; and~\citet{Ermak1975,Parisi1981}  had proposed the  unadjusted Langevin algorithm (ULA, name coined in~\citet{Roberts1996}) that turned out to be a practical Monte Carlo approximation of  the corresponding gradient descent algorithm applied to a particular functional (although this connection has only been fleshed out much more recently in papers such as~\citet{Cheng2018}). In the ensuing two decades, these two lines of work have progressed greatly: Otto et al.'s ideas have been consolidated and imbued with rigour~\citep{Villani2009, Ambrosio2005}, analogues have been established for other geometries on $\cal{P}(\cal{X})$~\citep{Duncan2019,Garbuno-Inigo2020,Lu2019}, and more practical methods have been published~\citep{Liu2016,Garbuno-Inigo2020,Lu2019,Reich2021,Chen2018}.

Here, we capitalize on these developments and obtain scalable, easy-to-implement algorithms that tackle (S1,2) for broad classes of models (any for which $\Theta$~and~$\cal{X}$ are euclidean  and the density $p_\theta(x,y)$ is differentiable in $\theta$ and $x$). We consider three methods: an approximation to gradient descent (Sec.~\ref{sec:pgd}), one to Newton's method (App.~\ref{app:pqn}), and a further `marginal gradient' method (App.~\ref{app:pmgd}) applicable to models for which the (M) step is tractable but the (E) step is not --- a surprisingly common situation in practice.  
%In general, the estimates they produce possess a bias, which we discuss in Sec.~\ref{sec:bias}.
 We then study their performance in three examples (Sec.~\ref{sec:examples}). We conclude with a discussion of our methods, their limitations, and future research directions (Sec.~\ref{sec:discussion}).  
Code for our examples can be found at~\href{https://github.com/juankuntz/ParEM}{https://github.com/juankuntz/ParEM}.
\paragraph{Related literature and contributions.}Procedures reminiscent of those in Sec.~\ref{sec:pgd} and Apps.~\ref{app:pqn}, \ref{app:pmgd} are commonplace in variational inference, e.g.\ see \citet[Sec.~2]{Kingma2019}. 
Here, practitioners choose a tractable parametric family $\cal{Q}:=(q_\phi)_{\phi\in\Phi}\subseteq\cal{P}(\cal{X})$, parametrized by $\phi$s in some set $\Phi$, and solve
\begin{equation}\label{eq:vi}(\theta_*,\phi_*)=\argmin_{(\theta,\phi)\in\Theta\times \Phi}F(\theta,q_\phi)\end{equation}
using an appropriate optimization algorithm. If $\cal{Q}$ is sufficiently rich, then  $(\theta_*,q_{\phi_*})$ will be close to an optimum of $(\theta,q)\mapsto F(\theta,q)$ if $(\theta_*,\phi_*)$ is an optimum of $(\theta,\phi)\mapsto F(\theta,q_\phi)$. How rich $\cal{Q}$ needs to be is a complicated question and, in practice, $\cal{Q}$'s choice is usually dictated by computational considerations. Because the optimization of interest is that of $(\theta,q_\phi)$ over $\Theta\times \cal{Q}$  rather than that of $(\theta,\phi)$ over $\Theta\times\Phi$, it often proves beneficial to  adapt the optimization routine appropriately. For instance, one could use \emph{natural gradients}~\citep{Martens2020} defined not w.r.t.\ the  Euclidean geometry on $\Phi$ but instead w.r.t.\ a geometry that accounts for the effect that changes in $\phi$ have in $q_\phi$, with changes in $q_\phi$  measured by the KL divergence. In this paper, we circumvent these issues by working directly in $\cal{P}(\cal{X})$. We are also guided by  similar considerations when choosing $\theta$ updates (see Apps.~\ref{app:pqn},~\ref{app:pmgd} in particular): the object of interest here is the distribution $p_\theta(\cdot,y)$ indexed by $\theta$ rather than $\theta$ itself (but, $p_\theta(\cdot,y)$ is unnormalized, and it is no longer obvious that natural gradients are sensible).

Well-known algorithms are corner cases of ours. If the parameter space is trivial ($\Theta=\{\theta\}$) and we use a single particle ($N=1$ in what follows), the methods in Sec.~\ref{sec:pgd} and Apps.~\ref{app:pqn},~\ref{app:pmgd} reduce to ULA applied to the unnormalized density $p_\theta(\cdot,y)$. If, on the other hand, the latent space is trivial,  the algorithm in Sec.~\ref{sec:pgd} collapses to gradient descent applied to $\theta\mapsto p_\theta(y)$ and that in App.~\ref{app:pqn} to Newton's method.
Lastly, although we find the EB setting a natural one for introducing our methods, the EM algorithm can also be used to tackle many other problems, e.g.\ see \citet[Chap.~8]{McLachlan2007}, and, subject to the limitations discussed in Sec.~\ref{sec:discussion}, so can ours.

The contributions of this paper are as follows:
\begin{compactenum}
\item[(1)] We identify various gradient flows associated with $F$ (Sec.~\ref{sec:pgd},~Apps.~\ref{app:pqn},~\ref{app:pmgd}), review the pertinent theory (App.~\ref{app:crash}), and provide theoretical evidence for their convergence to $F$'s optima (Thrm.~\ref{thrm:expconv}, App.~\ref{app:conv}).
\item[(2)] Building on the insights afforded by (1), we derive three novel particle-based alternatives to EM (Sec.~\ref{sec:pgd},~Apps.~\ref{app:pqn},~\ref{app:pmgd}), study them theoretically (Sec.~\ref{sec:pgd},~Apps.~\ref{app:biasfinitepop},~\ref{app:biasdisc}), consider modifications that enhance their practical utility (Secs.~\ref{sec:pgd},~\ref{sec:gen}), and demonstrate the latter via several examples (Sec.~\ref{sec:examples}).
\item[(3)] We pave the way to other novel  methods for maximum likelihood estimation in latent variable models~(Sec.~\ref{sec:discussion}), be they, for example, optimization-inspired ones like those in Sec.~\ref{sec:pgd} and Apps.~\ref{app:pqn},~\ref{app:pmgd} or pure Monte Carlo approaches like those in App.~\ref{app:MH}.
\end{compactenum}
\paragraph{Our setting, notation, assumptions, rigour, and lack thereof.}In this  methodological paper, we favour intuition and clarity of presentation over mathematical rigour. We believe that all of the statements we make can be argued rigorously under the appropriate technical conditions, but we do not dwell on what these are. Except where strictly necessary, we avoid measure-theoretic notation, and we commit the usual notational abuse of conflating measures and kernels with their densities w.r.t.\ to the Lebesgue measure (this can be remedied by interpreting equations weakly and replacing density ratios with Radon-Nikodym derivatives). We also focus on Euclidean parameter and latent spaces ($\Theta=\r^{D_\theta}$ and $\cal{X}=\r^{D_x}$ for $D_\theta,D_x>0$), although our results and methods apply almost unchanged were these to be differentiable Riemannian manifolds. Throughout, $\bm{1}_{d}$ and $I_d$ respectively denote the $d$-dimensional vector of ones and identity matrix, $\cal{N}(\mu,\Sigma)$ the normal distribution with mean vector $\mu$ and covariance matrix $\Sigma$, and $\cal{N}(x;\mu,\Sigma)$ its density evaluated at $x$. We also tacitly assume that 
$p_\theta(x,y)>0$ for all $\theta,x,$ and $y$; and that  $(\theta,x)\mapsto p_\theta(x,y)$ is sufficiently regular that any gradients or Hessians we use are well-defined and any integral-derivative swaps and applications of integration-by-parts we do are justified. Furthermore, we make the following assumption, the violation of which indicates a poorly parametrized model or insufficiently informative data.
\begin{assumption}\label{ass:nodiv}The marginal likelihood's super-level sets $\{\theta\in\Theta: p_\theta(y)\geq l\}$, for any $l > 0$, are bounded.
\end{assumption}

\begin{figure*}
  \centering
  \includegraphics[width=1\linewidth]{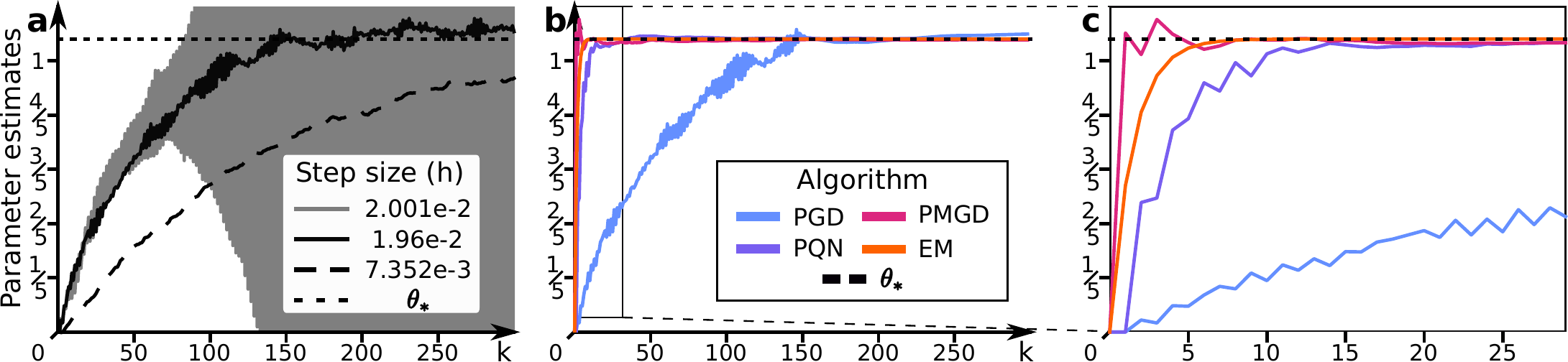}
  \caption{\textbf{Toy hierarchical model.} Parameter estimates for Ex.~\ref{ex:hier} with $D_x=100$ latent variables, $N=10$ particles, and both particles and estimates initialized at zero. \textbf{a} PGD estimates $\theta_k$ for three step sizes $h$. \textbf{b} PGD, PQN, PMGD, and EM parameter estimates. EM converges without averaging over time. For PGD, PQN,  and PMGD,  we use optimal step sizes (respectively, $h=1/51,2/3,1$, cf.\ App.~\ref{app:hierconvrates}) and start averaging once the estimates reach stationarity  (i.e.\ plot shows $\theta_k$ for $k<k_b$ and $\bar{\theta}_k$ for $k\geq k_b$ with $k_b=150, 15, 5$ for PGD, PQN, PMGD, respectively). \textbf{c} First $30$ steps in b.}\label{fig:hier}
\end{figure*}

\section{PARTICLE GRADIENT DESCENT}\label{sec:pgd}

The basic gradient descent algorithm for minimizing a differentiable function $f:\rn\to\r$,
\begin{equation}\label{eq:Euclideangradascent}x_{k+1}=x_k-h\nabla_x f(x_k),\end{equation}
is the Euler discretization with step size $h>0$ of $f$'s continuous-time \emph{gradient flow} $\dot{x}_t=-\nabla_x f(x_t)$, where $\nabla_x$ denotes the usual Euclidean gradient w.r.t. to $x$. To obtain an analogue of~\eqref{eq:Euclideangradascent} applicable to $F$ in~\eqref{eq:vfe}, we identify an analogue of $f$'s gradient flow and discretize it. Here, we require a sensible notion for $F$'s gradient. We use $\nabla F (\theta,q)= (\nabla_\theta F(\theta,q),\nabla_q F(\theta,q))$, where 
\begin{align}\nabla_\theta F(\theta,q)&=-\int \nabla_\theta \ell(\theta,x)q(x)dx,\label{eq:Fgradient1}\\
\nabla_q F(\theta,q)&=\nabla_x\cdot\left[ q\nabla_x \log\left(\frac{p_{\theta}(\cdot,y)}{q}\right)\right].\label{eq:Fgradient2}
\end{align}
This is the gradient obtained if we endow $\Theta$ with the Euclidean geometry and $\cal{P}(\cal{X})$ with the Wasserstein-2 one (see App.~\ref{app:grad}). It vanishes if and only if $\theta$ is a stationary point of $p_\theta(y)$ and $q$ is its corresponding posterior:
\begin{theorem}[$1^{st}$ order optimality condition]\label{thrm:statpoints}\(\nabla F(\theta,q)=0\) if and only if \(\nabla_\theta p_\theta(y)=0\) and \(q=p_\theta(\cdot|y)\).
\end{theorem}
\begin{proof}
Examining~(\ref{eq:Fgradient1},\ref{eq:Fgradient2}) we see that $\nabla_q F(\theta,q)=0$ if and only if $q\propto p_\theta(\cdot,y)$. Given that $q$ is a probability distribution, it follows that $\nabla_q F(\theta,q)=0$ if and only if $q=p_\theta(\cdot|y)$. The result then follows from
\begin{align}
\nabla_\theta p_\theta(y)&=\int \nabla_\theta p_\theta(x,y)dx=\int \nabla_\theta \ell(\theta,x)p_\theta(x,y)dx\nonumber\\
&=p_\theta(y)\int \nabla_\theta \ell(\theta,x)p_\theta(x|y)dx\label{eq:nsya8bnyfsanj7ufa}\\
&=-p_\theta(y)\nabla_\theta F(\theta,p_\theta(\cdot|y)).\nonumber
\end{align}\vspace{-35pt}\flushright{$\qquad$}\vspace{-15pt}
\end{proof}
The gradient flow corresponding to (\ref{eq:Fgradient1},\ref{eq:Fgradient2}) reads
\begin{align}
\dot{\theta}_t&= \int \nabla_\theta \ell(\theta_t,x)q_t(x)dx,\label{eq:IW2Gradflow1}\\
\dot{q}_t&=\nabla_x\cdot\left[ q_t\nabla_x \log\left(\frac{q_t}{p_{\theta_t}(\cdot,y)}\right)\right].\label{eq:IW2Gradflow2}
\end{align}
Given Assumpt.~\ref{ass:nodiv} and Thrm.~\ref{thrm:statpoints}, we expect that an extension of LaSalle's principle~\citep[Thrm.~1]{Carrillo2020} will show  that, as $t$ tends to infinity, 
$\theta_t$ approaches  a stationary point $\theta_*$ of $\theta\mapsto p_\theta(y)$ and $q_t$ the corresponding posterior $p_{\theta_*}(\cdot|y)$; see App.~\ref{app:conv} for more on this. Here, we settle for exponential convergence in the strongly log-concave case:

\begin{theorem}\label{thrm:expconv}Suppose there exists $\lambda>0$ and $C>0$ s.t.
\[\nabla^2 \ell(\theta,x) 	\preceq -\lambda I_{D_x+D_\theta},\quad\norm{\nabla_\theta \ell(\theta,x)}\leq C\]
for all $(\theta,x)$ in $\Theta\times\cal{X}$. The marginal likelihood $\theta\mapsto p_\theta(y)$ has a unique maximizer $\theta_*$ and there exists $C'>0$ s.t.\
\[\norm{\theta_t-\theta_*}\leq C'e^{-\lambda t}\enskip\text{and}\enskip \norm{q_t-p_{\theta_*}(\cdot|y)}_{L^1}\leq C'e^{-\lambda t}\]
for all $t\geq0$.
\end{theorem}
See App.~\ref{app:expoproof} for a proof. Eqs.~(\ref{eq:IW2Gradflow1},\ref{eq:IW2Gradflow2}) can rarely be solved analytically. To overcome this, note that (\ref{eq:IW2Gradflow1},\ref{eq:IW2Gradflow2}) is a mean-field Fokker-Planck equation satisfied by the law of a McKean-Vlasov SDE~\citep[Sec.~2.2.2]{Chaintron2022}:
\begin{align}\label{eq:IW2gradSDE1}
d\theta_t &= \left[\int \nabla_\theta \ell(\theta_t,x)q_t(x)dx\right]dt,\\
dX_t&=\nabla_x \ell(\theta_t,X_t)dt+\sqrt{2}dW_t,\label{eq:IW2gradSDE2}
\end{align}
where $q_t$ denotes $X_t$'s law and $(W_t)_{t\geq0}$ a  $D_x$-dimensional Brownian motion. To obtain an implementable algorithm, we now require a tractable approximation to the integral in~\eqref{eq:IW2gradSDE1} and a discretization of the time axis. For the former, we use a finite-sample approximation to $q_t$:  we generate $N\geq1$ particles $X^1_t,\dots,X^N_t$ with law $q_t$ by solving
\begin{equation}\label{eq:mlemckean3}dX_t^n=\nabla_x \ell(\theta_t,X_t^n)dt+\sqrt{2}dW_t^n\quad\forall n\in[N],\end{equation}
with $[N]:=\{1,\dots,N\}$ and $(W_t^1)_{t\geq0},\dots,(W_t^N)_{t\geq0}$ denoting $N$ independent Brownian motions, and exploit
\begin{align}&q_t\approx\frac{1}{N}\sum_{n=1}^N\delta_{X_t^n}\label{eq:finsampleapprox}\\
&\Rightarrow\quad  \int \nabla_\theta \ell(\theta_t,x)q_t(x)dx\approx \frac{1}{N}\sum_{n=1}^N\nabla_\theta \ell(\theta_t,X_t^n),\nonumber\end{align}
where $\delta_x$ denotes a Dirac delta at $x$. We then obtain the following approximation to (\ref{eq:IW2gradSDE1},\ref{eq:mlemckean3}):
\begin{align*}
d\theta_t &= \left[\frac{1}{N}\sum_{n=1}^N\nabla_\theta  \ell(\theta_t,X_t^n)\right]dt,\\
dX_t^n&=\nabla_x  \ell(\theta_t,X_t^n)dt+\sqrt{2}dW_t^n\enskip\forall n\in[N].
\end{align*}
To obtain an implementable algorithm (PGD in Alg.~\ref{alg:pgd}), we discretize the above using the Euler-Maruyama scheme. 

After running PGD for a large enough number of steps $K$, we approximate a stationary point  \(\theta_*\) of the marginal likelihood \(\theta\mapsto p_\theta(y)\) and its corresponding posterior \(p_{\theta_*}(\cdot|y)\) with either (a) the final parameter estimates $\theta_K$ and the final particle cloud's empirical distribution $q_K=N^{-1}\sum_{n=1}^N\delta_{X_K^n}$; or (b) with time-averaged versions thereof:
\begin{equation}\label{eq:estimators}\bar{\theta}_K:=\frac{1}{(K-k_b)}\sum_{k=k_b+1}^K\theta_k, \enskip \bar{q}_K:=\frac{1}{(K-k_b)}\sum_{k=k_b+1}^Kq_k,\end{equation}
where $k_b$ denotes the number of steps discarded as burn-in. 
\setlength{\textfloatsep}{10pt}
\begin{algorithm}[h]
\begin{algorithmic}[1]
\STATE{\textbf{Inputs:} step size $h$, step number $K$, particle number $N$, and initial particles $X^1_0,\dots,X_0^N$ \& parameters $\theta_0$.\hspace{-1pt}}
\FOR{$k=0,\dots, K-1$}
\STATE{Update the parameter estimates:\vspace{-2pt}
\begin{equation}\label{eq:IW2gradalg11}\theta_{k+1} = \theta_{k} + \frac{h}{N}\sum_{n=1}^N \nabla_\theta  \ell(\theta_k,X_k^n).\end{equation}\vspace{-8pt}}
\STATE{Update the particles: for all $n=1,\dots,N$,\begin{equation}
X_{k+1}^n=X_k^n+h\nabla_x \ell(\theta_k,X_k^n)+\sqrt{2h}W_k^n,\label{eq:IW2gradalg12}
\end{equation}
with $W_k^1,\dots,W_k^N$ denoting i.i.d.\ $\cal{N}(0,I_{D_x})$ R.V.s.}
\ENDFOR
\RETURN{$(\theta_k,q_k:=N^{-1}\sum_{n=1}^N\delta_{X_k^n})_{k=0}^K$.}
\end{algorithmic}
 \caption{Particle gradient descent (PGD).}
 \label{alg:pgd}
\end{algorithm}
\paragraph{PGD's behavior.}Given the analogy between \eqref{eq:Euclideangradascent} and (\ref{eq:IW2gradalg11},\ref{eq:IW2gradalg12}), we can formulate conjectures for PGD's behaviour based on that of (stochastic) gradient descent (note that~\eqref{eq:IW2gradalg11} involves noisy estimates of $F$'s $\theta$-gradient) and Thrm.~\ref{thrm:statpoints}:
\begin{compactenum}
\item[(C1)] If the step size $h$ is set too large, $\theta_k$ will be unstable.   
\item[(C2)] Otherwise, after a transient phase, $\theta_k$ will hover around a stationary point $\theta_*$ of $\theta\mapsto p_\theta(y)$, $q_k$ around the corresponding posterior $p_{\theta_*}(\cdot|y)$, and $(\bar{\theta}_k,\bar{q}_k)$ will converge to $(\theta_*,p_{\theta_*}(\cdot|y))$.
\item[(C3)] Small step sizes lead to long transient phases but low estimator variance in the stationary phase.  
\end{compactenum}
\begin{figure*}[t]
    \begin{center}
    \includegraphics[width=1\linewidth]{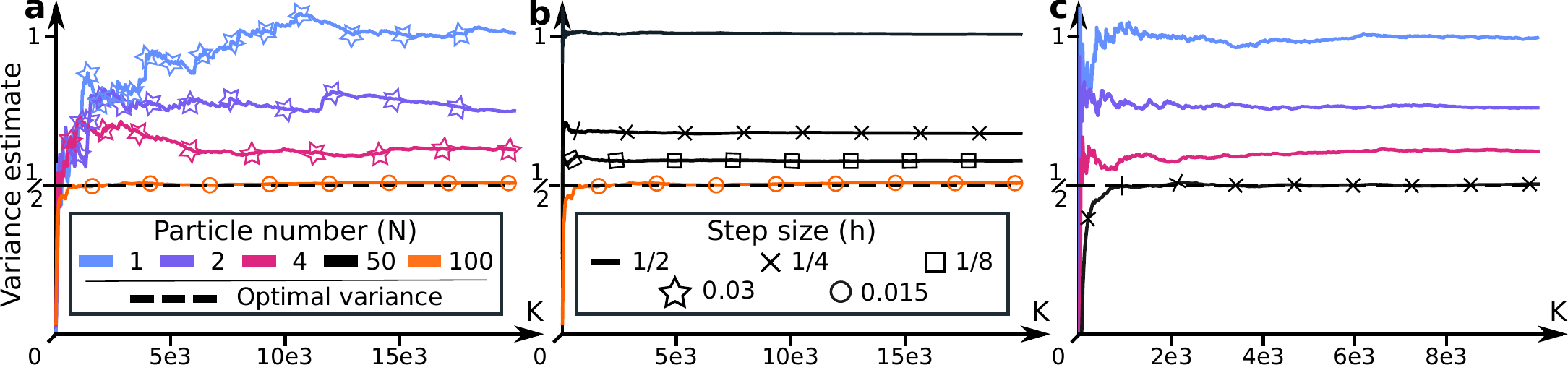}
    \end{center}
\caption{\textbf{Toy hierarchical model, bias.} PMGD estimates for the posterior variance in the $D_x=1$ case using the time-averaged posterior approximation $\bar{q}_K$  and no burn-in ($k_b=0$), as a function of step number $K$. \textbf{(a)} Even with a small step size $h$, using a single particle leads to a significant bias (blue). Growing the particle number $N$ reduces the bias (purple, magenta). The bias becomes negligible for  small $h$ and large $N$ (orange). \textbf{(b)} Even with a large $N$, a large $h$ leads to significant bias (solid). Decreasing $h$ reduces the bias (crosses, squares, circles). \textbf{(c)} Adding an accept-reject step removes (B1) regardless of the $h$ employed, and the remaining bias can be removed by choosing a sufficiently large $N$.}
\label{fig:hierbias}
\end{figure*}
Modulo the bias we discuss at the end of this section, (C1--3) are what we observe in our experiments:
\begin{example}\label{ex:hier}Consider a toy hierarchical model involving a single scalar unknown parameter $\theta$, $D_x$ i.i.d.\  mean-$\theta$ unit-variance Gaussian latent variables, and, for each of these, an independent  observed variable  with unit-variance Gaussian law centred at the latent variable:
$$p_\theta(x,y):=\prod_{d=1}^{D_x}\frac{1}{2\pi}\exp\left(-\frac{(x_d-\theta)^2}{2}-\frac{(y_d-x_d)^2}{2}\right).$$%\quad\forall \theta\in\r,\enskip x,y\in\r^{D_x}.

It is straightforward to verify that the marginal likelihood $\theta\mapsto p_\theta(y)$ has a unique maximum, $\theta_*=D_x^{-1}\sum_{d=1}^{D_x}y_d$, and obtain expressions for the corresponding posterior (see App.~\ref{app:hier}). Running PGD, we find that $\theta_k$ is unstable if the step size $h$ is too large (Fig.~\ref{fig:hier}a, grey). If $h$ is chosen well, $\theta_k$ approaches $\theta_*$ and hovers around it (Fig.~\ref{fig:hier}a, black solid) in such a way $\bar{\theta}_k$ converges to it (Fig.~\ref{fig:hier}b, blue). If $h$ is too small, the convergence is slow (Fig.~\ref{fig:hier}a, black dashed). 
\end{example}
\begin{table*}[b]  \caption{\textbf{Bayesian logistic regression.} Test errors achieved using time-averaged posterior approximation $\bar{q}_{400}$, with $N=1,10,100$, and corresponding computation times (averaged over $100$ replicates).
}
  \label{tab:blrmt}
  \centering
  \begin{tabular}{lllllll}
    \toprule
    &\multicolumn{2}{c}{$N=1$}&\multicolumn{2}{c}{$N=10$}&\multicolumn{2}{c}{$N=100$}\\
    \midrule
    & Error (\%)  & Time (s)& Error (\%)  & Time (s)& Error (\%)  & Time (s)\\
    \midrule
    PGD          & 3.58 $\pm$ 0.78 & 0.03 $\pm$ 0.01 & 3.55 $\pm$ 0.60 & 0.09 $\pm$ 0.01 & 3.46 $\pm$ 0.32 &  1.22 $\pm$ 0.34 \\
    PQN          & 3.54 $\pm$ 0.77 & 0.03 $\pm$ 0.00 & 3.49 $\pm$ 0.60 & 0.09 $\pm$ 0.00 & 3.47 $\pm$ 0.33 &  1.17  $\pm$ 0.26 \\
    PMGD         & 3.56 $\pm$ 0.69 & 0.03 $\pm$ 0.00 & 3.65 $\pm$ 0.68 & 0.09 $\pm$ 0.01 & 3.44 $\pm$ 0.33 &  1.15 $\pm$ 0.18 \\
    SOUL         & 3.53 $\pm$ 0.72 & 0.03 $\pm$ 0.00 & 3.60 $\pm$ 0.60 & 0.25 $\pm$ 0.01 & 3.43 $\pm$ 0.35 &  13.4 $\pm$ 0.23 \\
    \bottomrule
  \end{tabular}
\end{table*}

\paragraph*{Computational complexity and stochastic gradients.} PGD's complexity is $\cal{O}(KN[\text{eval. cost of }(\nabla_\theta \ell,\nabla_x \ell)])$. We can mitigate the $N$ factor by vectorizing computations across particles. In big data settings where evaluating $\ell$'s gradients is expensive, we replace them with stochastic estimates thereof similarly as in~\citep{Robbins1951,Welling2011}; cf.\  Sec.~\ref{sec:gen} for an example.

\paragraph*{Ill-conditioning, a heuristic, adaptive step sizes, PQN, and PMGD.}For many models, each component of $\nabla_\theta \ell$ is a sum of $D\gg1$ terms and, consequently, takes large values. On the other hand,  each component of $\nabla_x \ell$ typically involves far fewer terms (for instance, two in Ex.~\ref{ex:hier} while $D=D_x$). 
Hence, the parameter updates in (\ref{eq:IW2gradalg11}) are often much larger steps than the particle updates in~(\ref{eq:IW2gradalg12}). 
%Hence, (\ref{eq:IW2gradalg11},\ref{eq:IW2gradalg12}) often ends up taking much larger steps in the $\theta$-coordinates than in the $x$-coordinates. 
This ill-conditioning forces us to use small step sizes $h$ to keep $\theta_k$ stable. This results in `poor mixing' for the particles and overall slow convergence. In our experiments, we found that a simple heuristic mitigates the issue: in the parameter update, divide each component of $\nabla_\theta \ell$  by the corresponding number of terms ($D_x$ in Ex.~\ref{ex:hier}, see Sec.~\ref{sec:bnn} for another example). This amounts to pre-multiplying $\nabla_\theta \ell$ in \eqref{eq:IW2gradalg11} by positive definite matrix $\Lambda$ and does not alter  (\ref{eq:IW2gradalg11})'s fixed points. That is, we replace~(\ref{eq:IW2gradalg11}) with
\begin{align}\label{eq:IW2gradalg2}\theta_{k+1} = \theta_{k} + \frac{h}{N}\sum_{n=1}^N \Lambda\nabla_\theta \ell(\theta_k,X_k^n).\end{align}

For models with varying time-scales within the $\theta$-components, we found it helpful to adapt $\Lambda$ with $k$ similarly as in Adagrad or RMSProp (e.g., cf.~\citet{Ruder2016}); see Sec.~\ref{sec:gen} for an example. In cases where inverting $\ell$'s $\theta$-Hessian is not prohibitively expensive, we can alternatively mitigate the ill-conditioning using the PQN algorithm in App.~\ref{app:pqn}. (Indeed, for some simple models, (\ref{eq:IW2gradalg2},\ref{eq:IW2gradalg12}) coincides with PQN; in more complicated ones, it can be viewed as a crude approximation thereof.) Lastly, in cases where the E step is tractable, we can  circumvent this issue with the PMGD algorithm in App.~\ref{app:pmgd}.

\begin{figure*}[t]
  \centering
  \includegraphics[width=1\linewidth]{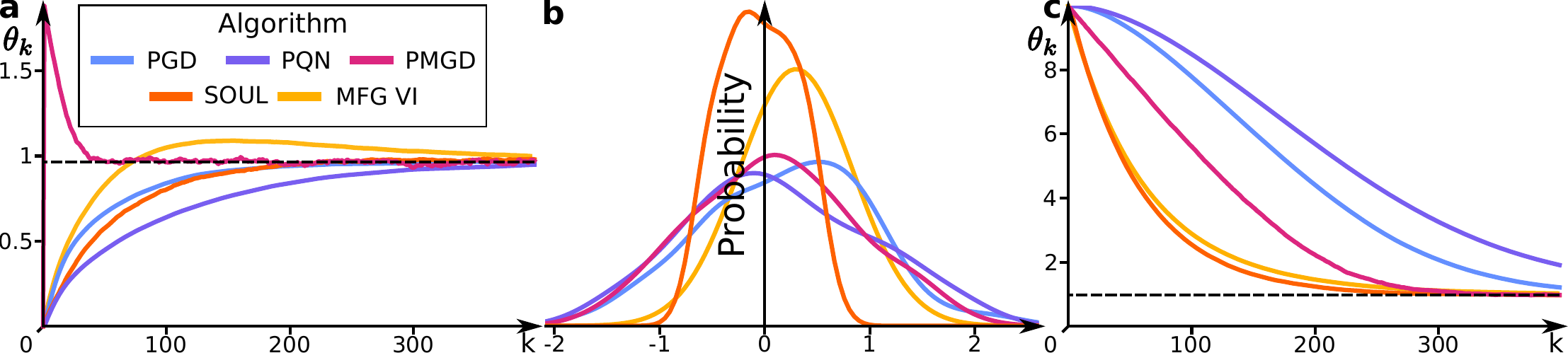}
  \caption{\textbf{Bayesian logistic regression.} \textbf{a} Parameter estimates $\theta_k$ 
  %as a function of $k$  
  initialized at zero (as in~\citet{Debortoli2021}). For PGD, PQN, PMGD, and SOUL, we use $N=100$ particles and a step size of $h=0.01$. \textbf{b} KDE of the second coordinate of the posterior approximation $100^{-1}\sum_{n=1}^{100}X_{400}^{n}$ for PGD, PQN, PMGD, and SOUL, and corresponding marginal of the VI approximation. \textbf{c} As in a, but with estimates initialized at ten.}\label{fig:blr}
\end{figure*}

\paragraph*{The bias.}For PGD, PQN, and PMGD, (C2) above is not quite true: the estimates produced by the algorithms are biased in the sense that $(\bar{\theta}_k,\bar{q}_k)$ does not converge exactly to $(\theta_*,p_{\theta_*}(\cdot|y))$, but rather to a point in its vicinity. This bias stems from two sources:
\begin{compactenum}
\item[(B1)] \textbf{$\bm{h>0}$.} Euler-Maruyama discretizations of the Langevin diffusion do not preserve stationary distributions: this can be seen by examining the mean-field limits of (\ref{eq:IW2gradalg11},\ref{eq:IW2gradalg12}), cf.\ App.~\ref{app:biasdisc}.
\item[(B2)] \textbf{$\bm{N<\infty}$.} Our use of finite particle populations: this is best understood by studying the continuum limits of (\ref{eq:IW2gradalg11},\ref{eq:IW2gradalg12}), cf.\ App.~\ref{app:biasfinitepop}.
\end{compactenum}
B1 can be mitigated by decreasing the step size and B2 by increasing the particle number:
\begin{example}\label{ex:hierbias}Consider again Ex.~\ref{ex:hier}. In this simple case, (B1,2) do not feature 
%materialize
 in the $\theta$-estimates  (Fig.~\ref{fig:hier})  because the model `is linear in $\theta$', cf.\ App.~\ref{app:hierfinitepopbias}. To observe (B1,2), we must examine the model's `non-linear aspects';  for instance, the posterior variance whose estimates are biased  (Fig.~\ref{fig:hierbias}).
\end{example}
Of course, increasing $N$ grows the algorithm's cost, and excessively lowering $h$ slows its convergence. 
%However, computations across particles are easily vectorized which substantially mitigates the cost growth (cf.\ Sec.~\ref{sec:examples}). 
It also seems  possible to eliminate (B1) altogether by adding population-wide accept-reject steps as described in App.~\ref{app:MH}, see Fig.~\ref{fig:hierbias}c. However, we do not dwell on this approach because a practical downside limits its scalability: the acceptance probability degenerates for large $D_x$ and $N$, forcing small choices of $h$ and slow convergence.

\begin{table*}[b]  \caption{\textbf{Bayesian neural network.} Test errors achieved using the final particle cloud $X_{500}^{1},\dots,X_{500}^N$, with $N=1,10,100$, and corresponding computation times (averaged over $10$ replicates).}
  \label{tab:bnn}
  \centering
  \begin{tabular}{lllllll}
    \toprule
    &\multicolumn{2}{c}{$N=1$}&\multicolumn{2}{c}{$N=10$}&\multicolumn{2}{c}{$N=100$}\\
    \midrule
    & Error (\%)  & Time (s)& Error (\%)  & Time (s)& Error (\%)  & Time (s)\\
    \midrule
    PGD          & 7.45 $\pm$ 2.03 & 4.10 $\pm$ 0.26 & 3.20 $\pm$ 1.12 & 10.4 $\pm$ 1.2 & 2.45 $\pm$ 0.99 &  76.6  $\pm$ 0.4 \\
    PQN          & 7.45 $\pm$ 1.60 & 4.12 $\pm$ 0.21 & 3.45 $\pm$ 1.04 & 10.0 $\pm$ 0.2 & 2.34 $\pm$ 0.81 &  74.0  $\pm$ 0.3 \\
    PMGD         & 7.24 $\pm$ 1.75 & 3.27 $\pm$ 0.13 & 3.75 $\pm$ 1.38 & 9.12 $\pm$ 0.2 & 2.45 $\pm$ 0.81 &  72.1  $\pm$ 0.5 \\
    SOUL         & 6.25 $\pm$ 1.54 & 5.02 $\pm$ 0.20 & 7.25 $\pm$ 1.38 & 36.5 $\pm$ 0.1 & 6.85 $\pm$ 1.42 &  364.0 $\pm$ 5.3 \\
    \bottomrule
  \end{tabular}
\end{table*}

\section{NUMERICAL EXPERIMENTS}\label{sec:examples}We examine the performance of our methods by applying them to train a Bayesian logistic regression model for breast cancer prediction (Sec.~\ref{sec:blr}), a Bayesian neural network for MNIST classification (Sec.~\ref{sec:bnn}), and a generator network for image reconstruction and synthesis (Sec.~\ref{sec:gen}).

\subsection{Bayesian logistic regression}\label{sec:blr}We consider the set-up described in~\citet[Sec.~4.1]{Debortoli2021} and employ the same dataset with $683$ datapoints, cf.\ App.~\ref{app:blr} for details. The latent variables are the $9$ regression weights. We  assign  an isotropic Gaussian prior $\cal{N}(\theta \bm{1}_{D_x},5 I_{D_x})$ to the weights,  and  we estimate the marginal likelihood's unique maximizer $\theta_*$ (cf.\ Prop.~\ref{prop:blr} in App.~\ref{app:blr} for the uniqueness).

\begin{figure*}[t]
%\vspace{-10pt}
  \centering
  \includegraphics[width=1\linewidth]{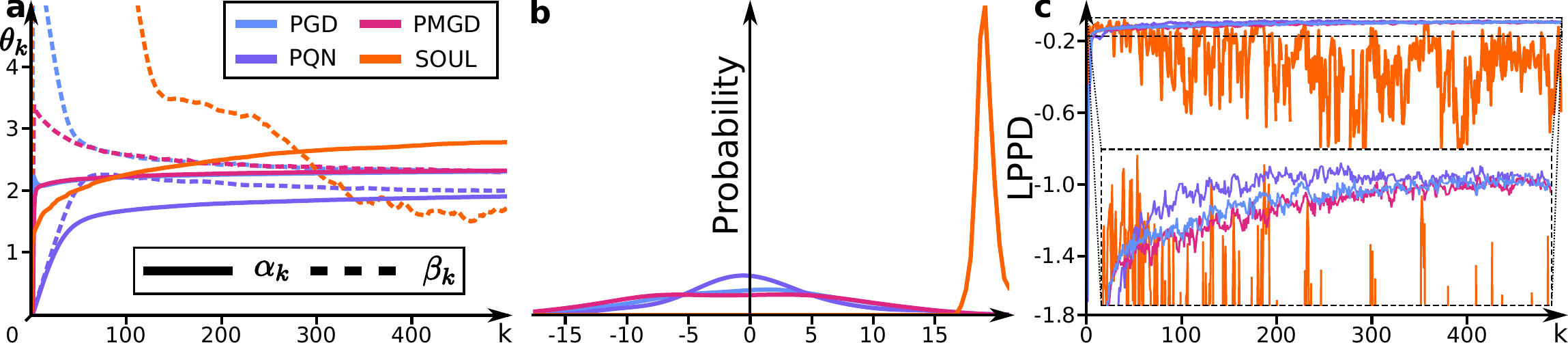}
  \caption{\textbf{Bayesian neural network.} \textbf{a} Parameter estimates as a function of $k$ with $N=100$ particles and step size of $h=0.1$. \textbf{b} KDE of a randomly-chosen coordinate of the posterior approximation $100^{-1}\sum_{n=1}^{100}X_{500}^{n}$.  \textbf{c} Log pointwise predictive density as a function of $k$. \textbf{c, inset} c zoomed-in to y-axis range $[-0.16,-0.08]$.}\label{fig:bnn}
\end{figure*}

We benchmark our algorithms against the Stochastic Optimization via Unadjusted Langevin (SOUL) algorithm\footnote{In \cite{Debortoli2021}, the authors allow for step sizes and particle numbers that change with $k$. To simplify the comparison and place all methods on equal footing, we fix a single step size $h$ and particle number $N$.}, recently proposed~\citep{Debortoli2021} to overcome the limited scalability of traditional MCMC EM variants. Because it is a coordinate-wise cousin of PGD~(Alg.~\ref{alg:pgd}), it allows for straightforward meaningful comparisons with our methods. SOUL approximates the (M) step by updating the parameter estimates using a single (stochastic) gradient step as we do in~(\ref{eq:IW2gradalg11}). For the (E) step, it instead runs a single ULA chain for $N$ steps, `warm-started' using the previous chain's final state ($X_{k}^1:=X_{k-1}^N$): for all $n\leq N-1$,
\begin{equation}X_{k}^{n+1}=X_{k}^n+h\nabla_x \ell(\theta_k,X_{k}^n)+\sqrt{2h}W_{k}^n;\label{eq:SOUL}\end{equation}
and then approximates $p_{\theta_{k}}(\cdot|y)$ using the chain's  empirical distribution $q_{k}:=N^{-1}\sum_{n=1}^N\delta_{X_{k}^n}$. 

The parameter estimates produced by PGD, PQN~(App.~\ref{app:pqn}), PMGD~(App.~\ref{app:pmgd}), and SOUL all converge to the same limit (Fig.~\ref{fig:blr}a). SOUL is known~\citep{Debortoli2021} to return accurate estimates of $\theta_*$ for this example, so we presume that this limit approximately equals $\theta_*$. All algorithms produce posterior approximations with similar predictive power  regardless of the particle number $N$ (Tab.~\ref{tab:blrmt}; see also  Tab.~\ref{tab:blrapp} in App.~\ref{app:blr}): the task is simple and it is straightforward to achieve good performance. In particular, the posteriors are unimodal and peaked (e.g.\ see \citet[Fig.~2]{Debortoli2021}) and approximated well using a single particle in the vicinity of their modes.  The variance of the stationary PGD, PQN, and PMGD estimates seems to decay linearly with $N$ (Tab.~\ref{tab:blrapp}); which is unsurprising given that these algorithms are Monte Carlo methods.

We found three noteworthy differences between SOUL and our methods. First, the computations in~(\ref{eq:IW2gradalg12},\ref{eq:W2gradalg}) are easily vectorized across particles while those in~\eqref{eq:SOUL} must be done in serial. This results in our algorithms running faster, with the gap in computation times growing with $N$ (Tab.~\ref{tab:blrmt}). Second, SOUL tends to produce narrower approximations than our methods (Fig.~\ref{fig:blr}b). This stems from the strong sequential correlations of the particles $X_{k}^{1},\dots,X_{k}^N$ in~\eqref{eq:SOUL}. In contrast, the particles in our algorithms are only weakly correlated through the (mean-field) parameter estimates.  
Last, if the parameter estimates are initialized far from $\theta_*$ and the particles are initialized far from $p_{\theta_*}(\cdot|y)$'s mode, then SOUL exhibits a shorter transient than our algorithms (Fig.~\ref{fig:blr}c). This is because SOUL updates a single particle $N$ times per parameter update and quickly locates the current posteriors's mode, while our algorithms are stuck slowly moving $N$ particles, one update per parameter update, to the posteriors's mode.  However, in this example, 
%we found little benefit in using any method with more than one particle, at least until the transient phase was over. For good predictive performance, we found it most expedient to use any method with a single particle; for low variance estimates of $\theta_*$, we found it most efficient to use any algorithm with a single during the transient phase and switching to PGD or PQN with several particles in the stationary phase (Tab.~\ref{tab:blrapp}).
%
 we found little benefit in using multiple particles until the transient phase is over. Low variance estimates of $\theta_*$ are most efficiently obtained using a single particle in the transient phase and switching to PGD or PQN with multiple particles in the stationary phase (App.~\ref{app:blr}); if predictive performance is the sole concern, then any method with a single particle performed well.% in this example.

As an additional baseline, we run mean-field Gaussian variational inference (MFG VI); c.f.\ App.~\ref{app:blr} for details. MFG VI's parameter estimates converge to the same limit as those of the other algorithms (Fig.~\ref{fig:blr}a,c). The algorithm achieves similar test errors ($3.65\% \pm 0.01\%$) as PGD, PQN, and PMGD but produces narrower posterior approximations  (Fig.~\ref{fig:blr}b).

\subsection{Bayesian neural network}\label{sec:bnn}To test our algorithms on an example with more complex posteriors, we turn to Bayesian neural networks whose posteriors are notoriously multimodal. In particular, we consider the setting of~\citet[Sec.~6.5]{Yao2022} and apply a simple two-layer neural network to classify MNIST images, cf.~App.~\ref{app:bnn} for details.  Similarly to \citet{Yao2022}, we avoid big data issues by  subsampling $1000$ data points with labels $4,9$.  The input layer has $40$ nodes and $784$ inputs, and the output layer has $2$ nodes. The latent variables are the weights, $w\in\r^{40\times 784}$, of the input layer and those, $v\in\r^{2\times 40}$, of the output layer. As in \citet{Yao2022}, we assign zero-mean isotropic Gaussian priors to the weights with respective variances $e^{2\alpha}$ and $e^{2\beta}$. However, rather than assigning hyperpriors to $\alpha$ and $\beta$, we instead learn them from the data (i.e.\ $\theta:=(\alpha, \beta)$). To avoid memory issues, we only store the current particle cloud and use its empirical distribution to approximate the posteriors (rather than the time-averaged version in~\eqref{eq:estimators}). 
\\\\
PGD~(Sec.~\ref{sec:pgd}), PQN~(App.~\ref{app:pqn}), PMGD~(App.~\ref{app:pmgd}), and SOUL~(Sec.~\ref{sec:blr}) all exhibit a short transient in their parameter estimates and predictive performances, after which the estimates  appear to converge to different local maxima of the marginal likelihood (Fig.~\ref{fig:bnn}a) and the performances of PGD, PQN, and PMGD show a slow, moderate increase (Fig.~\ref{fig:bnn}c). SOUL achieves noticeably worse predictive performance (Fig~\ref{fig:bnn}c) and shows little improvement with larger particle numbers $N$ (Tab.~\ref{tab:bnn}). We believe this is due to the peaked SOUL posterior approximations (Fig.~\ref{fig:bnn}b) caused by the strong correlations among the SOUL particles. Just as in Sec.~\ref{sec:blr}, PGD, PQN, and PMGD all run significantly faster than SOUL due to the former three's vectorization, and the gap also widens with $N$ (Tab.~\ref{tab:bnn}).

\subsection{Generator network}\label{sec:gen}

To test our methods on a more challenging example, we turn to generator networks~\citep{Goodfellow2020,Han2017,Nijkamp2020} applied to two image datasets: MNIST and CelebA (both $32\times32$). These are generative models used for a variety of tasks, including image reconstruction and synthesis. They assume that each image $y$ in the dataset is generated by independently sampling a latent variable $x$ from a Gaussian prior, 
%on $\r^{64}$,
mapping $x$ to the image space through a convolutional neural network $f_\theta$ parametrized by $\theta$, and adding Gaussian noise $\epsilon$: $y=f_\theta(x)+\epsilon$. We use $10,000$ training images for MNIST, $40,000$ for CelebA, and a network  with $13$ layers and $D_\theta\approx 350,000$ parameters similar to those in~\citet{Nijkamp2020}. In total, the model involves  $D_x=640,000$ latent variables for  MNIST and $D_x=2,560,000$ for CelebA ($64$ per training image). We train it as in~\citet{Han2017,Nijkamp2020} by searching for parameters $\theta$ that maximize the likelihood of the training set. To do so, we use PGD, slightly tweaked to cope with the problem's high dimensionality and exploding/vanishing gradient issues caused by $f_\theta$'s depth. In particular, we replace the gradients in~(\ref{eq:IW2gradalg11},\ref{eq:IW2gradalg12}) with subsampled versions thereof and adapt the step sizes in~\eqref{eq:IW2gradalg11} similarly as in RMSProp~\citep{Hinton2012}. To benchmark PGD's performance, we also train the model as a variational autoencoder (VAE; i.e.\ using variational approximations to the posteriors rather than particle-based ones,~\citet{Kingma2014}), with alternating back propagation (ABP;~\citet{Han2017}), and with short-run MCMC (SR;~\citet{Nijkamp2020}). The latter two are variants of  (\ref{eq:IW2gradalg11},~\ref{eq:SOUL}) specifically proposed  for training generator networks. They both  approximate the posterior $p_{\theta_{k}}(\cdot|y)$ using only~\eqref{eq:SOUL}'s final state (i.e.\ with $q_{k}:=\delta_{X_k^N}$) and, in the case of SR, the chains are not `persistent' (i.e.\ rather than initializing $X_k^1$ at $X_{k-1}^N$ it is sampled from the prior). For ABP and SR, we also subsample gradients and adapt the step size just as with PGD. See App.~\ref{app:gen} for the full details.

We evaluate the learned generators $f_\theta$ by applying them to inpaint occluded test images and synthesize fake images. In the inpainting task,  the generator learned with PGD outperformed the others for MNIST (Tab.~\ref{tab:gen}, see also Fig.~\ref{fig:inpaint} in App.~\ref{app:gen}). For CelebA, both SR and PGD did well. In the synthesis task,   all methods did poorly when we followed the usual approach of  generating images by drawing latent variables from the prior and mapping them through $f_\theta$ (cf.~Fig.~\ref{fig:syn_prior} in App.~\ref{app:gen}). For the reasons explained in App.~\ref{app:gen}, we instead opted to draw latent variables from a Gaussian approximation to the aggregate posterior~\citep{Aneja2021} which significantly improved the fidelity of the images generated (Fig.~\ref{fig:syn_1} in App.~\ref{app:gen}). With this approach, PGD outperformed the other algorithms, although all four methods performed comparably for CelebA (Tab.~\ref{tab:gen}). Using more refined approximations to the aggregate posterior led to further improvements (Fig.~\ref{fig:syn_500} in App.~\ref{app:gen}). 

\begin{table}[h]  \caption{\textbf{Generator network.} \textbf{(Inpainting)} Mean squared error averaged for $1000$ test images. \textbf{(Synthesis)} Fr\'{e}chet Inception distance~\citep{Heusel2017} computed using $200$ test images. (All results averaged over $3$ replicates).
%The table shows the values are the average and the standard deviation of the scores estimated over three independent trials.
}\label{tab:gen}
  \centering
  \begin{tabular}{lllll}
    \toprule
    &\multicolumn{2}{c}{Inpainting ($10^{-2}$)}&\multicolumn{2}{c}{Synthesis}\\
    \midrule
    & MNIST   & CelebA & MNIST  & CelebA \\
    \midrule
    PGD          & $\pmb{4.1} \pm 0.3$ & $\pmb{2.0} \pm 0.0$ & $\pmb{71} \pm 2.4$ & $\pmb{100} \pm 2.7$  \\
    ABP          & $5.2 \pm 0.1$ & $2.9 \pm 0.1$ & $92 \pm 3.0$ & $106 \pm 1.3$  \\
    SR           & $7.4 \pm 0.3$ & $\pmb{2.0} \pm 0.0$ & $95 \pm 1.5$ & $102 \pm 2.3$  \\
    VAE          & $10 \pm 0.7$ & $3.3 \pm 0.1$ & $148 \pm 9.3$ & $104 \pm 0.4$  \\
    \bottomrule
  \end{tabular}
\end{table}

\section{DISCUSSION}\label{sec:discussion}
%By viewing the problem that EM tackles as a joint problem over $\theta$ and $q$ rather than an alternating-coordinate-wise one, we open the door to numerous new algorithms for solving the problem (be they, for instance, optimization-inspired ones, along the lines of those in Sec.~\ref{sec:pgd} and Apps.~\ref{app:pqn},~\ref{app:pmgd}, or purely Monte-Carlo ones of the type in App.~\ref{app:MH}). 
In contrast to EM and its many variants, we view maximum likelihood estimation of latent variable models as a joint problem over $\theta$ and $q$ rather than an alternating-coordinate-wise one, and thereby open the door to numerous new algorithms for solving the problem (be they, for instance, optimization-inspired ones, along the lines of those in Sec.~\ref{sec:pgd} and Apps.~\ref{app:pqn},~\ref{app:pmgd}, or purely Monte-Carlo ones of the type in App.~\ref{app:MH}). This perspective, of course, is not entirely unprecedented: even in p.6 of our starting point~\citep{Neal1998}, the authors mention in passing the possibility of optimizing $F$ `simultaneously' over $\theta$ and $q$, and this idea has been taken up enthusiastically in the VI literature, e.g.~\citet[Sec.~2]{Kingma2019}. However, outside of variational inference, we have struggled to locate papers following up on the idea.

%alternatives to the EM algorithm
We propose three particle-based algorithms for maximum likelihood training of latent variable models: PGD (Sec.~\ref{sec:pgd}), PQN (App.~\ref{app:pqn}), and PMGD (App.~\ref{app:pmgd}).  Practically, we find these algorithms appealing because they are simple to implement and tune, apply to broad classes of models (i.e.\ those on Euclidean spaces with differentiable densities), and, above all, are scalable. For instance, as discussed in Sec.~\ref{sec:pgd}, PGD's total cost is $\cal{O}(KN[\text{eval. cost of }(\nabla_\theta \ell,\nabla_x \ell)])$ which, for many models in the literature, is linear in the dimensions of the data, latent variables, and parameters.  For big data scenarios where this still proves prohibitive, we advise replacing $\ell$'s derivatives with unbiased estimates thereof as we did for the generator network (Sec.~\ref{sec:gen}; see also~\citet{Robbins1951,Welling2011,Nemeth2021}). Lastly, much like in~\citet{Debortoli2021}, we circumvent the degeneracy with latent variable dimension that plagues common MCMC methods (e.g.\ see~\citet{Beskos2013,Vogrinc2022,Kuntz2018,Kuntz2019} and references therein) by avoiding accept-reject steps and employing ULA kernels (known to have favourable properties; cf.~\cite{Dalalyan2017,Durmus2017,Durmus2019}).

Theoretically, we find PGD, PQN, and PMGD attractive because they re-use the previously computed posterior approximation at each update step, and `warm-starts' along these lines are known to be beneficial for methods reliant on the ULA kernel~\citep{Dalalyan2017,Durmus2017,Durmus2019}. This stands in contrast with previous Monte Carlo EM alternatives~(cf.\ Sec.~\ref{sec:intro}) which, at best, initialize the chain for the parameter current update at the final state of the preceding update's chain. This results in our methods achieving better performance for models with complex multimodal posteriors (Sec.~\ref{sec:bnn}). It proved a disadvantage for models with simple peaked unimodal posteriors where piecemeal evolving an entire particle cloud leads to long transients for poor initializations (Sec.~\ref{sec:blr}).  However, this issue was easily mitigated by warm-starting our algorithms using a preliminary single-particle run (App.~\ref{app:blr}).

We see several interesting lines of future work including (a) the theoretical analysis of the algorithms proposed in this paper, (b) the study of variants thereof, and (c) the investigation of other particle-based methods obtained by viewing the EM problem `jointly over $\theta$ and $q$' rather than in a coordinate-wise manner. For (a), we believe that \cite{Dalalyan2017,Durmus2017,Durmus2019,Debortoli2021} might be good jumping-off points. Aside from the variants discussed in Sec.~2, for (b), we have in mind adapting step sizes and particle numbers as the algorithms run: it seems natural to use cruder posterior approximations and larger step sizes early on in $F$'s optimization, cf.\ \cite{Wei1990,Gu1998,Delyon1999,Younes1999,Kuhn2004,Cai2010,Debortoli2021,Robbins1951} for similar ideas. In particular, by decreasing the step size $h$ and increasing the particle number $N$ with the step number $k$, it is likely possible to eliminate the asymptotic bias (Sec.~\ref{sec:pgd}). For (c), this might amount to switching the geometry on $\Theta\times \cal{P}(\cal{X})$ w.r.t.\ which we define gradients and following a discretization procedure analogous to that in Sec.~\ref{sec:pgd}. For instance, using a Stein geometry leads to a generalization of SVGD~\citep{Liu2016} which makes more extensive use of the particle cloud at the price of a higher computational cost. Alternatively, one could search for analogues of other well-known optimization algorithms applied to $F$ aside from gradient descent (e.g.\ ones for Nesterov acceleration and mirror descent along the lines of~\citet{Ma2019,Cheng2018a,Taghvaei2019,Wang2022} and \citet{Ahn2021,Jiang2021,Hsieh2018,Chewi2020,Zhang2020}, resp.) or a Metropolis-Hastings method of the type in App.~\ref{app:MH}.

\textbf{Limitations.} Our algorithms, like EM and most alternatives thereto (but not all, e.g.\ \cite{Doucet02marginal,Johansen08particle}), only return stationary points of the marginal likelihood and not necessarily global optima. Moreover, at least as presented here, our algorithms are limited to Euclidean parameter and latent spaces and models with differentiable densities. This said, they apply almost unchanged were the spaces to be Riemannian manifolds (e.g.\ see \cite{Boumal2022}). For discrete spaces, it might be possible to adapt the techniques in~\cite{Zhang2022,Grathwohl2021,Sun2022}. Lastly, some common non-differentiabilities can be dealt with by incorporating proximal operators into our algorithms along the lines of~\cite{Parikh2014,Pereyra2016,Durmus2019a,Durmus2018,Bernton2018,Vidal2020,Debortoli2020,Salim2020,Salim2020a}.

\subsubsection*{Acknowledgements}
We thank Valentin De Bortoli, Arnaud Doucet, and Jordan Ang for insightful discussions. We also thank the anonymous referees for their helpful comments. JK and AMJ acknowledge support from the Engineering and Physical Sciences Research Council (EPSRC; grant \# EP/T004134/1) and the Lloyd's Register Foundation Programme on Data-Centric Engineering at the Alan Turing Institute. AMJ acknowledges further support from the EPSRC (grant \#  EP/R034710/1). JNL is supported by the Feuer International Scholarship in Artificial Intelligence.

%\subsubsection*{References}

\bibliographystyle{plainnat}
\bibliography{../MMLE}

\onecolumn\appendix
\aistatstitle{Particle algorithms for maximum likelihood training of latent variable models: \\
Supplementary Materials}

% Add SM TOC:
\renewcommand\ptctitle{}
\setcounter{parttocdepth}{1}
\part{} % Start the appendix part
\parttoc 

\section{AN INFORMAL CRASH COURSE IN CALCULUS ON $\Theta\times \cal{P}(\cal{X})$}\label{app:crash}
This appendix assumes that the reader is familiar with rudimentary Riemannian geometry not exceeding the level of \citet[Chap.~3]{Boumal2022}.

Otto et al.'s observation~\citep{Otto1998,Otto2001} was that, even though $\cal{P}(\cal{X})$ is not technically a Riemannian manifold, we can often treat it as one and apply the rules we have for calculus on Riemannian manifolds almost unchanged. While rigorously establishing these facts is an involved matter~\citep{Ambrosio2005,Villani2009}, the basic ideas are very accessible. Here we review these ideas, but in the slightly generalized setting of $\cal{M}:=\Theta\times \cal{P}(\cal{X})$. To treat $\cal{M}$ as a Riemannian manifold we require three things:
\begin{itemize}
\item for each $(\theta,q)$ in $\cal{M}$, a tangent space $\cal{T}_{(\theta,q)}\cal{M}$: a linear space containing the directions we can move in from $(\theta,q)$;
\item for each $(\theta,q)$ in $\cal{M}$, a cotangent space $\cal{T}_{(\theta,q)}^*\cal{M}$ dual to $\cal{T}_{(\theta,q)}\cal{M}$ with a duality pairing 
\[\iprod{\cdot}{\cdot}_{(\theta,q)}:\cal{T}_{(\theta,q)}\cal{M}\times \cal{T}_{(\theta,q)}^*\cal{M}\to\r;\]
\item and a Riemannian metric $g=(g_{(\theta,q)})_{(\theta,q)\in\cal{M}}$, with $g_{(\theta,q)}$ denoting an inner product on $\cal{T}_{(\theta,q)}\cal{M}$ for each $(\theta,q)$ in $\cal{M}$.
\end{itemize}
Once we have chosen the above, defining a sensible notion for the gradient of a functional on $\cal{M}$ will be a simple matter. 

\noindent\textbf{An abuse of notation.} The tangent spaces $(\cal{T}_{(\theta,q)}\cal{M})_{(\theta,q)\in\cal{M}}$ that we use will be copies of a single space $\cal{T}\cal{M}$ (and, in particular, independent of $(\theta,q)$). Hence, we drop the $(\theta,q)$ subscripts to simplify the notation. Similarly for the cotangent spaces and duality pairings.
\subsection{Tangent and cotangent  spaces}\label{app:tancon}
$\cal{M}$ is defined as the product of $\Theta$ and $\cal{P}(\cal{X})$, so we find sensible tangent spaces, $\cal{T}\Theta$ and $\cal{T}\cal{P}(\cal{X})$, for these two and set that for $\cal{M}$ to be their product:
\[\cal{T}\cal{M}=\cal{T}\Theta\times\cal{T}\cal{P}(\cal{X}).\]
The cotangent spaces then obey an analogous relationship,
\[\cal{T}^*\cal{M}=\cal{T}^*\Theta\times\cal{T}^*\cal{P}(\cal{X}),\]
and we can express the duality pairing for $(\cal{T}\cal{M},\cal{T}^*\cal{M})$ in terms of those for $(\cal{T}\Theta,\cal{T}^*\Theta)$ and $(\cal{T}\cal{P}(\cal{X}),\cal{T}^*\cal{P}(\cal{X}))$:
\[
\iprod{(\tau,m)}{(v,f)}=\iprod{\tau}{v}+\iprod{m}{f}  \quad\forall (\tau,m)\in\cal{T}\cal{M},\enskip (v,f)\in\cal{T}^*\cal{M}.
\]
\paragraph{Tangent and cotangent spaces for $\Theta$.}Throughout the paper we focus on Euclidean parameter spaces ($\Theta=\r^{D_\theta}$), in which case the tangent spaces are just copies of the parameter space: $\cal{T}\Theta=\r^{D_\theta}$. The cotangent spaces are also copies of $\r^{D_\theta}$ and the duality pairing is the Euclidean inner product:
\[\iprod{\tau}{v}:=\sum_{i=1}^{D_\theta}\tau_iv_i\quad\forall \tau\in\cal{T}\Theta,\enskip v\in\cal{T}^*\Theta.\]%
The above said, modulo the re-insertion of $\theta$ subscripts, the ensuing discussion would apply unchanged were $\Theta$ to be any sufficiently-differentiable finite-dimensional Riemannian manifold.
\\\\
\noindent\textbf{Tangent and cotangent spaces for $\cal{P}(\cal{X})$.}  To keep the exposition simple, we restrict $\cal{P}(\cal{X})$ to the set of probability measures with strictly positive densities w.r.t.\ to the Lebesgue measure $dx$ and identify a measure with its density. (Circumventing this restriction and giving a fully rigorous treatment of our results requires employing the techniques of \cite{Ambrosio2005}.) With this restriction, the tangent spaces are simple and do not depend on $q$:
\[\cal{T}\cal{P}(\cal{X}):=\left\{\text{functions }m:\cal{X}\to\r\text{ satisfying }\int m(x)dx=0\right\}.\]
The cotangent spaces can be identified with the space of equivalence classes of functions that differ by an additive constant,
\[\cal{T}^*\cal{P}(\cal{X}):=\{f:\rn\to\r\}/\r;\]
and the duality pairing is given by
\[\iprod{m}{f}:=\int f(x)m(x)dx\quad\forall m\in \cal{T}\cal{P}(\cal{X}),\enskip f\in \cal{T}^*\cal{P}(\cal{X}).\]
Note that, in the above and throughout, we commit the usual notational abuse using $f$ to denote both a function and the equivalence class to which it belongs. We also tacitly assume that the measurability and integrability conditions required for our integrals to make sense are satisfied.
\subsection{Riemannian metrics}\label{app:metrics}
We define each metric $g=(g_{(\theta,q)})_{(\theta,q)\in\cal{M}}$ in terms of a \emph{tensor} $G=(G_{(\theta,q)})_{(\theta,q)\in\cal{M}}$  and the duality pairing:
\[g_{(\theta,q)}((\tau,m),(\tau',m')):=\iprod{(\tau,m)}{G_{(\theta,q)}(\tau',m')}\quad\forall (\theta,q)\in\cal{M}.\]
By a tensor $G$ we mean a  collection indexed by $(\theta,q)$ in $\cal{M}$ of invertible, self-adjoint, positive-definite, linear maps from $\cal{T}\cal{M}$ to $\cal{T}^*\cal{M}$.  Most of the tensors $G_{(\theta,q)}$ we will  consider are `block-diagonal':
\[\iprod{(\tau,m)}{G_{(\theta,q)}(\tau',m')}=\iprod{\tau}{\rm{G}_{(\theta,q)}\tau'}+\iprod{m}{\sf{G}_{(\theta,q)}m'}=:\rm{g}_{(\theta,q)}(\tau,\tau')+\sf{g}_{(\theta,q)}(m,m'),\]
where $\rm{G}_{(\theta,q)}$ and $\sf{G}_{(\theta,q)}$ respectively denote tensors on $\Theta$ and $\cal{P}(\cal{X})$. In this case, we write diag$(\rm{G}_{(\theta,q)},\sf{G}_{(\theta,q)})$ for $G_{(\theta,q)}$ If any of the above do not depend on $\theta$, we omit it from the subscript, and similarly for $q$.

Although there are many options that one could consider for the  $\sf{G}_{(\theta,q)}$ block (e.g.\ see~\cite{Duncan2019,Garbuno-Inigo2020,Lu2019}), we focus on two, the first for practical reasons and the second for theoretical ones:
\begin{description}

\item[Wasserstein-2.] The tensor $\sf{G}_q^W$ is defined by its inverse  
\begin{equation}\label{eq:W2inv}(\sf{G}_q^W)^{-1}f:=-\nabla_x\cdot(q\nabla_x f).\end{equation}
Using integration-by-parts, we find that
\[\sf{g}_q^W(m,m')=\int \iprod{\nabla_x f(x)}{\nabla_x f'(x)}q(x)dx\quad\forall q\in\cal{P}(\cal{X}),\]
where $f,f'$ are the unique (up to an additive constant) solutions to  $m=(\sf{G}_q^W)^{-1}f$ and $m'=(\sf{G}_q^W)^{-1}f'$ and $\iprod{\cdot}{\cdot}$ denotes the Euclidean inner product on $\r^{D_x}$. The tensor's name stems from the fact that the distance metric induced by $\sf{g}_q^W$ on $\cal{P}(\cal{X})$ coincides with the Wasserstein-2 distance from optimal transport, e.g.\ see~\citet[p.\ 168]{Ambrosio2005}.
\item[Fisher-Rao.] The tensor is $\sf{G}_q^{FR}m:=m/q$ and has inverse 
\begin{equation}\label{eq:FRinv}((\sf{G}_q^{FR})^{-1}f)(x)=q(x)\left[f(x)-\int f(x) q(x)dx\right].\end{equation}
Hence,
\[\sf{g}_q^{FR}(m,m'):=\int \frac{m(x)}{q(x)}\frac{m'(x)}{q(x)}q(x)dx\quad\forall q\in\cal{P}(\cal{X}).\]
Its name stems from the fact that the usual Fisher-Rao metric on a parameter space $\Phi\subseteq\r^{D_\phi}$ indexing a parametric family   $(q_\phi)_{\phi\in\Phi}$ is obtained by   pulling $\sf{g}_q^{FR}$ back through $\phi\mapsto q_\phi$: 
\begin{align*}
\sf{g}_\phi^{pullback}(\beta,\beta')&=\sf{g}_{q_\phi}^{FR}(\iprod{\beta}{\nabla_\phi q_\phi},\iprod{\beta'}{\nabla_\phi q_\phi})=\int \frac{\iprod{\beta}{\nabla_\phi q_\phi(x)}}{q_\phi(x)}\frac{\iprod{\beta'}{\nabla_\phi q_\phi(x)}}{q_\phi(x)}q_\phi(x)dx\\
&=\int\iprod{\beta}{\nabla_\phi \log(q_\phi(x))}\iprod{\beta'}{\nabla_\phi \log(q_\phi(x))}q_\phi(x)dx=\iprod{\beta}{\cal{I}^\phi\beta'},\end{align*}
where $\cal{I}^\phi$ denotes the Fisher information matrix, i.e.\
\[\cal{I}^\phi:=\left(\int\frac{\partial \log(q_\phi)}{\partial\phi_i}(x)\frac{\partial \log(q_\phi)}{\partial\phi_j}(x)q_\phi(x)dx\right)_{ij=1}^{D_\phi}.\]
\end{description}

\subsection{Gradients}
Given a Riemannian metric $g$ on $\cal{M}$, the   gradient  of a functional $E$ on $\cal{M}$ is  defined as the unique  vector field  $\nabla^g E:\cal{M}\to\cal{T}\cal{M}$ satisfying
\begin{equation}\label{eq:gradientdef}g_{(\theta,q)}(\nabla^g E(\theta,q),(\tau,m))=\lim_{t\to0}\frac{E(\theta+t\tau,q+tm)-E(\theta,q)}{t}\quad\forall (\tau,m)\in\cal{T}\cal{M},\enskip (\theta,q)\in\cal{M}.\end{equation}%\quad\forall (\tau,m)\in\cal{T}_{(\theta,q)}\cal{M},\enskip (\theta,q)\in\cal{M}
%
%for all tangent vectors $(\tau,m)$ in $\cal{T}_{(\theta,q)}\cal{M}$ and points $(\theta,q)$ in $\cal{M}$. 
The following identity often simplifies gradient calculations:
\begin{equation}\label{eq:gradref}\nabla^gE(\theta,q)=G_{\theta,q}^{-1}\delta E(\theta,q)\quad\forall  (\theta,q)\in\cal{M},\end{equation}
where $\delta E:\cal{M}\to \cal{T}^*\cal{M}$ denotes $E$'s \emph{first variation}\footnote{Here lies the reason why we use the extra machinery of cotangent vectors, duality pairings, etc. Ideally, we would like to define $\delta E$ for a functional $E$ on $\cal{P}(\cal{X})$ to be the gradient w.r.t.\ the `flat $L^2$' metric on $\cal{P}(\cal{X})$: $\sf{g}^{L^2}_q(m,m'):=\int m(x)m'(x)dx$. (This is precisely what we do in Euclidean spaces, only w.r.t.\ the Euclidean metric.) However, doing so would require replacing $\delta E(q)$ with $\delta E(q)-\int \delta E(q) dx$ so that it lies in the tangent space (all tangent vectors must have zero mass). But the integral $\int \delta E(q) dx$ will not be well-defined in most cases and we hit a wall.}: the unique cotangent vector field  satisfying 
\begin{equation}\label{eq:ney8agneywa}\iprod{(\tau,m)}{\delta E(\theta,q)}=\lim_{t\to0}\frac{E(\theta+t\tau,q+tm)-E(\theta,q)}{t}\quad\forall (\tau,m)\in\cal{T}\cal{M},\enskip (\theta,q)\in\cal{M}.\end{equation}
%
%for all tangent vectors $(\tau,m)$ in $\cal{T}_{(\theta,q)}\cal{M}$ and points $(\theta,q)$ in $\cal{M}$.  
In turn, $\delta E$'s computation can be simplified using $\delta E=(\delta_\theta E,\delta_q E)$, where $\delta_{\theta}$ and $\delta_{q}$ denote the first variations on $\Theta$ and $\cal{P}(\cal{X})$ (defined analogously to~\eqref{eq:ney8agneywa} but for the maps $\theta\mapsto E(\theta,q)$ and $q\mapsto E(\theta,q)$, respectively).
\begin{lemma}\label{lem:Ffirstvar}In the case of the free energy, $F$ in~\eqref{eq:vfe}, $\delta F(\theta,q)=(\delta_\theta F(\theta,q),\delta_q F(\theta,q))$ where 
\[\delta_{\theta}F(\theta,q)=-\int \nabla_\theta \ell(\theta,x)q(x)dx,\quad \delta_{q}F(\theta,q) = \log\left(\frac{q}{p_\theta(\cdot,y)}\right),\quad\forall (\theta,q)\in\cal{M}.\]
\end{lemma}

\begin{proof}We need to show that, for any given $(\theta,q)$ in $\cal{M}$,
\begin{align}
 F(\theta+t\tau,q)&=F(\theta,q)-t\iprod{\tau}{\int \nabla_\theta \ell(\theta,x)q(x)dx}+o(t)\quad\forall  \tau\in\cal{T}\Theta,\label{eq:Ffvart}\\
F(\theta,q+tm)&=F(\theta,q)+t\iprod{m}{\log\left(\frac{q}{p_\theta(\cdot,y)}\right)}+o(t)\quad\forall  m\in\cal{T}\cal{P}(\cal{X}).\label{eq:Ffvarq}
\end{align}
We begin with \eqref{eq:Ffvart}: $\ell(\theta+t\tau,x)=\ell(\theta,x)+t\iprod{\tau}{\nabla_\theta \ell(\theta,x)}+o(t)$ and, so\footnote{\label{foot:ot}The $o(t)$ term in $\ell(\theta+t\tau,x)$'s expansion depends on $x$. Hence, to rigorously derive the ensuing expansion for $F(\theta,t+\tau,q)$, we require conditions on $\ell$ and/or $q$ guaranteeing that $\int o(t,x)q(x)dx=o(t)$. We abstain from stating such conditions to not complicate the exposition.},
\begin{align*}
F(\theta+t\tau,q)=&F(\theta,q) + \int \ell(\theta,x)q(x)dx - \int \ell(\theta+t\tau,x)q(x)dx\\
=&F(\theta,q)-\int\left[t\iprod{\tau}{ \nabla_\theta \ell(\theta,x)}+o(t)\right]q(x)dx\\
=&F(\theta,q)-t\iprod{\tau}{\int \nabla_\theta \ell(\theta,x)q(x)dx}+o(t).
\end{align*}
For~\eqref{eq:Ffvarq} instead note that $\log(z+t)(z+t)=\log(z)z+[\log(z)+1]t+o(t)$, whence
\begin{align*}
F(\theta,q+tm)=&\int \log(q(x)+tm(x))(q(x)+tm(x))dx-\int \ell(\theta,x)(q(x)+tm(x))dx\\
=&\int \left[\log(q(x))q(x)+[\log(q(x))+1]tm(x)+o(t)\right]dx\\
&-\int \ell(\theta,x)q(x)dx - t\int \log(p_\theta(x,y))m(x)dx \\
=&F(\theta,q)+ t\int [\log(q(x))-\log(p_\theta(x,y))]m(x)dx+t\int m(x)dx+o(t);
\end{align*}
and \eqref{eq:Ffvarq} follows because $\int m(x)dx=0$ given that $m$ belongs to $\cal{T}\cal{P}(\cal{X})$ (cf.\ App.~\ref{app:tancon}).
\end{proof}
For metrics $g$ with a block-diagonal tensor diag$(\rm{G}_{(\theta,q)},\sf{G}_{(\theta,q)})$, we have one final simplification:
\begin{equation}\label{eq:gradrefblock}\nabla ^gE(\theta,q)=(\rm{G}_{\theta,q}^{-1}\delta_{\theta}E(\theta,q),\sf{G}_{\theta,q}^{-1}\delta_{q}E(\theta,q))\quad\forall (\theta,q)\in\cal{M}.\end{equation}
\paragraph{The direction of maximum descent.}To gain some intuition regarding what we actually do by `taking a step in the direction of $-\nabla^gE(\theta,q)$', note that, for sufficiently regular functionals $E$,
\begin{equation}\label{eq:ETaylor1}E(\theta+\tau,q+m)\approx E(\theta,q)+\iprod{(\tau,m)}{\delta E(\theta,q)}=E(\theta,q)+g_{(\theta,q)}((\tau,m),\nabla^g E(\theta,q))\end{equation}
for any given point given point $(\theta,q)$ in $\cal{M}$ and `small' tangent vectors $(\tau,m)$ in $\cal{T}\cal{M}$, with the equality holding exactly in the limit as ``$(\tau,m)$'s size tends to zero''. To quantify ``$(\tau,m)$'s size'', we use the norm on $\cal{T}\cal{M}$ induced by our metric:
\[\norm{(\tau,m)}^g_{(\theta,q)}:=g_{(\theta,q)}((\tau,m),(\tau,m)).\]
Armed with the above, we can then ask `out of all tangent vectors $(\tau,m)$ of size $\varepsilon$, which lead to the greatest decrease in $E$ at $(\theta,q)$?'. That is, which $(\tau^*,m^*)$ solve
\[\min_{\norm{(\tau,m)}^g_{(\theta,q)}=\varepsilon} E(\theta+\tau,q+m)?\] 
Were we to swap $E(\theta+\tau,q+m)$ in the above with its approximation in~\eqref{eq:ETaylor1}, the Cauchy-Schwarz inequality would then tell us that $(\tau^*,m^*)$ equals the (appropriately rescaled) gradient $-\varepsilon \nabla^g E(\theta,q)/\norm{\nabla^g E(\theta,q)}^g_{(\theta,q)}$:
\[\argmin_{\norm{(\tau,m)}^g_{(\theta,q)}=\varepsilon} E(\theta+\tau,q+m)\approx\argmin_{\norm{(\tau,m)}^g_{(\theta,q)}=\varepsilon}g_{(\theta,q)}((\tau,m),\nabla^g E(\theta,q))=\frac{\varepsilon\nabla^g E(\theta,q)}{\norm{\nabla^g E(\theta,q)}^g_{(\theta,q)}}.\]
Assuming that the above equality holds exactly as $\varepsilon\to0$, we find that  $\nabla^g E(\theta,q)$ points in the direction of steepest descent for $E$ at $(\theta,q)$ in the geometry defined by $g$ (that is, using the norm induced by $g$ to measure the length of vectors).

\paragraph{Minimizing quadratic functionals.}We are now faced with the question `which geometry or metric  $g$ should we use to define gradients?'. While in practice this question often gets usurped by the more pragmatic `which geometries lead to gradient flows that can be efficiently approximated?', considering which geometries are most attractive, even if only in a theoretical sense, still proves insightful. A straightforward way to approach this question is noting that \eqref{eq:gradref} implies that $\nabla^gE(\theta,q)$ solves 
\[G_{\theta,q}(\tau,m)=\delta E(\theta,q).\]
It follows that $(\theta,q)-\nabla_{(\theta,q)}^gE(\theta,q)$ minimizes a quadratic approximation to $E$ around $(\theta,q)$:
\begin{equation}\label{eq:quadratic}\nabla^gE(\theta,q)=\argmin_{(\tau,m)\in\cal{T}\cal{M}}\left\{E(\theta,q)+\iprod{(\tau,m)}{\delta E(\theta,q)} +\frac{1}{2}\iprod{G_{\theta,q}(\tau,m)}{(\tau,m)}\right\}.\end{equation}
From this vantage point, it seems natural to pick $G$ so that the objective in~\eqref{eq:quadratic} closely approximates $E$ around $(\theta,q)$.  
We revisit this point for the free energy $F$ in App.~\ref{app:pqn}.
\section{PROOF OF THEOREM~\ref{thrm:expconv} AND FURTHER THEORETICAL DETAILS FOR SEC.~\ref{sec:pgd}}\label{app:stateqthrmproof}

\subsection{(\ref{eq:Fgradient1},\ref{eq:Fgradient2}) as a gradient}\label{app:grad}
Here, we use the geometry on $\cal{M}=\Theta\times\cal{P}(\cal{X})$ which leads to the gradient flow with the cheapest and most straightforward approximations that we know of (e.g.\ compare with the geometries in~\cite{Liu2016,Garbuno-Inigo2020}, analogously extended from $\cal{P}(\cal{X})$ to $\cal{M}$): the one obtained as the product of the Euclidean geometry on $\Theta$ and the Wasserstein-2 geometry on $\cal{P}(\cal{X})$. More formally, the geometry induced  by the metric with block-diagonal tensor diag$(\rm{I}_{D_\theta},\sf{G}_{q}^W)$ (cf.\ App.~\ref{app:metrics}), where $\rm{I}_{D_\theta}$ denotes the identity operator on $\cal{T}\Theta$ (i.e.~$\rm{I}_{D_\theta}\tau=\tau$ for all $\tau$ in $\cal{T}\Theta$) and $\sf{G}_{q}^W$ the Wasserstein-$2$ tensor on $\cal{P}(\cal{X})$ in~\eqref{eq:W2inv}. Combining Lem.~\ref{lem:Ffirstvar} and (\ref{eq:W2inv},\ref{eq:gradrefblock}) we find that $F$'s gradient is given by~(\ref{eq:Fgradient1},\ref{eq:Fgradient2}), and its corresponding gradient flow by~(\ref{eq:IW2Gradflow1},\ref{eq:IW2Gradflow2}).

%\subsection{Proof of Theorem~\ref{thrm:statpoints}}
%%
%Examining~(\ref{eq:Fgradient1},\ref{eq:Fgradient2}) we see that $\nabla_q F(\theta,q)=0$ if and only if $q\propto p_\theta(\cdot,y)$. Given that $q$ is a probability distribution, it follows that $\nabla_q F(\theta,q)=0$ if and only if $q=p_\theta(\cdot|y)$. The result then follows from
%%
%%
%\begin{align}
%\nabla_\theta p_\theta(y)&=\int \nabla_\theta p_\theta(x,y)dx=\int \nabla_\theta \ell(\theta,x)p_\theta(x,y)dx=p_\theta(y)\int \nabla_\theta \ell(\theta,x)p_\theta(x|y)dx\label{eq:nsya8bnyfsanj7ufa}\\
%&=-p_\theta(y)\nabla_\theta F(\theta,p_\theta(\cdot|y)).\nonumber
%\end{align}

\subsection{On the convergence of the gradient flow}\label{app:conv}

As we will show below, if $(\theta_t,q_t)_{t\geq0}$ satisfies~(\ref{eq:IW2Gradflow1},\ref{eq:IW2Gradflow2}), then
\begin{equation}
\label{eq:It}I_t:=-\frac{d F(\theta_t,q_t)}{dt}=||\dot{\theta}_t||^2+\int\norm{\nabla_x R_t(x)}^2q_t(x)dx\geq0,
\end{equation}
where  $R_t(x):=\log(p_{\theta_t}(x,y)/q_t(x))$ and  $\norm{\cdot}$ denotes the  Euclidean norm on $\r^{D_\theta}$ or $\r^{D_x}$, as appropriate. In other words, the free energy is non-increasing along~(\ref{eq:IW2Gradflow1},\ref{eq:IW2Gradflow2})'s solutions: $F(\theta_t,q_t)\leq F(\theta_0,q_0)$ for all $t\geq0$. Moreover, because $q\mapsto F(\theta,q)$ is minimized at $p_\theta(\cdot|y)$ (Thrm.~\ref{thrm:vfe}),
\[\log(p_{\theta_t}(y))=\int\log\left(\frac{p_{\theta_t}(x,y)}{p_{\theta_t}(x|y)}\right)p_{\theta_t}(x|y)dx=-F(\theta_t,p_{\theta_t}(\cdot|y))\geq -F(\theta_t,q_t)\geq -F(\theta_0,q_0)\quad\forall t\geq0;\]
and it follows from Assumpt.~\ref{ass:nodiv} that $\{\theta_t\}_{t\geq0}$ is relatively compact. Hence, an extension of LaSalle's principle along the lines of~\cite{Carrillo2020} should imply that,  as $t$ tends to infinity,  $(\theta_t,q_t)$ approaches the set of points that make~\eqref{eq:It}'s RHS vanish. But we can re-write the RHS as
\[g((\dot{\theta}_t,\dot{q}_t),(\dot{\theta}_t,\dot{q}_t))=g(\nabla F({\theta}_t,q_t),\nabla F({\theta}_t,q_t))\leq0,\]
where $g$ denotes the metric described in App.~\ref{app:grad} and $\nabla$ the corresponding gradient (whose components are given by~(\ref{eq:Fgradient1},\ref{eq:Fgradient2})). In other words,  $(\theta_t,q_t)$ approaches the set of pairs that make $F$'s gradient vanish. Thrm.~\ref{thrm:statpoints} tells us that these pairs  $(\theta_*,q_*)$ are precisely those for which $\theta_*$ is a stationary point of the marginal likelihood and $q_*$ is its corresponding posterior $p_{\theta_*}(\cdot|y)$.

\begin{proof}[Proof of~\eqref{eq:It}]Using the chain rule and integration by parts, we find that
\begin{align*}
\frac{d}{dt}\int \ell(\theta_t,x) q_t(x)dx&=\int \frac{d\ell(\theta_t,x)}{dt} q_t(x)dx+\int \ell(\theta_t,x) \dot{q}_t(x)dx\\
&=\int \iprod{\nabla_\theta \ell(\theta_t,x)}{\dot{\theta}_t} q_t(x)dx-\int \ell(\theta_t,x) \nabla_x\cdot \left[q_t(x)\nabla_x R_t(x)\right]dx\\
&=||\dot{\theta}_t||^2+\int \iprod{\nabla_x \ell(\theta_t,x)}{\nabla_x R_t(x)}q_t(x)dx,
\end{align*}
and 
\begin{align*}
\frac{d}{dt}\int \log(q_t(x))q_t(x)dx=&\int [\log(q_t(x))+1]\dot{q}_t(x)dx=-\int [\log(q_t(x))+1]\nabla_x\cdot \left[q_t(x)\nabla_x R_t(x)\right]dx\\
=&\int \iprod{\nabla_x \log(q_t(x))}{\nabla_x R_t(x)}q_t(x)dx,
\end{align*}
where $\iprod{\cdot}{\cdot}$ denotes the  Euclidean inner product on $\r^{D_\theta}$ or $\r^{D_x}$, as appropriate. Re-arranging, we obtain~\eqref{eq:It}.
\end{proof}

\subsection{Proof of Theorem~\ref{thrm:expconv}}\label{app:expoproof}

For models with sufficiently regular strongly log-concave densities, it is straightforward to give a more complete argument for $(\theta_t,q_t)_{t\geq0}$'s convergence than that in App.~\ref{app:conv}. In these cases, the marginal likelihood has a unique maximizer:
\begin{theorem}\label{thrm:strilogconc}Suppose that $(\theta,x)\mapsto \ell(\theta,x)$ is twice continuously differentiable. Moreover, that $\ell$ is strictly concave or, in other words, that its Hessian negative definite everywhere:
\begin{equation}\label{eq:Llc}\nabla^2\ell(\theta,x)=\begin{bmatrix}\nabla_\theta^2 \ell(\theta,x)&\nabla_\theta\nabla_x \ell(\theta,x)\\\nabla_x\nabla_\theta \ell(\theta,x)&\nabla_x^2 \ell(\theta,x)\end{bmatrix}\prec 0\quad\forall \theta\in\Theta,\enskip x\in\cal{X}.\end{equation}
Then, the marginal likelihood $\theta\mapsto p_\theta(y)$ has a unique maximizer and no other stationary point.
\end{theorem}

\begin{proof}Because $\nabla_\theta p_\theta(y)=p_\theta(y)\nabla_\theta \log(p_\theta(y))$ and $z\mapsto\log(z)$ is a strictly increasing function, it suffices to show that $\theta\mapsto \log(p_\theta(y))$ is strictly concave. To this end, note that
\[
\nabla_\theta^2 \log(p_\theta(y))=\nabla_\theta \frac{\nabla_\theta p_\theta(y)}{p_\theta(y)}= \frac{\nabla_\theta^2 p_\theta(y)}{p_\theta(y)}-\frac{\nabla_\theta p_\theta(y)\otimes \nabla_\theta p_\theta(y)}{p_\theta(y)^2},
\]
where $v\otimes v':=(v_iv_j)_{ij=1}^{D_\theta}$ for any vectors $v,v'\in\r^{D_\theta}$. But,
\begin{align*}
\nabla_\theta p_\theta(y) &= \int \nabla_\theta p_\theta(x,y)dx=\int \nabla_\theta \ell(\theta,x)p_\theta(x,y)dx,\\
\nabla_\theta^2 p_\theta(y) &= \int \nabla_\theta^2 \ell(\theta,x)p_\theta(x,y)dx+\int \nabla_\theta \ell(\theta,x)\otimes \nabla_\theta p_\theta(x,y)dx\\
&= \int \nabla_\theta^2 \ell(\theta,x)p_\theta(x,y)dx+\int \nabla_\theta \ell(\theta,x)\otimes \nabla_\theta \ell(\theta,x)p_\theta(x,y)dx;
\end{align*}
and, so,
\begin{align*}
\nabla_\theta^2 \log(p_\theta(y))&=\int \nabla_\theta^2 \ell(\theta,x)p_\theta(x|y)dx+\Sigma(\theta),
\end{align*}
where
\[\Sigma(\theta):=\int \nabla_\theta \ell(\theta,x)\otimes \nabla_\theta \ell(\theta,x)p_\theta(x|y)dx-\int \nabla_\theta \ell(\theta,x)p_\theta(x|y)dx\otimes \int \nabla_\theta \ell(\theta,x)p_\theta(x|y)dx.\]
By the Brascamp-Lieb concentration inequality~\citep[Thrm.~4.1]{Brascamp1976},
\begin{align*}
\Sigma(\theta) &\preceq  \int \nabla_\theta \nabla_x \ell(\theta,x)[-\nabla_x^2\log (p_\theta(x|y))]^{-1} \nabla_x \nabla_\theta \ell(\theta,x)p_\theta(x|y)dx\\
&	=-\int\nabla_\theta \nabla_x \ell(\theta,x)[\nabla_x^2\ell(\theta,x)]^{-1} \nabla_x \nabla_\theta \ell(\theta,x)p_\theta(x|y)dx.
\end{align*}
In short,
\begin{align*}
\nabla_\theta^2 \log(p_\theta(y))\preceq \int\left[ \nabla_\theta^2 \ell(\theta,x)-\nabla_\theta \nabla_x \ell(\theta,x)[\nabla_x^2\ell(\theta,x)]^{-1} \nabla_x \nabla_\theta \ell(\theta,x)\right]p_\theta(x|y)dx.
\end{align*}
The integrand is the Schur complement of $\ell$'s Hessian and, hence, negative definite for all $(\theta,x)$. Moreover, because $\nabla_x^2\ell$ is negative definite everywhere and $\ell$ is twice-continuously differentiable, the integrand varies continuously in $x$; whence it follows that the integral is negative definite for all $\theta$. In other words, $\theta\mapsto \log(p_\theta(y))$ is strictly concave.
\end{proof}
%
%If, furthermore, the density is strongly log-concave, the flow converges exponentially fast:
%%
%\begin{theorem}Suppose that the log likelihood $\ell$ is twice continuously differentiable; strongly concave, 
%%
%\[\nabla^2 \ell(\theta,x) 	\preceq -\lambda I_{D_x+D_\theta}\quad\forall (\theta,x)\in\Theta\times\cal{X},\]
%%
%for some $\lambda>0$; and has bounded $\theta$-gradient,
%%
%\[\norm{\nabla_\theta \ell(\theta,x)}\leq C\quad\forall (\theta,x)\in\Theta\times\cal{X},\]
%%
%for some $C>0$. Then, the marginal likelihood has a unique maximizer $\theta_*$ and, if $I_0<\infty$,
%%
%\[\norm{\theta_t-\theta_*}\leq C'e^{-\lambda t}\enskip\text{and}\enskip \norm{q_t-p_{\theta_*}(\cdot|y)}_{L^1}\leq C'e^{-\lambda t},\quad\forall t\geq0,\]
%%
%for some $C'>0$, where $\norm{\cdot}$ denotes the Euclidean norm on $\r^{D_\theta}$ and $\norm{\cdot}_{L^1}$ the $L^1$ norm on $\cal{P}(\cal{X})$ (recall that we are conflating $\cal{P}(\cal{X})$ with its subset of distributions with positive densities w.r.t.\ the Lebesgue measure and identifying each distribution with its density).
%\end{theorem}

We are now ready to tackle Theorem~\ref{thrm:expconv}'s proof:

\begin{proof}[Proof of Theorem~\ref{thrm:expconv}]As we will show below,
\begin{equation}\label{eq:dIt}
\frac{dI_t}{dt}\leq 2 \lambda  I_t;
\end{equation}
from which it follows that $I_t\leq e^{-2 \lambda t}I_0$. Hence,
\[\int_0^\infty \norm{\dot{\theta}_t}dt\leq \int_0^\infty e^{-\lambda t}\sqrt{I_0}dt=\frac{\sqrt{I_0}}{\lambda}.\]
Thus, $\theta_\infty:=\int_0^\infty \dot{\theta}_tdt$ is well-defined and $\theta_t$ converges to $\theta_\infty$ exponentially fast:
\[\norm{\theta_\infty-\theta_t}\leq \norm{\int_t^\infty \dot{\theta_s}ds}\leq \int_t^\infty\norm{ \dot{\theta_s}}ds\leq \frac{\sqrt{I_0}}{\lambda}e^{-\lambda t}.\]

Next, using the fact that $\log(p_{\theta_t}(x|y)/q_t(x))=R_t-\log(p_{\theta_t}(y))$, we find that
\[\int\norm{\nabla_x R_t(x)}^2q_t(x)dx=\int\norm{\nabla_x\log\left(\frac{p_{\theta_t}(x|y)}{q_t(x)}\right)}^2q_t(x)dx.\]
Because $\nabla^2_x\log(p_{\theta_t}(x|y))=\nabla^2_x \ell(\theta_t,x)\preceq -\lambda I_{D_x}$ for all $t\geq0$, a logarithmic Sobolev inequality and the Csiszár-Kullback-Pinsker inequality, Thrm.~1 and (12) in \cite{Markowich2000} respectively, then imply that
\[\frac{1}{2}\norm{q_t-p_{\theta_t}(\cdot|y)}_{L^1}^2\leq KL(q_t||p_{\theta_t}(\cdot|y))\leq\int\norm{\nabla_x R_t(x)}^2q_t(x)dx\leq e^{-2\lambda t}I_0.\]
As we will show below, the boundedness assumption on $\ell$'s $\theta$-gradient implies that $\theta\mapsto p_\theta(\cdot|y)$ is a Lipschitz map from $(\Theta,\norm{\cdot})$ to $(\cal{P}(\cal{X}),\norm{\cdot}_{L_1})$:
\begin{equation}\label{eq:lipschitz}\norm{p_\theta(\cdot|y)-p_{\theta'}(\cdot|y)}\leq 2C \norm{\theta-\theta'}.
\end{equation}
Applying the triangle inequality we then find that $q_t$ converges exponentially fast to $p_{\theta_\infty}(\cdot|y)$:
\begin{align*}
\norm{q_t-p_{\theta_\infty}(\cdot|y)}_{L^1}&\leq \norm{p_{\theta_t}(\cdot|y)-p_{\theta_\infty}(\cdot|y)}_{L^1}+\norm{q_t-p_{\theta_t}(\cdot|y)}_{L^1}\leq 2C\norm{\theta_t-\theta_\infty}+\sqrt{2 I_0}e^{-\lambda t}\\
&\leq (2C+\sqrt{2})\sqrt{I_0}e^{-\lambda t}.
\end{align*}
Given Thrm.~\ref{thrm:strilogconc}, the only thing we have left to do is argue that the limit $\theta_\infty$ is a stationary point of the marginal likelihood. This follows from \eqref{eq:nsya8bnyfsanj7ufa}, the bounded convergence theorem, and our assumption that $\nabla_\theta \ell$ is bounded:
\begin{align*}
\frac{\nabla_\theta p_{\theta_\infty}(y)}{p_{\theta_\infty}(y)}&=\int\nabla_\theta \ell(\theta_\infty,x)p_{\theta_\infty}(x|y)dx=\lim_{n\to\infty} \int\nabla_\theta \ell(\theta_n,x)p_{\theta_n}(x|y)dx\\
&=\lim_{n\to\infty}\left[\dot{\theta}_n +\int\nabla_\theta \ell(\theta_n,x)[p_{\theta_n}(x|y)-q_n(x)]dx\right]=0.
\end{align*}
\end{proof}

\begin{proof}[Proof of~\eqref{eq:dIt}]Here, we adapt the arguments in \citet[Sec.~5]{Markowich2000} and \citet[Sec.~2.3]{Arnold2001}. Let's start: $(d||\dot{\theta}_t||^2/dt)=2\iprod{\dot{\theta}_t}{\ddot{\theta}_t}$ and, using the notation introduced in~\eqref{eq:Llc},
\begin{align*}
\ddot{\theta}_t&=\frac{d }{dt}\dot{\theta}_t=\frac{d }{dt}\int\nabla_\theta \ell(\theta_t,x)q_t(x)dx=\int \left[\frac{d }{dt}\nabla_\theta \ell(\theta_t,x)\right]q_t(x)dx+\int\nabla_\theta \ell(\theta_t,x)\dot{q}_t(x)dx\\
&=\int \nabla_\theta^2 \ell(\theta_t,x)\dot{\theta}_tq_t(x)dx-\int \nabla_\theta \ell(\theta_t,x)\nabla_x\cdot \left[q_t(x)\nabla_x R_t(x)\right]dx\\
&=\int \nabla_\theta^2 \ell(\theta_t,x)\dot{\theta}_tq_t(x)dx+\int [\nabla_\theta\nabla_x  \ell(\theta_t,x)]\nabla_x R_t(x)q_t(x)dx,
\end{align*}
where the last equality follows from integration by parts. Hence,
\begin{align}
\frac{d}{dt}||\dot{\theta}_t||^2=&2\int\left[\iprod{\dot{\theta}_t}{ \nabla_\theta^2 \ell(\theta_t,x)\dot{\theta}_t}+\iprod{\dot{\theta}_t}{[\nabla_\theta \nabla_x  \ell(\theta_t,x)]\nabla_x R_t(x)}\right]q_t(x)dx.
\end{align}
Similarly, 
\begin{align*}
\frac{d }{dt}\int\norm{R_t(x)}^2q_t(x)dx=&\int\left[\frac{d }{dt}\norm{\nabla_x R_t(x)}^2\right]q_t(x)dx+\int\norm{\nabla_x R_t(x)}^2\dot{q}_t(x)dx.
\end{align*}
But, with $l_t(x):=\log(q_t(x))$,
\begin{align*}
&\frac{d}{dt}\nabla_x R_t(x)=\frac{d}{dt}\nabla_x \ell(\theta_t,x)-\frac{d}{dt}\nabla_xl_t(x)=[\nabla_x\nabla_\theta \ell(\theta_t,x)]\dot{\theta}_t-\nabla_x\frac{d}{dt}l_t(x),\\
\Rightarrow&\frac{d}{dt}\norm{\nabla_x R_t(x)}^2=2\iprod{\nabla_x R_t(x)}{\frac{d}{dt}\nabla_x R_t(x)}\\
&=2\iprod{\nabla_x R_t(x)}{[\nabla_x\nabla_\theta \ell(\theta_t,x)]\dot{\theta}_t}-2\iprod{\nabla_x R_t(x)}{\nabla_x\frac{d}{dt} l_t(x)},\\
\Rightarrow&\int\left[\frac{d }{dt}\norm{\nabla_x R_t(x)}^2\right]q_t(x)dx\\
&=2\int\left[\iprod{\nabla_x R_t(x)}{[\nabla_x\nabla_\theta \ell(\theta_t,x)]\dot{\theta}_t}-\iprod{\nabla_x R_t(x)}{\nabla_x\frac{d}{dt} l_t(x)}\right]q_t(x)dx;
\end{align*}
and
\begin{align*}
&\int\norm{\nabla_x R_t(x)}^2\dot{q}_t(x)dx=\int\iprod{\nabla_x \norm{\nabla_x R_t(x)}^2}{\nabla_x R_t(x)}q_t(x)dx\\
&=2\int\iprod{\nabla_x R_t(x) }{\nabla_x^2 R_t(x)\nabla_x R_t(x)}q_t(x)dx\\
&=2\int\left[\iprod{\nabla_x R_t(x) }{\nabla_x^2 \ell(\theta_t,x)\nabla_x R_t(x)}-\iprod{\nabla_x R_t(x) }{\nabla_x^2  l_t(x)\nabla_x R_t(x)}\right]q_t(x)dx.
\end{align*}
Putting the above together, we find that
\begin{align*}
\frac{dI_t}{dt}=&2\int \iprod{(\dot{\theta}_t,\nabla_x R_t(x))}{\nabla^2 \ell(\theta_t,x) (\dot{\theta}_t,\nabla_x R_t(x))}q_t(x)dx\\
&-2\int\iprod{\nabla_x R_t(x) }{\nabla_x\frac{d}{dt} l_t(x)+\nabla_x^2  l_t(x)\nabla_x R_t(x)}q_t(x)dx\\
\leq&-2\lambda I_t -2\int\iprod{\nabla_x R_t(x) }{\nabla_x\frac{d}{dt} l_t(x)+\nabla_x^2  l_t(x)\nabla_x R_t(x)}q_t(x)dx=:-2\lambda I_t -2A.
\end{align*}
(The inequality follows from our assumption that $(\theta,x)\mapsto p_\theta(x,y)$ is $\lambda$-strongly log-concave.) We now need to show that $A$ is no greater than zero. To this end, note that
\begin{align*}
\frac{d}{dt} l_t(x)=\frac{\dot{q}_t(x)}{q_t(x)}&=-\frac{\nabla_x\cdot [q_t(x)\nabla_x R_t(x)]}{q_t(x)}=-\frac{\iprod{\nabla_x q_t(x)}{\nabla_x R_t(x)}}{q_t(x)}-\Delta_x R_t(x)\\
&=-\iprod{\nabla_x l_t(x)}{\nabla_x R_t(x)}-\Delta_x R_t(x),
\end{align*}
where $\Delta_x$ denotes the Laplacian operator; from which it follows that
\begin{align*}
\nabla_x\frac{d}{dt} l_t(x)=&-\nabla_x^2l_t(x)\nabla_x R_t(x)-\nabla_x^2R_t(x)\nabla_x l_t(x)-\nabla_x\Delta_x R_t(x).
\end{align*}
Bochner's formula tells us that
\[-\iprod{\nabla_x R_t(x)}{\nabla_x\Delta_x R_t(x)}=\textrm{tr}([\nabla^2_x R_t(x)]^T\nabla^2_x R_t(x))-\frac{1}{2}\Delta_x\norm{\nabla_x R_t(x)}^2,\]
where $\textrm{tr}(\cdot)$ denotes the trace operator. But
\begin{align*}
-\frac{1}{2}\int q_t(x)\Delta_x\norm{\nabla_x R_t(x)}^2dx&=\frac{1}{2}\int \iprod{\nabla_x q_t(x)}{\nabla_x\norm{\nabla_x R_t(x)}^2}dx\\
&=\int \iprod{\nabla_x l_t(x)}{\nabla_x^2 R_t(x)\nabla_x R_t(x)}q_t(x)dx.
\end{align*}
Hence,
\begin{align*}
A=\int\textrm{tr}([\nabla^2_x R_t(x)]^T\nabla^2_x R_t(x))q_t(x)dx\geq0.
\end{align*}
\end{proof}
\begin{proof}[Proof of \eqref{eq:lipschitz}]The mean value theorem tells us that, for each $\theta,\theta',x$, there exists a $\psi$ such that
\begin{align*}
\mmag{p_\theta(x|y)-p_{\theta'}(x|y)}=\mmag{\iprod{\theta-\theta'}{\nabla_\theta p_{\psi}(x|y)}}\leq\norm{\theta-\theta'}\norm{\nabla_\theta p_{\psi}(x|y)}.
\end{align*}
We will now show that $\norm{\nabla_\theta p_{\psi}(x|y)}\leq 2C  p_{\psi}(x|y)$, from which the claim will follow:
\begin{align*}\int \mmag{p_\theta(x|y)-p_{\theta'}(x|y)}dx&\leq \norm{\theta-\theta'}\int\norm{\nabla_\theta p_{\psi}(x|y)}dx\leq 2C\norm{\theta-\theta'}.\end{align*}
To obtain $\norm{\nabla_\theta p_{\psi}(x|y)}\leq 2C  p_{\psi}(x|y)$, note that
\begin{align*}
\nabla_\psi p_\psi(x|y)&=\frac{\nabla_\theta p_\psi(x,y)}{p_\psi(y)}-\frac{\nabla_\theta p_\psi(y)}{p_\psi(y)}\frac{p_\psi(x,y)}{p_\psi(y)}\\
&=\left[\nabla_\theta \ell(\psi,x)-\int \nabla_\theta \ell(\psi,x')p_\psi(x'|y)dx'\right]p_\psi(x|y).
\end{align*}
But,
\begin{align*}
\norm{\nabla_\theta \ell(\psi,x)-\int \nabla_\theta \ell(\psi,x')p_\psi(x'|y)}&\leq \norm{\nabla_\theta \ell(\psi,x)}+\norm{\int \nabla_\theta \ell(\psi,x')p_\psi(x'|y)dx'}\\
&\leq C+\int \norm{\nabla_\theta \ell(\psi,x')}p_\psi(x'|y)dx'\leq 2C.
\end{align*}
\end{proof}

\section{PARTICLE QUASI-NEWTON (PQN)}\label{app:pqn}
A variant of PGD (Alg.~\ref{alg:pgd}) that also seems to resolve the ill-conditioning discussed in Sec.~\ref{sec:pgd} and, furthermore, achieves faster convergence is PQN (Alg.~\ref{alg:pqn}). 
\begin{algorithm}[h]
\begin{algorithmic}[1]
\STATE{\textbf{Inputs:} step size $h$, step number $K$, particle number $N$, and initial particles $X^1_0,\dots,X_0^N$ and  parameters $\theta_0$.}
\FOR{$k=0,\dots, K-1$}
\STATE{Update the parameter estimates:\vspace{-2pt}
\begin{equation}\label{eq:newtonalg}\theta_{k+1} = \theta_{k} - h\left[\sum_{n=1}^{N}\nabla_\theta^2\ell(\theta_k,X_k^n)\right]^{-1}\sum_{n=1}^N \nabla_\theta \ell(\theta_k,X_k^n).\end{equation}\vspace{-8pt}}
\STATE{Update the particles: for all $n=1,\dots,N$,\[
X_{k+1}^n=X_k^n+h\nabla_x \ell(\theta_k,X_k^n)+\sqrt{2h}W_k^n,\]
with $W_k^1,\dots,W_k^N$ denoting i.i.d.\ $\cal{N}(0,I_{D_x})$ R.V.s.}
\ENDFOR
\RETURN{$(\theta_k,q_k:=N^{-1}\sum_{n=1}^N\delta_{X_k^n})_{k=0}^K$.}
\end{algorithmic}
 \caption{Particle Quasi-Newton (PQN).}
 \label{alg:pqn}
\end{algorithm}
In short, it amounts to replacing the parameter estimates' update equation~(\ref{eq:IW2gradalg11}) with \eqref{eq:newtonalg}, where $\nabla_\theta^2\ell(\theta,x)$ denotes the log-likelihood's $\theta$-Hessian (which we assume is full-rank for all $(\theta,x)$ in $\Theta\times\cal{X}$).  In PQN, we also use $\theta_K,q_K$, or~\eqref{eq:estimators} to obtain estimates of the marginal likelihood's stationary points and their associated posteriors.  (\ref{eq:newtonalg},\ref{eq:IW2gradalg12}) arises as a discretization of~\eqref{eq:IW2Gradflow2} and
\begin{align}
\dot{\theta}_t&= -\left[\int\nabla_\theta^2\ell(\theta_t,x)q_t(x)dx\right]^{-1}\int \nabla_\theta \ell(\theta_t,x)q_t(x)dx,\label{eq:newtonflowapprox}%\\
%\dot{q}_t&=-\nabla_x\cdot\left[ q_t\nabla_x \log\left(\frac{p_{\theta_t}(\cdot,y)}{q_t}\right)\right],\label{eq:newtonflowapprox2}
\end{align}
In turn, (\ref{eq:newtonflowapprox},\ref{eq:IW2Gradflow2}) is satisfied by the law of the following McKean-Vlasov SDE:
\begin{align}\label{eq:newtonSDE1}
d\theta_t = -\left[\int\nabla_\theta^2\ell(\theta_t,x)q_t(x)dx\right]^{-1}\left[\int \nabla_\theta \ell(\theta_t,x)q_t(x)dx\right]dt,\quad
dX_t=\nabla_x \ell(\theta_t,X_t)dt+\sqrt{2}dW_t,%\label{eq:newtonSDE2}
\end{align}
where $q_t$ denotes $X_t$'s law and $(W_t)_{t\geq0}$ a standard $D_x$-dimensional Brownian motion. We obtain~(\ref{eq:newtonalg}) by following the same steps as in Sec.~\ref{sec:pgd}, only with (\ref{eq:newtonSDE1}) replacing (\ref{eq:IW2gradSDE1},\ref{eq:IW2gradSDE2}) and an extra approximation in~\eqref{eq:finsampleapprox}:
\[\left[\int\nabla_\theta^2\ell(\theta_t,x)q_t(x)dx\right]^{-1}\approx \left[\frac{1}{N}\sum_{n=1}^N\nabla_\theta^2\ell(\theta_t,X_t^n)\right]^{-1} .\]

(\ref{eq:newtonflowapprox},\ref{eq:IW2Gradflow2}) form an approximation to $F$'s Newton flow (analogous to (\ref{eq:IW2Gradflow1},\ref{eq:IW2Gradflow2}) except that we follow the Newton direction rather than the negative gradient), see Apps.~\ref{app:basicnewton}--\ref{app:newtonflow} below. Our full-rank assumption implies that (\ref{eq:newtonflowapprox},\ref{eq:IW2Gradflow2})'s fixed points are $F$'s stationary points, and Thrm.~\ref{thrm:statpoints} applies as before.

At first glance, \eqref{eq:newtonalg} mitigates the ill-conditioning discussed in Sec.~\ref{sec:pgd} for the same reason that \eqref{eq:IW2gradalg2} does: $\nabla_\theta^2\ell$'s entries generally have a similar number of terms to $\nabla_\theta \ell$'s, which prevents excessively large parameter updates. In fact, for the toy model in Ex.~\ref{ex:hier}, $\nabla_\theta^2\ell\equiv -\Lambda^{-1}= -D_x$ and (\ref{eq:IW2gradalg2},\ref{eq:newtonalg}) coincide.  A bit less superficially, this might be because the RHS of the equations in~(\ref{eq:newtonflowapprox},\ref{eq:IW2Gradflow2}) approximate $F$'s Newton direction (c.f.\ App.~\ref{app:newtonflow}) at $(\theta_t,q_t)$ and, hence, better account for the effect that the updates have on $F$'s value. This is also the reason why we believe that (\ref{eq:newtonalg},\ref{eq:IW2gradalg12}) often converges faster than (\ref{eq:IW2gradalg11},\ref{eq:IW2gradalg12}), e.g.\ see Fig.~\ref{fig:hier}b,c and App.~\ref{app:hierconvrates}. The price to pay is the extra cost incurred by the Hessian evaluations and the matrix inversion in~\eqref{eq:newtonalg}, which, absent any special structure in $ \nabla_\theta^2\ell$ (e.g.\, diagonal or banded), results in PQN's computational complexity equalling
\[\cal{O}(K[D_\theta^3+N[\text{eval. cost of } (\nabla_\theta\ell,\nabla_x\ell,\nabla_\theta^2\ell)]]).\]
The evaluation costs of $(\nabla_\theta\ell,\nabla_x\ell,\nabla_\theta^2\ell)$ is often linear in $D_x$ and $D_\theta$, making PQN an attractive choice  for models with $D_x\gg D_\theta$ like those in Ex.~\ref{ex:hier} and Sec.~\ref{sec:bnn}; see also the Bayesian GANs in~\cite{Saatci2017} and the generalized Bradley-Terry models in~\cite{Caron2012} for more examples.
\subsection{A differential geometry perspective on Newton's method for minimizing functions on Euclidean spaces}\label{app:basicnewton}
Throughout this section and Apps.~\ref{sec:vfetaylor}, \ref{app:newtonflow}, we assume that the reader is acquainted with the contents of App.~\ref{app:crash}. To motivate the flow~(\ref{eq:newtonflowapprox},\ref{eq:IW2Gradflow2}) we discretized to obtain PQN, recall that Newton's method for minimizing a (say, twice-differentiable and strictly convex) function $f:\rn\to\r$,
\[x_{k+1}=x_k-h[\nabla_x^2f(x_k)]^{-1}\nabla_xf(x_k)\quad\forall k=1,2,\dots,\]
is the Euler discretization of the \emph{Newton flow}:
\begin{equation}\label{eq:basicnewtonflow}\dot{x}_t=-[\nabla_x^2f(x_t)]^{-1}\nabla_xf(x_t)\quad\forall t\geq0,\end{equation}
At each point in time $t$, the flow follows the \emph{Newton direction} $v_N(x):=-[\nabla_x^2f(x)]^{-1}\nabla_xf(x)$ at $x_t$ (i.e.\ with $x=x_t$). The appropriate analogue of~\eqref{eq:gradref} shows that $v_N$ is precisely $f$'s gradient $\nabla^{\rm{g}^N} f$ w.r.t.\ the Riemmanian metric $\rm{g}^N$ associated with the tensor $(\nabla_x^2f(x))_{x\in\rn}$. This is an appealing choice because the geometry induced by $\rm{g}^N$ on $\rn$ makes $f$ isotropic, at least to second order:
\begin{align}f(x+tv)&= f(x) +t\iprod{\nabla_xf(x)}{v}+\frac{t^2}{2}\iprod{\nabla_x^2f(x) v}{v}+o(t^2)\label{eq:ftaylor}\\
&=f(x) +t\rm{g}^N(\nabla^{\rm{g}^N}f(x),v)+\frac{t^2}{2}\rm{g}^N_x(v,v)+o(t^2),\nonumber
\end{align}
by Taylor's Theorem. In other words, by replacing $\nabla_x$ with $\nabla^{\rm{g}^N}$ we mitigate bad conditioning in $f$ which, for the reasons discussed in \citet[Secs.~9.4.4,~9.5.1]{Boyd2004} and illustrated in \citet[Figs.~9.14,~9.15]{Boyd2004}, generally makes $v_N(x)$ a much better update direction than the Euclidean gradient $\nabla_x f(x)$. In what follows, we derive the analogue of the Newton direction for the free energy $F$. Doing so requires identifying an appropriate notion for $F$'s  Hessian, which we achieve using an expansion of the form in~\eqref{eq:ftaylor}.
\subsection{A second order Taylor expansion for $F$}\label{sec:vfetaylor}
By definition,
\begin{align*}F(\theta+t\tau,q+tm)=&\int \log(q(x)+tm(x))(q(x)+tm(x))dx\\
&-\int \ell(\theta+t\tau,x) (q(x)+tm(x))dx.\end{align*}
But $\log(z+t)(z+t)=\log(z)z+[\log(z)+1]t+t^2/(2z)+o(t^2)$ and, so,
\begin{align}
\int\log(q(x)+tm(x))(q(x)+tm(x))dx =&\int\log(q(x))q(x)dx+t\int\log(q(x))m(x)dx\label{eq:fndsa87fbwaytfbwatfwa}\\
&+\frac{t^2}{2}\int\left(\frac{m(x)}{q(x)}\right)^2q(x)dx+o(t^2).\nonumber
\end{align}
(Here, we have used that $\int m(x)dx=0$ because $m$ belongs to $\cal{T}\cal{P}(\cal{X})$. Rigorously arguing the above requires considerations similar to those in Footnote~\ref{foot:ot}.) Similarly,
\begin{align*}
\int \ell(\theta+t\tau,x)(q(x)+tm(x))=& \int\log(p_{\theta}(x,y))q(x)dx+t\int \iprod{\nabla_\theta \ell(\theta,x)}{\tau} q(x)dx\\
&+t\int \log(p_{\theta}(x,y))m(x)dx+\frac{t^2}{2}\int \iprod{\tau}{\nabla_\theta^2\ell(\theta,x)\tau} q(x)dx\\
&+t^2 \int \iprod{\nabla_\theta \ell(\theta,x)}{\tau} m(x)dx+o(t^2).
\end{align*}
Putting the above together with~\eqref{eq:fndsa87fbwaytfbwatfwa} and applying Lem.~\ref{lem:Ffirstvar}, we obtain that
\begin{align}\label{eq:FTaylor}
F(\theta+t\tau,q+tm)=F(\theta,q) +t\iprod{(\tau,m)}{\delta F(\theta,q)}+\frac{t^2}{2}\iprod{(\tau,m)}{\cal{H}_F(\theta,q)(\tau,m)}+o(t^2).
\end{align}
where $\cal{H}_F(\theta,q)$ denotes the linear map from $\cal{T}\cal{M}$ to $\cal{T}^*\cal{M}$ defined by
\begin{align}\label{eq:FHessop}\cal{H}_F(\theta,q)(\tau,m)=\left(-\left[\int\nabla_\theta^2\ell(\theta,x) q(x)dx\right]\tau-\int \nabla_\theta \ell(\theta,x)m(x)dx,\frac{m}{q}-\iprod{\nabla_\theta \ell(\theta,\cdot)}{\tau}\right).\end{align}
A comparison of (\ref{eq:ftaylor},\ref{eq:FTaylor}) seems to imply that $\cal{H}_F(\theta,q)$ might be a sensible analogue for $F$'s Hessian. Alternatively, we may view $\cal{H}_F(\theta,q)$  as  the `matrix'
\begin{align}\label{eq:FHessmatr}
\cal{H}_F(\theta,q):=\begin{bmatrix}
-\int\nabla_\theta^2\ell(\theta,x)q(x)dx&-\nabla_\theta \ell(\theta,\cdot)\\-\nabla_\theta \ell(\theta,\cdot)&q^{-1}\end{bmatrix}.%\quad\forall \theta\in\Theta,\enskip q\in\cal{P}(\cal{X}),\enskip \tau \in\cal{T}_\theta\Theta,\enskip m\in\cal{T}_q\cal{D}.
\end{align}

\subsection{The Newton direction and flow, and tractable approximations thereof} \label{app:newtonflow}
Suppose that $F$'s Hessian operator, $\cal{H}_F$ in~\eqref{eq:FHessop}, is invertible everywhere on $\cal{M}$. Similarly as with $f$ in App.~\ref{app:basicnewton}, we set $F$'s Newton direction at $(\theta,q)$ to be
\begin{equation}\label{eq:Fnewtondir}(\tau_N,m_N)(\theta,q):=-[\cal{H}_F(\theta,q)]^{-1}\delta F(\theta,q),\end{equation}
where $\delta F$ denotes $F$'s first variation in Lem.~\ref{lem:Ffirstvar}. Alternatively, assuming further that $\cal{H}_F$ is positive definite everywhere, we can view $(\tau_N,m_N)$ as $F$'s negative gradient $\nabla^{g^N}F$ with respect to the metric,
\[g^N_{(\theta,q)}((\tau,m),(\tau,m)):=\iprod{(\tau,m)}{\cal{H}_F(\theta,q)(\tau,m)},\]
which makes $F$ isotropic, at least to second order: by~\eqref{eq:FTaylor},
\[F(\theta+t\tau,q+tm)=F(\theta,q) +tg^N_{(\theta,q)}((\tau,m),\nabla^{g^N} F(\theta,q))+\frac{t^2}{2}g^N_{(\theta,q)}((\tau,m),(\tau,m))+o(t^2).\]
Unfortunately, we know of no closed-form expressions for $(\tau_N,m_N)$ or computationally tractable approximations to the corresponding flow. However, it is straightforward to find approximations to $\cal{H}_F$ that have both:
\\\\
\noindent\textbf{Block diagonal approximations $\cal{H}_F$ and quasi-Newton directions.} Consider the block-diagonal approximation to $\cal{H}_F$ obtained by zeroing the off-diagonal blocks in~\eqref{eq:FHessmatr}:
\begin{equation}\label{eq:mfeua9fmuea8nfdsaa}\cal{H}_{F}(\theta,q)\approx\rm{diag}\left(-\int\nabla_\theta^2\ell(\theta,x)q(x)dx,q^{-1}\right)=:\rm{diag}(\rm{G}_{(\theta,q)},\sf{G}_q^{FR}).\end{equation}
In other words, $\sf{G}_q^{FR}$ is the Fisher-Rao tensor on $\cal{P}(\cal{X})$ (cf.\ App.~\ref{app:metrics}), while  $\rm{G}_{(\theta,q)}$ is the tensor obtained by integrating the negative log-likelihood's $\theta$-Hessian w.r.t.\ $q$. Using~\eqref{eq:gradrefblock} and Lem.~\ref{lem:Ffirstvar}, we find that the  resulting `quasi-Newton' direction $(\tau_{QN},m_{QN})$ equals
\begin{align*}
\tau_{QN}(\theta,q)&=-\left[\int\nabla_\theta^2\ell(\theta,x)q(x)dx\right]^{-1}\int \nabla_\theta \ell(\theta,x)q(x)dx,\\
(m_{QN}(\theta,q))(x)&=q(x)\left[\log\left(\frac{p_{\theta}(x,y)}{q(x)}\right)-\int\log\left(\frac{p_{\theta}(x,y)}{q(x)}\right)q(x)dx\right];
\end{align*}
and the corresponding gradient flow reads
\[\dot{\theta}_t=\tau_{QN}(\theta_t,q_t)\quad \dot{q}_t=m_{QN}(\theta_t,q_t).\]
While it is likely possible that the above flow can be approximated computationally using techniques along the lines of those in~\citet{Lu2019,Zhang2021}, this would require estimating the log-density $\log(q(x))$ of particle approximations $q$, a complication we opted to avoid in this paper. Instead, we (crudely) further approximate~\eqref{eq:mfeua9fmuea8nfdsaa}  by replacing the Fisher-Rao block $\sf{G}_q^{FR}$ with a Wasserstein-2 block $\sf{G}_q^{W}$ (cf.\ App.~\ref{app:metrics}). The $\tau_{QN}(\theta,q)$-component of the quasi-Newton remains unchanged, the $m_{QN}(\theta,q)$-component is now given by $\nabla_x\cdot[q\nabla_x \log(q/p_{\theta}(\cdot,y))]$, and we obtain the flow in~(\ref{eq:newtonflowapprox},\ref{eq:IW2Gradflow2}).

\section{PARTICLE MARGINAL GRADIENT DESCENT (PMGD)}\label{app:pmgd}
For a surprising number of models in the literature, the (M) step is tractable. In particular:
\begin{assumption}\label{ass:Mstep}For each \(q\) in \(\cal{P}(\cal{X})\),  \(\theta\mapsto F(\theta,q)\) has a unique stationary point $\theta_*(q)$.\end{assumption}
Moreover, we are able to compute this point $\theta_*(x^{1:N}):=\theta_*(q)$ whenever $q=N^{-1}\sum_{n=1}^N\delta_{x^n}$ for $x^{1:N}=(x^1,\dots,x^N)$ in $\cal{X}^N$. In these cases, we can run PMGD (Alg.~\ref{alg:pmgd}) instead of PGD (Alg.~\ref{alg:pgd}).

\begin{algorithm}[h]
\begin{algorithmic}[1]
\STATE{\textbf{Inputs:} step size $h$, step number $K$, particle number $N$, and initial particles $X^1_0,\dots,X_0^N$ and parameters $\theta_0$. }
\FOR{$k=0,\dots, K-1$}
\STATE{Update the particles: for all $n=1,\dots,N$,\begin{align}\label{eq:W2gradalg}
X_{k+1}^n&=X_k^n+h\nabla_x \ell(\theta_*(X_k^{1:N}),X_k^n)+\sqrt{2h}W_k^n
\end{align}
with $W_k^1,\dots,W_k^N$ denoting i.i.d.\ $\cal{N}(0,I_{D_x})$ R.V.s.}
\ENDFOR
\RETURN{$(\theta_k:=\theta_*(X_k^{1:N}),q_k:=N^{-1}\sum_{n=1}^N\delta_{X_k^n})_{k=0}^K$.}
\end{algorithmic}
 \caption{Particle Marginal Gradient Descent (PMGD).}
 \label{alg:pmgd}
\end{algorithm}

PMGD's update equation~\eqref{eq:W2gradalg} approximates the  Wasserstein-2 gradient flow (cf.\ App.~\ref{app:marflow} below) of the `marginal objective' $F_*(q):=F(\theta_*(q),q)$: 
\begin{align}\label{eq:Fstargrad}
\dot{q}_t=-\nabla F_*(q_t),\quad\text{where}\quad\nabla F_*(q)= \nabla_x\cdot\left[ q\nabla_x \log\left(\frac{p_{\theta_*(q)}(\cdot,y)}{q}\right)\right].
\end{align}
In particular, \eqref{eq:Fstargrad} is satisfied by the law of the following McKean-Vlasov SDE:
\begin{align}\label{eq:W2gradSDE}
dX_t=\nabla_x \ell(\theta_*(q_t),X_t)dt+\sqrt{2}dW_t,
\end{align}
where $q_t$ denotes $X_t$'s law and $(W_t)_{t\geq0}$ a standard Brownian motion. We obtain~\eqref{eq:W2gradalg} by following the same steps as in  Sec.~\ref{sec:pgd}, only with (\ref{eq:W2gradSDE}) substituting (\ref{eq:IW2gradSDE1},\ref{eq:IW2gradSDE2}) and the approximations in~\eqref{eq:finsampleapprox} replaced by
\[q_t\approx\frac{1}{N}\sum_{n=1}^N\delta_{X_t^n}\quad\Rightarrow\quad \theta_*(q_t)\approx \theta_*\left(\frac{1}{N}\sum_{n=1}^N\delta_{X_t^n}\right).\]

Thrm.~\ref{thrm:statpoints} is easily adapted to this setting:
\begin{theorem}\label{thrm:statpointsmar}$\theta=\theta_*(q)$ and $\nabla F_*(q)=0$ if and only if $\nabla_\theta p_{\theta}(y)=0$ and $q=p_{\theta}(\cdot|y)$.
\end{theorem}
\begin{proof}Given Thrm.~\ref{thrm:statpoints}, we need only show that $\nabla F(\theta,q)=0$ if and only if $\theta=\theta_*(q)$ and $\nabla F_*(q)=0$. However, Assumpt.~\ref{ass:Mstep} implies that $\nabla_\theta F(\theta,q)=0$ if and only if $\theta=\theta_*(q)$. The result then follows because~\eqref{eq:Fstargrad} implies that $\nabla F_*(q)=0$ if and only if $q=p_{\theta_*(q)}(\cdot|y)$.
\end{proof}
Exploiting the availability of $\theta_*(q)$ seems to improve the convergence. For example, see Fig.~\ref{fig:hier}b,c (in fact, for this simple model, it is straightforward to find theoretical evidence supporting this, cf.\ App.~\ref{app:hierconvrates}). PMGD's complexity is
$$\cal{O}(K[N[\text{eval. cost of }\nabla_x\ell]+[\text{eval. cost of }\theta_*]]).$$
Lastly, we point out that in cases where $\theta_*(x^{1:N})$ is not analytically tractable, but $D_\theta$ is small (at least in comparison to $D_x$), we can instead approximately compute $\theta_*(X_{k}^{1:N})$ using an appropriate optimization routine (warm-starting $\theta_*(X_{k}^{1:N})$'s computation using  $\theta_*(X_{k-1}^{1:N})$).

%\section{THE THREE FLOWS}\label{app:flows}
%

%
%
\subsection{The marginal objective's gradient}\label{app:marflow}

Here, we use the Wasserstein-2 geometry on $\cal{P}(\cal{X})$: that induced by the Wasserstein-2 metric $\sf{g}^W$ with tensor $(\sf{G}^W_q)_{q\in\cal{P}(\cal{X})}$, cf.\ App.~\ref{app:metrics}. As we will now show, the marginal objective $F_*$'s gradient $\nabla F_*(q_t)$ w.r.t.\ to this metric is given by (\ref{eq:Fstargrad})'s RHS.. Given~\eqref{eq:W2inv}, substituting $\cal{P}(\cal{X})$  for $\cal{M}$ in~\eqref{eq:gradref}, we find that
\[\nabla F_*(q)= -\nabla_x \cdot [q \nabla_x \delta F_*(q)],\]
where $\delta F_*$ denotes $F_*$'s first variation (defined analogously to~\eqref{eq:ney8agneywa}). Hence, we need only show that $\delta F_* = \log(q/p_{\theta_*(q)}(\cdot,y))$ or, equivalently, that
\begin{equation}\label{eq:dfs}
F_*(q+tm)=F_*(q)+t\iprod{\log\left(\frac{q}{p_{\theta_*(q)}(\cdot,y)}\right)}{m}+o(t).
\end{equation}
To argue~\eqref{eq:dfs}, we assume that   $\theta_*:\cal{P}(\cal{X})\to\r$ defines a differentiable functional: for each $q$ in $\cal{P}(\cal{X})$ there exists a linear map $D_q\theta_*$  from $\cal{T}\cal{P}(\cal{X})$ to $\cal{T}\Theta$ satisfying
\[(D_q\theta_*) m=\lim_{t\to0}\frac{\theta_*(q+tm)-\theta_*(q)}{t}\quad\forall m\in\cal{T}_q\cal{P}(\cal{X}).\]
Because, with $\norm{\tau}$ denoting the Euclidean norm of $\tau:=\theta_*(q+tm)-\theta_*(q)$,
\[\ell(\theta_*(q+tm),x)=\ell(\theta_*(q),x)+t\iprod{\tau}{\nabla_\theta \ell(\theta_*(q),x)}+o(\norm{\tau}),\]
it follows from $\theta_*$'s differentiability that
\begin{align*}
\ell(\theta_*(q+tm),x)=\ell(\theta_*(q),x)+t\iprod{(D_q\theta_*) m}{\nabla_\theta \ell(\theta_*(q),x)}+o(t).
\end{align*}
For this reason,
\begin{align}
&\int \ell(\theta_*(q+tm),x)(q(x)+tm(x))dx\nonumber\\
&=\int\left[\ell(\theta_*(q),x)+t\iprod{(D_q\theta_*) m}{\nabla_\theta \ell(\theta_*(q),x)}+o(t)\right](q(x)+tm(x))dx\nonumber\\
&=\int \ell(\theta_*(q),x)q(x)dx+t\int \ell(\theta_*(q),x)m(x)dx\nonumber\\
&\quad +t\iprod{(D_q\theta_*) m}{\int\nabla_\theta \ell(\theta_*(q),x)q(x)dx}+o(t).\label{eq:nfe8yanf8eabnfyabnfea}
\end{align}
(Rigorously arguing the above requires considerations similar to those in Footnote~\ref{foot:ot}.) But, by definition, $\theta_*(q)$ minimizes $\theta\mapsto F(\theta,q)$, and we have that
\[\int\nabla_\theta \ell(\theta_*(q),x)q(x)dx=\nabla_\theta F(\theta_*(q),q)=0.\]
Given that $F_*(q)=F(\theta_*(q),q)$, combining the above with (\ref{eq:fndsa87fbwaytfbwatfwa},\ref{eq:nfe8yanf8eabnfyabnfea}) then yields~\eqref{eq:dfs}.
%

%\section{PARTICLE APPROXIMATIONS TO THE QUASI-NEWTON AND MARGINAL GRADIENT FLOWS}\label{app:parapp}
%

\section{EXPERIMENTAL DETAILS AND FURTHER NUMERICAL RESULTS}\label{app:examples}

We implement the methods using Python 3, JAX~\citep{Jax2018}, and PyTorch~\citep{Paszke2019}, and we carry out all experiments using a Google Colab Pro subscription.

\subsection{Toy hierarchical model}\label{app:hier}
\paragraph{Synthetic data.} We generate the data $y$ synthetically by sampling $p_\theta(x,y)$ in~Ex.~\ref{ex:hier} with $\theta$ set to $1$. 

\paragraph{The  marginal likelihood's global maximum and the corresponding posterior.}To obtain closed-form expressions for these, we rewrite the model density, $p_\theta(x,y)$ in Ex.~\ref{ex:hier}, in matrix-vector notation:
\begin{equation}\label{eq:hiervm}p_\theta(x,y)=\cal{N}(y;x,I_{D_x})\cal{N}(x;\theta\bm{1}_{D_x},I_{D_x})\quad\forall \theta\in\r,\enskip x,y\in\r^{D_x}.\end{equation}
Combining the expressions in  \citet[p.\ 92]{Bishop2006} with the Sherman-Morrison formula, we then find that
\begin{align}\label{eq:ndw7a8ndwa87bndwa}
p_\theta(y)=\cal{N}(y;\theta\bm{1}_{D_x},2I_{D_x}),\quad p_\theta(x|y)=\cal{N}\left(x;\frac{y+\theta\bm{1}_{D_x}}{2},\frac{1}{2}I_{D_x}\right).%\cal{N}(x;[y+\theta\bm{1}_{D_x}]/2,I/2).
\end{align}
Because 
\[\nabla_\theta \log(p_\theta(y))=\bm{1}_{D_x}^T(y-\bm{1}_{D_x}\theta)=\bm{1}_{D_x}^Ty-D_x\theta,\]
it follows the data's empirical mean is the marginal likelihood's  unique maximizer $\theta_*$, and plugging it into~\eqref{eq:ndw7a8ndwa87bndwa} we obtain an expression for the corresponding posterior:
\begin{align}\label{eq:hieroptimum}
\theta_*&=\frac{\bm{1}^T_{D_x}y}{D_x},\quad p_{\theta_*}(x|y)=\cal{N}\left(x;\frac{1}{2}\left[y+\frac{\bm{1}^T_{D_x}y}{D_x}\right],\frac{1}{2}I_{D_x}\right).
\end{align}

\paragraph{Implementation details for PGD, PQN, and PMGD.} Taking derivatives of~\eqref{eq:hiervm}'s log, we find that
\begin{equation}\label{eq:hiergradients}
\nabla_\theta \ell(\theta,x)=\bm{1}^T_{D_x}(x-\theta\bm{1}_{D_x}),\quad\nabla_\theta^2\ell\equiv -D_x,\quad
\nabla_x \ell(\theta,x)=y-x-(x-\theta\bm{1}_{D_x}).
\end{equation}
Given that 
\begin{equation}\label{eq:hierman}\nabla_\theta F(\theta,q)=-\int \nabla_\theta \ell(\theta,x)q(x)dx=-\bm{1}^T_{D_x}\left[\int xq(x)dx-\theta\bm{1}_{D_x}\right],\end{equation}
Assumpt.~\ref{ass:Mstep} is satisfied with
\begin{equation}\label{eq:hierthetaopt}
\theta_*(q)=\frac{\bm{1}^T_{D_x}}{D_x}\int xq(x)dx\enskip\forall q\in\cal{P}(\cal{X})\quad\Rightarrow\quad\theta_*(x^{1:N})=\frac{\bm{1}_{ND_x}^Tx^{1:N}}{ND_x}\enskip\forall x^{1:N}\in\cal{X}^N.
\end{equation}
Given (\ref{eq:hiergradients},\ref{eq:hierthetaopt}), PGD's~(Alg.~\ref{alg:pgd}) updates then read~(\ref{eq:hiergrad1},\ref{eq:hiergrad2}), PQN's~(Alg.~\ref{alg:pqn}) read~(\ref{eq:hiergrad2},\ref{eq:hiernewton}), and PMGD's~(Alg.~\ref{alg:pmgd}) reads~\eqref{eq:hiermarginal}:
\begin{align}
\theta_{k+1} &= \theta_{k} + h D_x\left[\theta_*(X^{1:N}_k)-\theta_k\right],\label{eq:hiergrad1}\\
X_{k+1}^{1:N}&=X_k^{1:N}+h[y^N+\theta_k\bm{1}_{ND_x}-2X_k^{1:N}]+\sqrt{2h}W_k^{1:N},\label{eq:hiergrad2}\\
\theta_{k+1} &= \theta_{k} + h \left[\theta_*(X^{1:N}_k)-\theta_k\right],\label{eq:hiernewton}\\
X_{k+1}^{1:N}&=X_k^{1:N}+h\left[y^N+\theta_*(X_k^{1:N})\bm{1}_{ND_x}-2X_k^{1:N}\right]+\sqrt{2h}W_k^{1:N},\label{eq:hiermarginal}
\end{align}
where $y^N$ stacks $N$ copies of $y$, $X^{1:N}_k:=(X^1_k,\dots,X^N_k)$, and similarly for $X_{k+1}^{1:N}$ and $W_k^{1:N}$.  Because the $\theta$-gradient in~\eqref{eq:hiergradients} is a sum of $D_x$ terms, $\Lambda$ in~\eqref{eq:IW2gradalg2} simply equals $D_x^{-1}$ and the tweaked version  PGD parameter update~(\ref{eq:IW2gradalg2}) coincides with PQN's~(\ref{eq:hiernewton}).

\paragraph{Implementation details for EM.} Given (\ref{eq:ndw7a8ndwa87bndwa},\ref{eq:hierthetaopt}), the EM steps read
\[\textrm{(E)}\quad q_{k}:=\cal{N}\left(\frac{y+\theta_{k}\bm{1}_{D_x}}{2},\frac{1}{2}I_{D_x}\right),\qquad \textrm{(M)}\quad \theta_{k+1}:=\frac{1}{2}\left(\frac{\bm{1}_{D_x}^Ty}{D_x}+\theta_{k}\right).\]
\subsection{Bayesian logistic regression}\label{app:blr}

\paragraph{Dataset.}We use the Wisconsin Breast Cancer dataset $\cal{Y}$~\citep{Wolberg1990}, created by Dr.\ William H.\ Wolberg at the University of Wisconsin Hospitals, and freely available at
\begin{center}\href{https://archive.ics.uci.edu/ml/datasets/breast+cancer+wisconsin+(original)}{https://archive.ics.uci.edu/ml/datasets/breast+cancer+wisconsin+(original)}.\end{center}
It contains $683$ datapoints\footnote{After removal of the $16$ datapoints with missing features.} each with nine features $f\in\r^9$ extracted from  a digitized image of a fine needle aspirate of a breast mass and an accompanying label $l$ indicating whether the mass is benign ($l=0$) or malign ($l=1$). We normalize the features so that each has mean zero and unit standard deviation across the dataset. We split the dataset into $80/20$ training and testing sets, $\cal{Y}_{\text{train}}$ and $\cal{Y}_{\text{test}}$.

\paragraph{Model.}Emulating~\citet[Sec.~4.1]{Debortoli2021}, we employ standard Bayesian logistic regression with Gaussian priors. That is, we assume that the datapoints' labels are conditionally independent given the features $f$ and regression weights $x\in\r^{D_x:=9}$, each label with Bernoulli law and mean $s(f^Tx)$, where $s(z):=e^z/(1+e^z)$ denotes the standard logistic function; and we assign the prior $\cal{N}(\theta \bm{1}_{D_x},5 I_{D_x})$ to the weights $x$, where $\theta$ denotes the (scalar) parameter to be estimated. The model's density is given by:
\[p_\theta(x,\cal{Y}_{\text{train}})=\cal{N}(x;\theta \bm{1}_{D_x},5 I_{D_x})\prod_{(f,l)\in\cal{Y}_{\text{train}}}s(f^Tx)^l[1-s(f^Tx)]^{1-l};\]
and it follows that
\begin{equation}\label{eq:blrdens}\ell(\theta,x)=\sum_{(f,l)\in\cal{Y}_{\text{train}}}[lf^Tx-\log (1+e^{f^Tx})]-\frac{\norm{x-\bm{1}_{D_x}\theta}^2}{5}.\end{equation}
The marginal likelihood has a unique maximizer:
\begin{proposition}\label{prop:blr}If $f^T\bm{1}_{D_x}\neq 0$ for at least one $(l,f)$ in $\cal{Y}_{\text{train}}$, then $\theta\mapsto p_\theta(\cal{Y}_{\text{train}})=\int p_\theta(x,\cal{Y}_{\text{train}})dx$ has a single maximizer $\theta_*$ and no other stationary points.
\end{proposition}
\begin{proof}Given Thrm.~\ref{thrm:strilogconc} in App.~\ref{app:conv}, we need only argue that $\ell$ is strictly concave. Taking gradients of~\eqref{eq:blrdens}, we find that
\[\nabla^2\ell(\theta,x)=\frac{1}{5}\begin{bmatrix}-D_x&\bm{1}_{D_x}^T\\\bm{1}_{D_x}&-I_{D_x}\end{bmatrix}-\sum_{(f,l)\in\cal{Y}_{\text{train}}}s(f^Tx)[1-s(f^Tx)]f\otimes f.\]
The leftmost matrix has a single nonnegative eigenvalue. It equals zero, its geometric multiplicity is one, and its corresponding eigenvector is the vector of ones $\bm{1}_{D_x+1}$. However,
\begin{align*}v^T\left[\sum_{(f,l)\in\cal{Y}_{\text{train}}}s(f^Tx)[1-s(f^Tx)]f\otimes f\right] v=\sum_{(f,l)\in\cal{Y}_{\text{train}}}s(f^Tx)[1-s(f^Tx)](f^Tv)^2\geq 0\end{align*} 
for all $v$ in $\r^{D_x}$. By assumption, $f^T\bm{1}_{D_x}\neq 0$ for at least one feature vector $f$ in the test set, and the above inequality is strict if $v\neq\bm{1}_{D_x}$. It then follows that
\[z^T\nabla^2\ell(\theta,x)z<0\quad \forall z\in\r^{D_x+1},\enskip \theta\in \Theta,\enskip x\in\cal{X};\]
or, in other words, that $\ell$ is strictly concave.
\end{proof}
\paragraph{Implementation details.}Taking gradients of~\eqref{eq:blrdens}, we obtain
\begin{align*}
\nabla_\theta \ell(\theta,x)=\frac{\bm{1}^T_{D_x}x-D_x\theta}{5},\quad \nabla_\theta^2\ell\equiv- \frac{D_x}{5},\quad
\nabla_x \ell(\theta,x)=\frac{\theta\bm{1}_{D_x}-x}{5}+\sum_{(f,l)\in\cal{Y}_{\text{train}}}[l-s(f^Tx)]f,
\end{align*}
The same manipulations as in~\eqref{eq:hierman} show that Assumpt.~\ref{ass:Mstep} is satisfied with $\theta_*(q)$ and $\theta_*(x^{1:N})$ as in~\eqref{eq:hierthetaopt}. Hence, PGD's~(Alg.~\ref{alg:pgd}) updates read~(\ref{eq:blrgrad1},\ref{eq:blrgrad2}), PQN's~(Alg.~\ref{alg:pqn}) read~(\ref{eq:blrgrad2},\ref{eq:blrnewton}), and PMGD's~(Alg.~\ref{alg:pmgd}) reads~\eqref{eq:blrmarginal}:
\begin{align}
\theta_{k+1} &= \theta_{k} + h (D_x/5)[\theta_*(X^{1:N}_k)-\theta_k],\label{eq:blrgrad1}\\
X_{k+1}^{n}&=X_k^{n}+h\left(\frac{\theta_k\bm{1}_{D_x}-X_k^n}{5}+\sum_{(f,l)\in\cal{Y}_{\text{train}}}[l-s(f^TX_k^n)]f\right)+\sqrt{2h}W_k^{n}\enskip\forall n\in[N],\label{eq:blrgrad2}\\
\theta_{k+1} &= \theta_{k} + h [\theta_*(X^{1:N}_k)-\theta_k],\label{eq:blrnewton}\\
X_{k+1}^{n}&=X_k^{n}+h\left(\frac{\theta_*(X_k^{1:N})\bm{1}_{D_x}-X_k^n}{5}+\sum_{(f,l)\in\cal{Y}_{\text{train}}}[l-s(f^TX_k^n)]f\right)+\sqrt{2h}W_k^{n}\enskip\forall n\in[N].\label{eq:blrmarginal}
\end{align}
SOUL's (Sec.~\ref{sec:blr}) updates read \eqref{eq:blrgrad1}, $X_{k+1}^{0}=X_{k}^{N}$, and
\[X_{k+1}^{n+1}=X_{k+1}^n+h\left(\frac{\theta_k\bm{1}_{D_x}-X_{k+1}^n}{5}+\sum_{(f,l)\in\cal{Y}_{\text{train}}}[l-s(f^TX_{k+1}^n)]f\right)+\sqrt{2h}W_{k+1}^n\enskip\forall n\in[N-1].\]

For MFG VI, we use a product-form Gaussian $q_\phi:= \mathcal{N}(\mu, \textrm{diag}(\sigma^2(s))$ as the variational approximation, where $\textrm{diag}(\sigma^2(s))$ denotes a diagonal matrix with $\sigma^2(s):=\mathrm{Softplus}(s)$  on its diagonal and $\phi:=(\mu, s)$ in $\mathbb{R}^{2D_x}$ denote the variational parameters. The variational free energy then reads
\begin{align*}
F(\theta, \phi) = - \frac{1}{2} \left [\sum_{i=1}^{D_x} \log (\sigma_{i}^2(s)) + D_x [\log (2\pi) + 1] \right ] - \int \ell(\theta,x)q_\phi(x)dx.
\end{align*}
Using  the reparametrization trick \citep{Kingma2014}, we find that
\begin{align*}
F(\theta, \phi) \approx - \frac{1}{2} \left [\sum_{i=1}^{D_x} \log (\sigma_{i}^2(s)) + D_x [\log (2\pi) + 1] \right ] - \frac{1}{N}\sum_{n=1}^N \ell(\theta,\mu+\textrm{diag}(\sigma(s))\epsilon^n),
\end{align*}
where $\epsilon^1,\dots,\epsilon^N$ denote i.i.d.\ samples drawn from $\cal{N}(0,I_{D_x})$ (we set $N$ to $100$: the maximum number of particles we use for PGD, PQN, PMGD, and SOUL). We then minimize the RHS over $(\theta,\phi)$ using gradient descent.

\paragraph{Predictive performance metrics.}Given a new feature vector $\hat{f}$, we would ideally predict its label $\hat{l}$ using the posterior predictive distribution associated with the  marginal likelihood's maximizer  $\theta_*$. In other words, using
\[
p_{\theta_*}(\hat{l}|\hat{f},\cal{Y}_{\text{train}})=\int p(\hat{l}|\hat{f},x)p_{\theta_*}(x|\cal{Y}_{\text{train}})dx=\int s(\hat{f}^Tx)^{\hat{l}}[1-s(\hat{f}^Tx)]^{1-\hat{l}}p_{\theta_*}(x|\cal{Y}_{\text{train}})dx.\]
However, $p_{\theta_*}(x|\cal{Y}_{\text{train}})$ is unknown. So, we  replace it with a particle approximation $q=M^{-1}\sum_{m=1}^M\delta_{Z^m}$ thereof obtained using PGD, PQN, PMGD,  SOUL, or MFG VI\footnote{For MFG VI, we set $M$ to $20100$ and draw $Z^1,\dots,Z^M$ independently from $q_{\phi_*}$, where $\phi_*$ denotes the optimized variational parameters.}:
\begin{align}p_{\theta_*}(\hat{l}|\hat{f},\cal{Y}_{\text{train}})&\approx \int s(\hat{f}^Tx)^{\hat{l}}[1-s(\hat{f}^Tx)]^{1-\hat{l}}q(dx)\nonumber\\
&=\frac{1}{M}\sum_{m=1}^M s(\hat{f}^TZ^m)^{\hat{l}}[1-s(\hat{f}^TZ^m)]^{1-\hat{l}}=:g(\hat{l}|\hat{f})\label{eq:dnw8a7dbwyahdnaudaad}.\end{align}
We use two metrics to evaluate the approximation's predictive power. First, the average classification error over the test set $\cal{Y}_{\text{test}}$, i.e.\ the fraction of mislabelled test points were we to assign to each of them the label maximizing \eqref{eq:dnw8a7dbwyahdnaudaad}'s RHS:
\begin{align}
\text{Error}:=\frac{1}{\mmag{\cal{Y}_{\text{test}}}}\sum_{(f,l)\in\cal{Y}_{\text{test}}}\mmag{l-\hat{l}(f)},\quad\text{where}\quad \hat{l}(f):=\argmax_{\hat{l}\in\{0,1\}}g(\hat{l}|f).\label{eq:terror}
\end{align}
The second metric is the so-called log pointwise predictive density (LPPD, e.g.~\cite{Vehtari2017}):
\begin{align}
\text{LPPD}:=\frac{1}{\mmag{\cal{Y}_{\text{test}}}}\sum_{(f,l)\in\cal{Y}_{\text{test}}}\log(g(l|f)).\label{eq:lppd}
\end{align}
Interest in this metric stems from the assumption that the data is drawn independently from a `data-generating process' $p(dl,df)$, in which case, for large test sets,
\begin{align*}
\text{LPPD}&\approx \int\log(g(l|f))p(dl,df)\\
&=\int\left[\int\log\left(\frac{g(l|f))}{p(l|f)}\right)p(dl|f)\right]p(df)+\int \log(p(l|f))p(dl,df)\\
&=-\int KL(g(\cdot|f))||p(\cdot|f))p(df)+\int \log(p(l|f))p(dl,df).
\end{align*}
In other words, the larger LPPD is, the smaller we can expect the mean KL divergence between our classifier $g(l|f)$ and the optimal classifier $p(l|f)$.

\paragraph{Numerical results.}To investigate the algorithms' performances, we ran them $100$ times, each time using a different random $80/20$ training/testing split of the data. In all runs we employed a step size of $h=0.01$ (which ensured that no algorithm was on the verge of becoming unstable while simultaneously not being excessively small), $K=400$ steps, and $N=1,10,100$ particles. Tab.~\ref{tab:blrmt} shows the test errors~\eqref{eq:terror} and computation times, and Tab.~\ref{tab:blrapp} the corresponding LPPDs~\eqref{eq:lppd} and stationary empirical variances of the parameter estimates. For the predictive performance metrics, we initialized the estimates and particles at zero (as in~\cite{Debortoli2021}) and used the time-averaged approximations $\bar{q}_{400}$, cf.~\eqref{eq:estimators}, with a burn-in of $k_b=200$. (Warm-starting did not lead to any improvements here.) By $k=200$, the PQN parameter estimates have not yet reached the stationary phase (Fig.~\ref{fig:blr}a). Hence, for the variance estimates, we warm-start the algorithms using a preliminary run of PGD (with $K=400$, $h=0.01$, and a single particle $N=1$) and then compute the estimates using a full $K=400$ run of the corresponding algorithm.

\begin{table*}[b]  \caption{\textbf{Bayesian logistic regression.} Log pointwise predictive densities and stationary variances  achieved using time-averaged posterior approximation $\bar{q}_{400}$, with $N=1,10,100$, and corresponding computation times (averaged over $100$ replicates). See details in the text.}
  \label{tab:blrapp}
  \centering
  \begin{tabular}{lllllll}
    \toprule
    &\multicolumn{2}{c}{$N=1$}&\multicolumn{2}{c}{$N=10$}&\multicolumn{2}{c}{$N=100$}\\
    \midrule
    & LPPD ($\times10^{-2})$  & Var.\ ($\times10^{-4}$) & LPPD ($\times10^{-2})$  & Var.\ ($\times10^{-4}$)& LPPD ($\times10^{-2})$  & Var.\ ($\times10^{-4}$) \\
    \midrule
    PGD          & -9.73 $\pm$ 1.04 & 14.1 $\pm$ 13.6 & -9.40 $\pm$ 0.28 & 1.25 $\pm$ 1.01 & -9.38 $\pm$ 0.08 &  0.13 $\pm$ 0.10 \\
    PQN          & -9.65 $\pm$ 0.87 & 7.33 $\pm$ 6.63 & -9.41 $\pm$ 0.27 & 0.72 $\pm$ 0.73 & -9.41 $\pm$ 0.09 &  0.06 $\pm$ 0.06 \\
    PMGD         & -9.61 $\pm$ 0.86 & 106  $\pm$ 36.7 & -9.48 $\pm$ 0.27 & 10.7 $\pm$ 4.38 & -9.39 $\pm$ 0.07 &  1.03 $\pm$ 0.35 \\
    SOUL         & -9.73 $\pm$ 0.94 & 11.7 $\pm$ 10.7 & -9.41 $\pm$ 0.27 & 2.78 $\pm$ 2.23 & -9.39 $\pm$ 0.09 &  0.28 $\pm$ 0.6 \\
    \bottomrule
  \end{tabular}
\end{table*}
%
%\begin{table}[h!]
%\caption{\textbf{Bayesian logistic regression.} See text in App.~\ref{app:blr} for details.}
%\label{tab:blr}
%  \centering
%  \begin{tabular}{lllll}
%    \toprule
%    Alg.\ & LPPD ($\times10^{-2})$    & Error (\%)  & Variance ($\times10^{-4}$)  & Time (s)\\
%    \midrule
%    \multicolumn{5}{c}{$N=1$}                   \\
%    \midrule
%    PGD & -9.73 $\pm$ 1.04 &  3.58 $\pm$ 0.78 & 14.1 $\pm$ 13.6 & 0.03   $\pm$ 0.01 \\
%    PQN & -9.65 $\pm$ 0.87 &  3.54 $\pm$ 0.77 & 7.33 $\pm$ 6.63 & 0.03 $\pm$ 0.00 \\
%    PMGD & -9.61 $\pm$ 0.86 &  3.56 $\pm$ 0.69 & 106  $\pm$ 36.7 & 0.03 $\pm$ 0.00 \\
%    SOUL & -9.73 $\pm$ 0.94 &  3.53 $\pm$ 0.72 & 11.7 $\pm$ 10.7 & 0.03 $\pm$ 0.00 \\
%    \midrule
%    \multicolumn{5}{c}{$N=10$}                   \\
%    \midrule
%    PGD & -9.40 $\pm$ 0.28 & 3.55  $\pm$ 0.71  & 1.25  $\pm$  1.01 & 0.09 $\pm$ 0.01  \\
%    PQN & -9.41 $\pm$ 0.27 & 3.49  $\pm$ 0.66  & 0.72 $\pm$  0.73 & 0.09  $\pm$ 0.00  \\
%    PGMD & -9.48 $\pm$ 0.27 & 3.65  $\pm$ 0.64  & 10.7 $\pm$  4.38 & 0.09  $\pm$  0.01  \\
%    SOUL & -9.41 $\pm$ 0.27 & 3.60  $\pm$ 0.60  & 2.78 $\pm$  2.23 & 0.25  $\pm$  0.01  \\
%    \midrule
%    \multicolumn{5}{c}{$N=100$}                   \\
%    \midrule
%    PGD & -9.38 $\pm$ 0.08 & 3.46 $\pm$ 0.32 & 0.13 $\pm$ 0.10 & 1.22 $\pm$ 0.34 \\
%    PQN & -9.41 $\pm$ 0.09 & 3.47 $\pm$ 0.33 & 0.06 $\pm$ 0.06 & 1.17  $\pm$ 0.26 \\
%    PMGD & -9.39 $\pm$ 0.07 & 3.44 $\pm$ 0.33 & 1.03 $\pm$ 0.35 & 1.15 $\pm$ 0.18 \\
%    SOUL & -9.39 $\pm$ 0.09 & 3.43 $\pm$ 0.35 & 0.28 $\pm$ 0.61 & 13.4 $\pm$ 0.23 \\
%    \bottomrule
%  \end{tabular}
%\end{table}

\subsection{Bayesian neural network}\label{app:bnn}
\paragraph{Dataset.}We use the MNIST~\citep{Lecun1998} dataset $\cal{Y}$, available under the terms of the Creative Commons Attribution-Share Alike 3.0 license at
\begin{center}\href{http://yann.lecun.com/exdb/mnist/}{http://yann.lecun.com/exdb/mnist/}.\end{center}
%
%(Copyright held by Yann LeCun and Corinna Cortes.) 
It contains $70,000$ $28\times28$ grayscale images $f\in \r^{784}$ of handwritten digits each accompanied its corresponding label $l$. We avoid big data issues by  subsampling $1000$ datapoints with labels $4$ and $9$ just as in~\citet{Yao2022} (except that we pick the labels $4$ and $9$ rather than $1$ and $2$ to make the problem more challenging). We normalize the $784$ features so that each has mean zero and unit standard deviation across the dataset. We split the dataset into $80/20$ training and testing sets, $\cal{Y}_{\text{train}}$ and $\cal{Y}_{\text{test}}$.
\paragraph{Model.}Following~\citet{Yao2022}, we employ a Bayesian two-layer neural network with tanh activation functions, a softmax output layer, and Gaussian priors on the weights (however, we simplify matters by setting all network biases to zero). That is, we assume that the datapoints' labels are conditionally independent given the features $f$ and network weights $x:=(w,v)$ (where $w\in\r^{D_w:=40\times 784=31360}$ and $w^0\in\r^{D_v:=2\times 40=80}$) with law
\begin{equation}\label{eq:bnnclass}p(l|f,x)\propto \exp\left(\sum_{j=1}^{40}v_{lj}\tanh\left(\sum_{i=1}^{784}w_{ji}f_i\right)\right).\end{equation}
Also as in \cite{Yao2022}, we assign we assign the prior $\cal{N}(\bm{0}_{D_w},e^{2\alpha}I_{D_w})$ to the input layer's weights  and $\cal{N}(\bm{0}_{D_v},e^{2\beta}I_{D_v})$ to those of the output layer, where $\bm{0}_d$ denotes the $d$-dimensional vector of zeros. However, rather than assigning a hyperprior to $\alpha, \beta$, we instead learn them from the data (i.e.\ $\theta:=(\alpha, \beta)$). The model's density is given by:
\[p_\theta(x,\cal{Y}_{\text{train}})=\cal{N}(w;\bm{0}_{D_w},e^{2\alpha}I_{D_w})\cal{N}(v; \bm{0}_{D_v},e^{2\beta}I_{D_v})\prod_{(f,l)\in\cal{Y}_{\text{train}}}p(l|f,x).\]
\paragraph{Implementation details.}The necessary $\theta$-gradients and $\theta$-Hessian are straightforward to compute by hand:
\[\nabla_\theta \ell(\theta,x)=\begin{bmatrix}\norm{w}^2e^{-2\alpha}-D_w\\
\norm{v}^2e^{-2\beta}-D_v\end{bmatrix},\quad \nabla_\theta^2\ell(\theta,x)=-\begin{bmatrix}2\norm{w}^2e^{-2\alpha}&0\\
0&2\norm{v}^2e^{-2\beta}\end{bmatrix}.
\]
For the $x$-gradients, we use JAX's grad function (implementing a version of autograd). Given that 
\begin{align*}
\nabla_\theta F(\theta,q)=-\int \nabla_\theta \ell(\theta,x)q(x)dx=-\begin{bmatrix}e^{-2\alpha}\int \norm{w}^2q(w)dw-D_w\\
e^{-2\beta}\int \norm{v}^2q(v)dv-D_v\end{bmatrix},
\end{align*}
where $q(w)$ and $q(v)$ respectively denote $q$'s $w$ and $v$ marginals, Assumpt.~\ref{ass:Mstep} is satisfied with
\begin{align*}
\theta_*(q)&=\begin{bmatrix}\alpha_*(q)\\\beta_*(q)\end{bmatrix}=\begin{bmatrix}\frac{1}{2}\log\left(D_w^{-1}\int \norm{w}^2q(w)dw\right)\\
\frac{1}{2}\log\left(D_v^{-1}\int \norm{v}^2q(v)dv\right)\end{bmatrix}\quad\forall q\in\cal{P}(\cal{X}),\\
\Rightarrow\theta_*(x^{1:N})&=\begin{bmatrix}\alpha_*(x^{1:N})\\\beta_*(x^{1:N})\end{bmatrix}=\begin{bmatrix}\frac{1}{2}\log\left([ND_w]^{-1}\sum_{n=1}^N\norm{w^n}^2\right)\\
\frac{1}{2}\log\left([ND_v]^{-1}\sum_{n=1}^N\norm{v^n}^2\right)\end{bmatrix}\enskip\forall x^{1:N}=(w^{1:N},v^{1:N})\in\cal{X}^N.
\end{align*}
The PGD, PQN, PMGD, and SOUL updates are obtained by plugging the  expressions above into~(\ref{eq:IW2gradalg11},\ref{eq:IW2gradalg12}), (\ref{eq:newtonalg},\ref{eq:IW2gradalg12}), \eqref{eq:W2gradalg}, and (\ref{eq:IW2gradalg11},\ref{eq:SOUL}), respectively. To avoid memory issues, we only store the current particle cloud and use its empirical distribution to approximate the posteriors (rather than a time-averaged version thereof). We initialize the parameter estimates at zero and the weights at samples drawn independently from the priors.

Given the high dimensionality of the latent variables, PGD and SOUL prove less stable than PQN and PMGD (the former lose stability around $h\approx 10^{-4}$ and the latter around $h\approx 1$). Hence, we stabilize PGD and SOUL using the heuristic~\eqref{eq:IW2gradalg2} discussed in Sec.~\ref{sec:pgd}. This simply entails respectively dividing the $\alpha$ and $\beta$ gradients by $D_w$ and $D_v$. We then set $h:=0.1$ which ensures that no algorithm is close to losing stability.

\paragraph{Predictive performance metrics.} We use the test error and log pointwise predictive density defined as in~(\ref{eq:terror},\ref{eq:lppd}), only with the classifier $g(l|f)$ now given by $N^{-1}\sum_{n=1}^Np(l|f,X^n_K)$, where $p(l|f,x)$ is as in~\eqref{eq:bnnclass} and $X^{1:N}_K$ denotes the final particle cloud produced by PGD, PQN, PMGD, or SOUL.

\subsection{Generator network}\label{app:gen}

\paragraph{Datasets.}We use two datasets of images, the MNIST dataset described in App.~\ref{app:bnn}, and the CelebA dataset~\citep{Liu2015}. The latter contains $202,599$ $178\times 218$ color images of celebrity faces. In both cases, we resize the images to be $32\times32$ and normalize  pixel values so that they lie in $[-1,1]$. We also randomly pick $M:=10,000$ (MNIST) or $M:=40,000$ (CelebA) images $y^{1:M}:=(y^m)_{m=1}^M$ for training and reserve the rest for testing.

\paragraph{Model.}
The model assumes that each image $y^m$ is generated independently of all others by:
\begin{compactenum}
\item drawing a latent variable $x^m$ from a zero-mean unit-variance Gaussian distribution $p(x):=\mathcal{N}(x|0,I_{d_x})$ on a $d_x:=64$-dimensional latent space $\mathbb{R}^{d_x}$;
\item mapping $x^m$ to the image space $\mathbb{R}^{d_y}$ (with $d_y=32\times 32$ for MNIST and $d_y=3\times32\times32$ for CelebA) via a generator $f_\theta$: a neural network parameterized by some parameters $\theta$ in $\mathbb{R}^{D_\theta}$;
\item adding zero-mean $0.01^2$-variance Gaussian noise: $y^m=f_\theta(x^m)+\epsilon^m$ where $(\epsilon^m)_{m=1}^M$ is a sequence of i.i.d. R.V.s with law $\mathcal{N}(0,0.01^2 I_{d_y})$.
\end{compactenum}
In full, the model's density is given by
\begin{equation}\label{eq:nlvmiid}
p_\theta (x^{1:M},y^{1:M}) = \prod_{m=1}^M p_\theta(x^{m}, y^{m}),\end{equation}
where
$$
p_\theta(x^m,y^m)= p_\theta(y^m|x^m)p(x^m),\quad\textrm{with}\quad p_\theta(y^m|x^m) := \mathcal{N}(y^m|f_\theta(x^m), 0.01^2 I_{d_y}).
$$
For $f_\theta$ we use a convolutional neural network with an architecture emulating that in~\cite{Nijkamp2020}, see below for details. In total, it has $355,457$ parameters and $64\times M=640,000$ latent variables for MNIST, and $357,507$ parameters and $64\times M=2,560,000$ latent variables for CelebA.

\paragraph{Network architecture.} The network is composed of layers of $4$ basic types:
\begin{compactitem}
\item $l_\theta$: fully-connected linear layers,
\item $c_\theta$: convolutional layers,
\item $c_\theta^T$: transpose convolutional layers,
\item $b_\theta:$ batch normalization layers.
\end{compactitem}
These are interwoven with GELU activation functions. First, the above are assembled to create $2$ further types of layers:
\begin{compactitem}
\item `projection' layers $\pi_\theta:=\textrm{GELU} \circ b_\theta\circ c_\theta^T \circ \textrm{GELU}\circ b_\theta\circ l_\theta$;
\item `deterministic' layers $d_\theta=\textrm{GELU} \circ b_\theta \circ c_\theta \circ \textrm{GELU}\circ b_\theta \circ c_\theta + I$ where $I$ denotes the identity operator (i.e., the layer has a skip connection).
\end{compactitem}
The network itself then consists of a projection layer followed by two deterministic layers, a transpose convolutional layer, and a $\tanh$ activation function:

$$f_\theta =  \tanh\circ c_\theta^T\circ d_\theta \circ d_\theta \circ \pi_\theta.$$

%Please the code in [model.py](https://github.com/ParticleEM/ParEM_neural_latent_variable_model/blob/master/parem/model.py) for more details. Or, alternatively, see \citet[Tab.~4]{Nijkamp2020} for layer-specific input, intermediate, and output dimensions.

\paragraph{Training.}

Training the model entails searching for parameters $\theta_*$ maximizing the likelihood $\theta\mapsto p_\theta(y^{1:M})$ of the training set $y^{1:M}$, at least locally. We do so using $4$ different approaches: PGD (Alg.~\ref{alg:pgd}), alternating back propagation (ABP;~\citet{Han2017}), short-run MCMC (SR;~\citet{Nijkamp2020}), and variational inference (i.e.\ appending to the model an inference network, so turning it into a variational autoencoder, VAE;~\citet{Kingma2014}). In all cases, we use PyTorch to implement the algorithm and compute the necessary gradients. 

\paragraph{Training (PGD).} We use PGD slightly modified to better cope with the high evaluation cost of the log-likelihood's  gradients.  In particular,  we replace $\nabla_{\theta}$ in the parameter update~(\ref{eq:IW2gradalg11}) with an unbiased estimator thereof obtained by subsampling the training set:
\begin{align*}\nabla_{\theta} \ell(\theta,x^{1:M})&=\sum_{m=1}^M \nabla_\theta\log(p_\theta(y^m|x^m))=M\left[\frac{1}{M}\sum_{m=1}^M \nabla_\theta\log(p_\theta(y^m|x^m))\right]\\
&\approx M\left[\frac{1}{M_{\mathcal{B}}}\sum_{m\in\mathcal{B}}\nabla_\theta\log(p_\theta(y^m|x^m))\right]=\frac{M}{M_{\mathcal{B}}}\sum_{m\in\mathcal{B}}\nabla_\theta\log(p_\theta(y^m|x^m)),\end{align*}

where $\mathcal{B}$ denotes a random subset of $[M]:=\{1,\dots, M\}$ and $M_\mathcal{B}$ its cardinality. To mitigate the varying magnitudes among  the entries of
$$\nabla_\theta\log(p_\theta(y^m|x^m))=\frac{1}{0.01^2}[y^m-f_\theta(x^m)]^T\nabla_\theta f_\theta(x^m)$$
and improve the training, we use a modified version of the  heuristic~\eqref{eq:IW2gradalg2} discussed in Sec.~\ref{sec:pgd}. As in the heuristic, we rescale each entry by a scalar, only that, this time, we allow the scalars to vary with the iteration count $k$. We choose these scalars as in RMSprop~\citep{Hinton2012} using the default values of PyTorch $1.12$'s implementation (cf.\ the \href{https://pytorch.org/docs/stable/generated/torch.optim.RMSprop.html}{documentation}).
In full, we update the parameter estimates $\theta_k$ using
\begin{equation}\label{eq:genthetaupdate}
\theta_{k+1} = \theta_k + h\lambda\Lambda_k\frac{M}{NM_{\mathcal{B}}}\left[\sum_{n=1}^N \sum_{m\in\cal{B}_{k}}\nabla_\theta\log p_{\theta_k}
( y^{m}|X^{n,m}_k) \right],
\end{equation}
where $(X^n)_{n=1}^N=((X^{n,m})_{m=1}^M)_{n=1}^N$ denotes the particle cloud at the $k^{th}$ iteration,  $\Lambda_k$ a diagonal matrix containing the RMSprop step sizes, $\lambda$ a scalar that we tune by hand to mitigate differences between the scales of log-likelihood's $\theta$ and $x$ gradients, and $\mathcal{B}_k$ indexes the image batch used in the $k^{th}$ parameter update (these are drawn uniformly at random without replacement until the dataset is exhausted, at which point the dataset is shuffled and the procedure is repeated).

In the particle updates, we subsample the $x$-gradients using the same image batches. Given the product-form structure in~\eqref{eq:nlvmiid}, this has the effect of only updating particle components index by the batches. That is, the $k^{th}$ update reads:
\begin{align}\label{eq:genxupdate}
X^{n,m}_{k+1}&=X^{n,m}_k + h\nabla_x \log p_{\theta_k}
(X^{n,m}_k, y^{m}) + \sqrt{2h} W^{n,m}_k \quad \forall m\in\cal{B}_k,\enskip X^{n,m}_{k+1}=X^{n,m}_{k},\quad\forall m\not\in\cal{B}_k,\quad n\in [N].
\end{align}
We initialized all particles by drawing independent samples from the Gaussian prior  $p(x)$.

\paragraph{Training (ABP and SR).} As mentioned in Sec.~\ref{sec:gen}, both ABP~\citep{Han2017} and SR~\citep{Nijkamp2020} are variants\footnote{Our implementations of ABP and SR are slight tweaks of their original presentations. In particular, in~\citet{Han2017} where ABP was introduced, the problems considered were small enough that no gradient subsampling was required, while here it is. As for SR, in~\citet{Nijkamp2020}, the authors additionally adaptively set the step size $h$ using a variational optimization approach. We abstain from doing so to simplify the comparison and place PGD, ABP, and SR in as equal footing as possible.} of (\ref{eq:IW2gradalg11},\ref{eq:SOUL})  proposed specifically for training generator networks. Just as for PGD above, for ABP and SR we subsample the gradients and adapt the parameter step sizes using RMSProp. However, ABP and SR use only~\eqref{eq:SOUL}'s final state to approximate the posterior $p_{\theta_{k}}(\cdot|y^{1:M})$  (they approximate it with $\delta_{X_k^N}$);  and so the parameter updates read
\begin{equation}\label{eq:genthetaupdate2}
\theta_{k+1} = \theta_k + h\lambda\Lambda_k\frac{M}{M_{\mathcal{B}}}\left[ \sum_{m\in\cal{B}_{k}}\nabla_\theta\log p_{\theta_k}
(y^{m}|X^{N,m}_k) \right],
\end{equation}
where $M_{\mathcal{B}},\Lambda_k,\lambda$ are as in~\eqref{eq:genthetaupdate}. To update the particles, we run a version of \eqref{eq:SOUL} with gradients subsampled  similarly as in~\eqref{eq:genxupdate}:
\begin{equation}X_{k}^{n+1,m}=X_{k}^{n,m}+h\nabla_x  \log p_{\theta_k}(X_{k}^{n,m},y^m)+ \sqrt{2h} W^{n,m}_k \quad \forall m\in\cal{B}_k,\quad X_{k}^{n+1,m}=X_{k}^{n,m}\quad\forall m\not\in\cal{B}_k,\quad n\in [N-1].\label{eq:genxupdate2}\end{equation}
In the case of ABP, the above chain is `persistent': $X_k^1$ is initialized at $X_{k-1}^N$ for $k>0$. In that of SR, it is not: $X_k^1$ is drawn from the prior $p(x)$ for $k>0$. In both cases, $X_0^1$ is drawn from the prior.

With the above choices, PGD, ABP, and SR all carry a similar computational cost for the same iteration number $K$ and particle number $N$.

\paragraph{Hyperparameters (PGD, ABP, and SR).} We chose the hyper-parameters featuring in (\ref{eq:genthetaupdate}--\ref{eq:genxupdate2}) as follows:
\begin{itemize}
\item \textbf{$\pmb{K}$, $\pmb{N}$.} We found that, modulo some noise, test errors for all three algorithms decrease monotonically with increasing iteration count $K$ and particle number $N$. We chose $K$ large enough that increasing it further lead to no more noticeable improvements in test errors ($K=39,500$ for MNIST and $K=78,250$ for CelebA). We set the particle number $N$ to $10$ which we found to be a good compromise between training times and performance on test errors (larger $N$ values did result in small, but noticeable, decreases in test errors).
\item \textbf{$\pmb{h}$, $\pmb{\lambda}$.} We chose these parameters small enough that no algorithm was on the verge of becoming unstable but large enough that the training was not excessively slow ($h=10^{-3}$, $\lambda=10^{-4}$ for MNIST and $h=10^{-4}$, $\lambda=2.5\times10^{-4}$ for CelebA). We did not observe a significant change in test errors by varying these values by $\pm$ an order of magnitude (the only noticeable effect was that smaller values led to slower training).
\item \textbf{$\pmb{M_{\cal{B}}}$.} In general, we observed that the larger the batch size, the quicker the training. Its value seemed to  affect little the test errors after training. We set the batch size to $128$, which ensured that no virtual memory was required during training while not being excessively small.
\end{itemize}

\paragraph{Training (VAE).} VAEs~\citep{Kingma2014} are variational inference methods where a parametric approximation $q_\phi(x^{1:M})$ to the posterior $p(x^{1:M}|y^{1:M})$ is chosen and training consists of solving~\eqref{eq:vi} with an appropriate optimization algorithm. VAEs use approximations of the sort
\begin{align*}
q_\phi(x^{1:M}|y^{1:M})=\prod_{m=1}^Mq_\phi(x^{m}|y^{m}),\quad\text{where}\quad q_\phi(x^m|y^m)=\cal{N}(x^m;g_\phi^{\text{mean}}(y^m),\textrm{diag}(g_\phi^{\text{var}}(y^m))),
\end{align*}
with $g_\phi^{\text{mean}},g_\phi^{\text{var}}:\r^{d_y}\to\r^{d_x}$ denoting neural networks parametrized by $\phi$ and $\textrm{diag}(g_\phi^{\text{var}}(y^m))$ a diagonal matrix with $g_\phi^{\text{var}}(y^m)$ on its diagonal.  We follow the common choice, e.g.\ ~\citet[App.C.2]{Kingma2014}, of setting 
$$g_\phi^{\text{mean}}=l_\phi^{\text{mean}} \circ g_\phi, \quad g_\phi^{\text{var}}=\textrm{SoftPlus} \circ l_\phi^{\text{var}} \circ g_\phi$$
where $l_\phi^{\text{mean}}$ and $l_\phi^{\text{var}}$ are fully connected linear layers and $g_\phi$ denotes a third network whose parameters are shared across $g_\phi^{\text{mean}}$ and $g_\phi^{\text{var}}$. For $g_\phi$ we use a simple convolutional network with RELU activation functions:
$$g_\phi=l_\phi   \circ mp_\phi \circ \textrm{ReLU} \circ c_\phi \circ mp_\phi \circ \textrm{ReLU} \circ c_\phi,$$
where $l_\phi$ denotes a fully connected layer, $c_\phi$ convolutional layers, and $mp_\phi$ max pooling layers.  In total, the networks $g_\phi^{\text{mean}},g_\phi^{\text{var}}$ involve $1,119,552$ parameters for MNIST and $1,119,840$ for CelebA.

We train the model by simultaneously running RMSprop~\citep{Hinton2012} for the model parameters $\theta$ and Adam~\citep{Kingma2014a} for the variational parameters $\phi$, both with learning rates of $10^{-3}$ and all other values set to their defaults in PyTorch $1.12$'s implementations of RMSprop and Adam (cf.\  \href{https://pytorch.org/docs/stable/generated/torch.optim.RMSprop.html}{here} and \href{https://pytorch.org/docs/stable/generated/torch.optim.Adam.html}{here}). 

\paragraph{Inpainting.}

\begin{figure}
    \centering
    \includegraphics[width=1\linewidth]{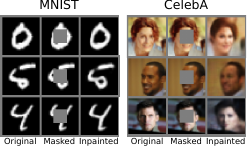}
    \caption{Inpainted images obtained using the generator $f_\theta$ trained with PGD as described in the text.}
    \label{fig:inpaint} 
\end{figure}

We use  generators $f_{\theta_K}$ trained with PGD, ABP, SR, and VAE to recover images $y=(y_i)_{i=1}^{d_y}$ that have been corrupted by masking some of their pixels. To do so, we follow the approach taken in~\citet[Sec.~5.3]{Nijkamp2020}: we search for latent variables that maximize the likelihood of the corrupted image $y_c$,
\begin{equation}\label{eq:recon}
x_{\text{mle}} = \textrm{argmax}_{x\in\mathbb{R}^{dx}} \log p (y_{c}|x)= \textrm{argmin}_{x\in\mathbb{R}^{dx}} \|y_c - f_{\theta_K}(x)\|^2_2=\textrm{argmin}_{x\in\mathbb{R}^{dx}}\sum_{i\not\in\mathcal{M}}[y_{i}-f_{\theta_K}(x)_i]^2,\end{equation}
where $\mathcal{M}$ indexes the masked pixels. Then, we recover the image by mapping $x_{\text{mle}}$ through the generator: $y\approx f_\theta(x_\textrm{mle})$. To (approximately) solve the above we use $4$ randomly initialized runs of Adam~\citep{Kingma2014a}, each a thousand steps long. We set the learning rate adaptively using Pytorch's \href{https://pytorch.org/docs/stable/generated/torch.optim.lr_scheduler.ReduceLROnPlateau.html}{ReduceLROnPlateau} scheduler with an initial learning rate of $1$ and all other Adam parameters set to their Pytorch 1.12 \href{https://pytorch.org/docs/stable/generated/torch.optim.Adam.html}{defaults}. With this approach, and the PGD-trained generator, we obtained the inpaintings shown in Fig.~\ref{fig:inpaint}.

\begin{figure}[t]
  \centering
  \includegraphics[width=1\linewidth]{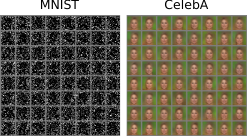}
  \caption{Images synthesized by drawing samples from the Gaussian prior $p(x)$ and mapping them through the generator $f_{\theta_K}$ trained with PGD.}\label{fig:syn_prior}
\end{figure}

\paragraph{Image synthesis.}Regardless of the algorithm we used for training, synthesizing images following the usual approach of drawing latent variables $x$ from the prior $p(x)$ and mapping them through the trained generator $f_{\theta_K}$ gave poor results (e.g.~Fig.~\ref{fig:syn_prior}). This is a known issue for these types of models. For example, as explained in~\citet{Aneja2021}:
\begin{quotation}
``Variational autoencoders (VAEs) are one of the powerful likelihood-based generative models with applications in many domains. However, they struggle to generate high-quality images, especially when samples are obtained from the prior without any tempering. One explanation for VAEs’ poor generative quality is the prior hole problem: the prior distribution fails to match the aggregate approximate posterior. Due to this mismatch, there exist areas in the latent space with high density under the prior that do not correspond to any encoded image. Samples from those areas are decoded to corrupted images.''
\end{quotation}
More specifically, in all four algorithms, we approximate the posterior $p_{\theta_k}(x^{1:M}|y^{1:M})$ using a product-form\footnote{This is not quite true for PGD, ABP, and SR as~(\ref{eq:genthetaupdate}--\ref{eq:genxupdate2}) correlate $(X^{n,m}_k)_{n\in[N],m\in[M]}$ through $\theta_k$. However, for large $M$ and $M_{\cal{B}}$, these correlations are small, and we ignore them here to simplify the discussion.} distribution: 
$$
q_k(dx^{1:M})=\prod_{m=1}^Mq_k^m(dx^m);$$
 where $q_k^m(dx) :=N^{-1}\sum_{n=1}^N\delta_{X^{n,m}_k}(dx)$ for PGD, $q_k^m(dx) :=\delta_{X^{N,m}_k}(dx)$ for ABP and SR, and $q_k^m(dx):=q_{\phi_k}(dx|y^m)$ for VAE. Emulating\footnote{These calculations are only formal in the case of PGD, ABP, and SR because $q_k^1,\dots,q_k^M$ have no densities w.r.t.\ the Lebesgue measure.} the calculations in~\citet{Hoffman2016}, we find that 
\begin{align}\label{eq:priorbound} 
F(\theta_k,q_k)=\sum_{m=1}^M\int\log\left(\frac{q_k^m(x^m)}{p(x^m)p_{\theta_k}(y^m|x^m)}\right)q^m(x^m)dx^m\geq \sum_{m=1}^MKL(q^m_k||p)\geq MKL(q^{agg}_k||p),
\end{align}
where 
$$q^{agg}_k(dx):=\frac{1}{M}\sum_{m=1}^Mq^m(dx)$$ 
denotes the aggregate (approximate) posterior~\citep{Aneja2021}. To derive the rightmost inequality in~\eqref{eq:priorbound} note that
\begin{align*}
&\sum_{m=1}^MKL(q^m_k||p)-MKL(q^{agg}_k||p)=\sum_{m=1}^M\int\log\left(\frac{q^m_k(x)}{q^{agg}_k(x)}\right)q^m_k(x)dx\\
&\qquad\qquad=\sum_{m=1}^M\int\log\left(\frac{q^m_k(x)}{Mq^{agg}_k(x)}\right)q^m_k(x)dx+M\log(M)\\
%\frac{1}{M}\sum_{m=1}^M\int \log(q^m_k(x))q^m_k(x)dx-\int \log(q^{agg}_k(x))q^{agg}_k(x)dx\\
&\qquad\qquad=M\int\left[\sum_{m=1}^M\log\left(\frac{q^m_k(x)}{Mq^{agg}_k(x)}\right)\frac{q^m_k(x)}{Mq^{agg}_k(x)}\right]q^{agg}_k(x)dx+Mlog(M)
\end{align*}
For each $x$, the term inside the square brackets is the negative entropy of the distribution $\left(\frac{q^m_k(x)}{Mq^{agg}_k(x)}\right)_{m=1}^M$ and, hence, bounded below by $-\log(M)$; and~\eqref{eq:priorbound} follows.

\begin{figure}[t]
  \centering
  \includegraphics[width=1\linewidth]{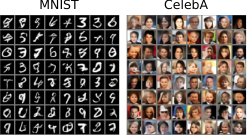}
  \caption{Images synthesized by sampling  a Gaussian approximation of the PGD aggregate posterior.}\label{fig:syn_1}
\end{figure}

\eqref{eq:priorbound}  shows  that the free energy is bounded below by $M$ times the KL divergence between the aggregate posterior $q_k^{agg}$ and the prior $p$. As noted in~\citep{Rosca2018}: 

\begin{quotation}
``VAEs are unable to match the marginal latent posterior (aggregate posterior) to the prior. This will result in a failure to learn the data distribution, and manifests in a discrepancy in quality between samples and reconstructions from the model.''
\end{quotation}
In our experiments, we observed the same phenomenon for PGD, ABP, SR, and VAE. In all four cases, it appears that the model learns to generate qualitatively meaningful images in the regions of the latent space where the aggregate posterior places mass but not in those where the prior does.

To overcome the bottleneck in~\eqref{eq:priorbound} and improve the image generation, a variety of schemes that learn the prior as well as the generator $f_\theta$ have been proposed in the literature (e.g.\ \citet{Tomczak2018,Bauer2019,Klushyn2019,Dai2019,Pang2020,Aneja2021}). We limit ourselves to simply fitting a Gaussian  $\mathcal{N}(\mu,\Sigma)$ to the (trained) aggregate posterior $q_K^{agg}$. For PGD, ABP, and SR, $q_K^{agg}$ is an empirical distribution of the form $J^{-1}\sum_{j=1}^J\delta_{Z^j}$ and we fit $\mathcal{N}(\mu,\Sigma)$ using $q_K^{agg}$'s empirical mean and covariance:
$$\mu:=\frac{1}{J}\sum_{j=1}^JZ^{j},\quad \Sigma:=\frac{1}{J-1}\sum_{j=1}^J(Z^j-\mu)(Z^{j}-\mu)^T.$$
In the case of VAE, we first build an empirical distribution $J^{-1}\sum_{j=1}^J\delta_{Z^j}$ by drawing $10$ samples from each $q^1_k,\dots, q^M_k$, and then proceed as above. Synthesizing images by drawing samples from $\cal{N}(\mu,\Sigma)$ and mapping them through $f_\theta$ then produces substantially higher quality images~(Fig.~\ref{fig:syn_1}). Building more refined approximations of the aggregate posterior further improves the images~(Fig.~\ref{fig:syn_500}).

\begin{figure}[t]
  \centering
  \includegraphics[width=1\linewidth]{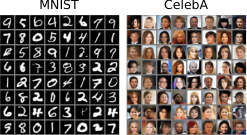}
  \caption{Images synthesized using  by sampling a $500$-component mixture of Gaussian approximation of the PGD aggregate posterior. To fit the mixture, we applied scikit-learn's \href{https://scikit-learn.org/stable/modules/generated/sklearn.mixture.GaussianMixture.html}{default procedure} to $J^{-1}\sum_{j=1}^J\delta_{Z^j}$.}\label{fig:syn_500}
\end{figure}
  
\paragraph{Performance metrics.} To evaluate the performance the trained generators $f_{\theta_K}$ in the inpainting task, we mask and inpaint $1000$ images $y^1,\dots,y^{1000}$ randomly chosen from the test set. For each of these, we solve~\eqref{eq:recon} to obtain matching latent variable vectors $x^1,\dots,x^{1000}$ and inpainted images $f_\theta(x^1),\dots,f_\theta(x^{1000})$. We then compute the latter's mean squared error (averaged over both pixels and test images):
\begin{align*}
MSE = \frac{1}{1000 d_y}\sum_{m=1}^{1000}\sum_{i=1}^{d_y}[y^m_i-f_\theta(x^m)_i]^2.
\end{align*}

To evaluate the performance of the trained generators $f_{\theta_K}$ in the synthesis task, we synthesize $200$ images as described above, randomly pick $200$ images from the test set, and compute the corresponding  Fr\'echet Inception Distance (FID; \cite{Heusel2017}) with the  Inception v3 classifier~\citep{Szegedy2016} between these two ensembles --- we use TorchMetrics's~\citep{Detlefsen2022} implementation of FID. In the case of the greyscale MNIST images, this requires mirroring the image across the three colour channels. 
We recognize that there are conceptual difficulties with this ad hoc approach, especially given that the training data for Inception v3 differs qualitatively from the MNIST images. However, we verified that there is a qualitative (as judged by eye) improvement in image quality associated with increasing FID score computed in this way and felt it sensible to follow this now-established approach for this dataset (e.g.\ see the papers reporting FID scores for MNIST on \href{https://paperswithcode.com/sota/image-generation-on-mnist}{paperswithcode.com}).

\section{THE MEAN-FIELD LIMITS AND THE TIME-DISCRETIZATION BIAS}\label{app:biasdisc}
%The mean-field limits and the time-discrezation bias

The mean-field ($N\to\infty$) limits  of PGD's update equations, (\ref{eq:IW2gradalg11},\ref{eq:IW2gradalg12}), are (\ref{eq:gradmf1},\ref{eq:gradmf2}), those of PQN's, (\ref{eq:newtonalg},\ref{eq:IW2gradalg12}), are  (\ref{eq:newtonmf},\ref{eq:gradmf2}), and that of PMGD's, \eqref{eq:W2gradalg}, is \eqref{eq:margradmf}:
\begin{align}
\theta_{k+1} &= \theta_{k} + h\int \nabla_\theta \ell(\theta_k,x)q_k(x)dx,\label{eq:gradmf1}\\ 
X_{k+1}&=X_k+h\nabla_x \ell(\theta_k,X_k)+\sqrt{2h}W_k,\label{eq:gradmf2}\\
\theta_{k+1} &= \theta_{k} - h\left[ \nabla_\theta^2\ell(\theta,x)q_k(x)dx\right]^{-1}\int \nabla_\theta \ell(\theta_k,x)q_k(x)dx,\label{eq:newtonmf}\\
X_{k+1}&=X_k+h\nabla_x \ell(\theta_*(q_k),X_k)+\sqrt{2h}W_k.\label{eq:margradmf}
\end{align}
where, in all cases, $q_k$ denotes $X_k$'s law and, in~\eqref{eq:margradmf} we are assuming that Assumpt.~\ref{ass:Mstep} holds. We can re-write (\ref{eq:gradmf1},\ref{eq:newtonmf}) as
\begin{equation}
\label{eq:mfupdate}\theta_{k+1}=u(\theta_k,q_k),
\end{equation}
where $u:\Theta\times\cal{P}(\cal{X})\to\Theta$ denotes an `update' operator satisfying
\begin{equation}\label{eq:update}\forall (\theta,q)\in\Theta\times\cal{P}(\cal{X}),\quad \nabla_\theta F(\theta,q) =0\enskip \Rightarrow\enskip u(\theta,q)=\theta.\end{equation}
Now, (\ref{eq:gradmf2},\ref{eq:mfupdate})'s joint law $q_k(d\theta,dx)$ satisfies
\begin{equation}\label{eq:mfFP}
q_{k+1}(d\psi,dz)=\int_{\theta,x}q_{k}(d\theta,dx)\delta_{u(\theta,q_k)}(d\psi)K_\theta(x,dz),
\end{equation}
where $K_\theta$ denotes the ULA kernel:
\begin{equation}\label{eq:ULAkernel}K_{\theta}(x,dz)=K_{\theta}(x,z)dz:=\cal{N}\left(z;x+h\nabla_x \ell(\theta,x),2hI_{D_x}\right)dz\quad\forall x\in\cal{X},\enskip\theta\in \Theta.\end{equation}
What we would like is for 
\[\pi(d\theta,dx):=\delta_{\theta_*}(d\theta)p_{\theta_*}(dx|y)\]
to be a fixed point of \eqref{eq:mfFP} whenever $\theta_*$ is a maximizer of $\theta\mapsto p_\theta(y)$ and $p_{\theta_*}(dx|y)$ is the corresponding posterior. However, applying Thrm.~\ref{thrm:statpoints} and~\eqref{eq:update}, we find that
\begin{align*}
\int_{\theta,x}\pi(d\theta,dx)\delta_{u(\theta,\pi(dx))}(d\psi)K_\theta(x,dz)&=\int_{\theta,x}\delta_{\theta_*}(d\theta)p_{\theta_*}(dx|y)\delta_{u(\theta,p_{\theta_*}(dx|y))}(d\psi)K_\theta(x,dz)\\
&=\delta_{u(\theta_*,p_{\theta_*}(dx|y))}(d\psi)\int_x p_{\theta_*}(dx|y)K_{\theta_*}(x,dz)\\
&=\delta_{\theta_*}(d\psi)\int_x p_{\theta_*}(dx|y)K_{\theta_*}(x,dz).
\end{align*}
Hence, $\pi$ is a fixed point of~\eqref{eq:mfFP} if and only $p_{\theta_*}(dx|y)$ is a stationary distribution of $K_{\theta_*}(x,dz)$:
\begin{equation}
\label{eq:ULAstat}p_{\theta_*}(z|y)=\int p_{\theta_*}(x|y) K_{\theta_*}(x,z)dx
\end{equation}
However, we know this is not the case because the ULA kernel is biased, e.g.\ see \cite{Roberts1996}.

The case of~\eqref{eq:margradmf} is similar: $q_k$ satisfies
\[q_{k+1}(z)=\int q_{k}(x)K_{\theta_*(q_k)}(x,z)dx.\]
Given that $\theta_*(p_{\theta_*}(\cdot|y))=\theta_*$ (Thrm.~\ref{thrm:statpoints}), we have that $p_{\theta_*}(\cdot|y)$ is a fixed point of the above if and only if \eqref{eq:ULAstat} holds, which it does not.

An obvious way to get~\eqref{eq:ULAstat} to hold is replacing the ULA kernel $K_\theta$ with a kernel whose stationary distribution is the posterior $p_\theta(\cdot|y)$ (e.g.~by adding an accept-reject step to the ULA kernel). This removes the time-discretization bias in the mean-field regime. In App.~\ref{app:MH}, we will see another (slightly less obvious) way to do so.
\subsection{Rates of convergence for Ex.~\ref{ex:hier}}\label{app:hierconvrates}
To investigate the rate of convergence of PGD, PQN, and PMGD in the case of the toy hierarchical model (Ex.~\ref{ex:hier}), we examine the mean-field limits~(\ref{eq:gradmf1}--\ref{eq:margradmf}) which respectively read:
\begin{align}
\theta_{k+1} &= \theta_{k} + hD_x[\nu_k-\theta_k],\label{eq:gradmf1hier}\\ 
X_{k+1}&=X_k+h[y+\theta_k\bm{1}_{D_x}-2X_k]+\sqrt{2h}W_k,\label{eq:gradmf2hier}\\
\theta_{k+1} &= \theta_{k} + h\left[\nu_k-\theta_k\right],\label{eq:newtonmfhier}\\
X_{k+1}&=X_k+h\left[y+\nu_k\bm{1}_{D_x}-2X_k\right]+\sqrt{2h}W_k.\label{eq:margradmfhier}
\end{align}
where, in all cases,  $\nu_k:=\Ebb{\bm{1}_{D_x}^TX_k/D_x}$ denotes the mean of the average of $X_k$'s components. Left-multiplying~(\ref{eq:gradmf2hier},\ref{eq:margradmfhier}) by $D_x^{-1}\bm{1}_{D_x}^T$ and taking expectations in, we respectively find that
\begin{align}
\nu_{k+1}&=\nu_k+h[\bm{1}_{D_x}^Ty/D_x+\theta_k-2\nu_k]\label{eq:gradmf2hiermu}\\
\nu_{k+1}&=\nu_k+h\left[\bm{1}_{D_x}^Ty/D_x-\nu_k\right].\label{eq:margradmfhiermu}
\end{align}
Note that both (\ref{eq:gradmf1hier},\ref{eq:gradmf2hiermu}) and (\ref{eq:newtonmfhier},\ref{eq:gradmf2hiermu}) have a unique fixed point $(\theta_\infty,\nu_\infty)$ given by $\theta_\infty=\nu_\infty=\theta_*$, where $\theta_*= \bm{1}_{D_x}^Ty/D_x$ denotes the marginal likelihood's unique maximizer (cf.\ App.~\ref{app:hier}). Re-writing (\ref{eq:gradmf1hier},\ref{eq:gradmf2hiermu}) and (\ref{eq:newtonmfhier},\ref{eq:gradmf2hiermu}) in matrix-vector notation,
\begin{align*}
\begin{bmatrix}\theta_{k+1}\\\nu_{k+1}\end{bmatrix} 
&=A_h^G\begin{bmatrix}\theta_{k}\\\nu_{k}\end{bmatrix} +\begin{bmatrix}0\\h\bm{1}_{D_x}^Ty/D_x\end{bmatrix}\quad\text{where}\quad A_h^G:=\begin{bmatrix}1-hD_x&hD_x\\h&1-2h\end{bmatrix},\\
\begin{bmatrix}\theta_{k+1}\\\nu_{k+1}\end{bmatrix} 
&=A_h^N\begin{bmatrix}\theta_{k}\\\nu_{k}\end{bmatrix} 
+\begin{bmatrix}0\\h\bm{1}_{D_x}^Ty/D_x\end{bmatrix}\quad\text{where}\quad A_h^N:=\begin{bmatrix}1-h&h\\h&1-2h\end{bmatrix},
\end{align*}
then clarifies that $\theta_k$'s speed of convergence to $\theta_*$ is $\cal{O}(\rho_{G,h}^k)$ in the case of (\ref{eq:gradmf1hier},\ref{eq:gradmf2hiermu}) and $\cal{O}(\rho_{N,h}^k)$ in that of (\ref{eq:newtonmfhier},\ref{eq:gradmf2hiermu}), where $\rho_{G,h}$ denotes $A_h^G$'s spectral radius and $\rho_{N,h}$ denotes $A_h^N$'s.  After some quick algebra, we find that
\[\rho_{G,h}=\max\left\{\mmag{1-h\left(1+\frac{D_x}{2}\pm\frac{\sqrt{D^2_x+4}}{2}\right)}\right\},\quad \rho_{N,h}=\max\left\{\mmag{1-h\left(\frac{3}{2}\pm\frac{\sqrt{5}}{2}\right)}\right\}.\]
As for PMGD's mean-field limit~\eqref{eq:margradmfhier}, recall that we use $\theta_*(q_k)=D_x^{-1}\bm{1}^T_{D_x}\int xq_k(x)dx=\nu_k$ to estimate $\theta_*$, where $q_k$ denotes $X_k$'s law, and note that $\nu_k$ in~\eqref{eq:margradmfhiermu} converges to $\theta_*$ at a   rate of $\cal{O}(\rho_{M,h}^k)$, where  $\rho_{M,h}:=\mmag{1-h}$. Two observations are in order:
\begin{itemize}
\item\textbf{Dependence on $D_x$.} In the case of PGD, the radius $\rho_{G,h}$ depends on the dimension $D_x$ of the latent space. For large dimensions, $\rho_{G,h}\approx \mmag{1-h(1+D_x)}$ implying that PGD is  stable only for very small step sizes (roughly, those smaller than $2/(1+D_x)$), which explains the need for the tweak~\eqref{eq:IW2gradalg2}. On the other hand, the radii for PQN and PMGD are independent of $D_x$. For these reasons, tuning the step size for PGD proves challenging and delicately depends on $D_x$, while tuning it for PQN and PMGD is straightforward and does not require taking $D_x$ into account.
\item\textbf{Relative speeds.} For all step sizes, $\rho_{G,h}$ and $\rho_{N,h}$ are both bounded below by $\rho_{M,h}$ (see Fig.~\ref{fig:radii}), implying that PMGD always converges faster than PGD and PQN, at least in the mean-field regime.  It is not necessarily the case that  $\rho_{N,h}\leq \rho_{G,h}$: for small step sizes, this fails to hold. However, the range of $h$s for which $\rho_{N,h}> \rho_{G,h}$ decreases precipitously with the latent space dimension $D_x$ (Fig.~\ref{fig:radii}). Hence, we expect PQN to outperform PGD unless we use very small step sizes (likely, those too small to achieve any reasonable convergence speed). 
\end{itemize}

\begin{figure}[h!]
    \begin{center}
    \includegraphics[width=0.9\textwidth]{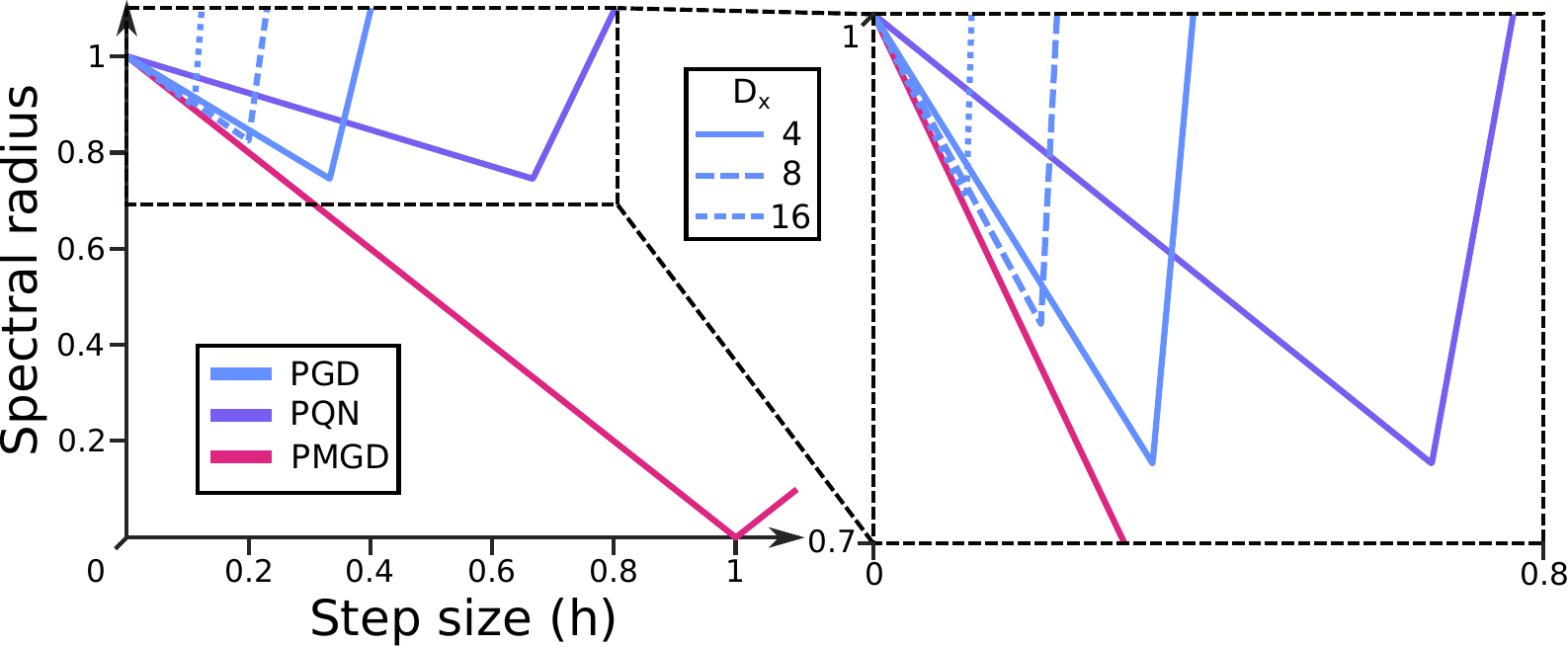}
    \end{center}
\caption{Spectral radii $\rho_{G,h}$ (PGD), $\rho_{N,h}$ (PQN), and $\rho_{M,h}$ (PMGD) as a function of step size $h$.}\label{fig:radii}
\end{figure}

Consider now the matter of choosing the step size $h$ that achieves the fastest convergence for each algorithm. That is, the $h$ that minimizes the corresponding radius. It is straightforward to verify that
\begin{align*}
h_G=\frac{2}{2+D_x},\enskip h_N=\frac{2}{3},\enskip  h_M=1, \quad\text{with} \quad \rho_{G}=\frac{\sqrt{D_x^2+4}}{D_x+2},\enskip\rho_{N}=\frac{\sqrt{5}}{3},\enskip\rho_{M}=0,
\end{align*}
where $h_G,h_N,h_M$ respectively denote the optimal  step sizes for PGD, PQN, and PMGD, and $\rho_G,\rho_N,\rho_M$ the corresponding radii. Note that,  $\rho_{G}\geq \rho_{N}\geq \rho_{M}$ whenever $D_x\geq 4$. Hence, except for very low dimensional cases with well-tuned step sizes, PQN will outperform PGD. PMGD  will always outperform either. Moreover, $\rho_G\to1$ as $D_x\to\infty$ and, hence, PGD's convergence speed degenerates with increasing latent space dimension regardless of the step size $h$ that we use. That of PQN does not and, if we tune the algorithm well, will be $\cal{O}([2/3]^k)$ for all $D_x$. Setting $h:=1$ in~\eqref{eq:margradmfhiermu}, we find that,  regardless of $D_x$, PMGD's parameter estimates will converge in a single step, at least in the mean-field regime. Of course, this fails to materialize when we run the algorithm in practice because the noise in~\eqref{eq:margradmfhier} exacts an $\cal{O}(1/\sqrt{KND_x})$ error in our time-averaged estimates, which is what we see in Fig.~\ref{fig:hier}\textbf{c} (similar considerations also apply to PGD and PQN in stationarity).  Lastly, we ought to mention that we have observed these behaviours replicated across other numerical experiments, hinting that they might hold more widely. However, until an analysis establishing so becomes available, we only count this as anecdotal evidence.
\section{THE CONTINUUM LIMITS AND THE FINITE-POPULATION-SIZE BIAS}\label{app:biasfinitepop}
%The continuum limits and the finite-population-size bias

For the sake of simplicity, suppose that Assumpt.~\ref{ass:Mstep} holds. The continuum limits ($h\to0$) of PGD's update equations, (\ref{eq:IW2gradalg11},\ref{eq:IW2gradalg12}), are (\ref{eq:gradcl1},\ref{eq:gradcl2}), those of PQN's, (\ref{eq:newtonalg},\ref{eq:IW2gradalg12}), are  (\ref{eq:newtoncl},\ref{eq:gradcl2}), and that of PMGD's, \eqref{eq:W2gradalg}, is \eqref{eq:margradcl}:
\begin{align}
d\theta_t &= \frac{1}{N}\left[\sum_{n=1}^N \nabla_\theta \ell(\theta_t, X_t^n)\right]dt,\label{eq:gradcl1}\\ 
dX_t^n&=\nabla_x \ell(\theta_t, X_t^n)dt+\sqrt{2}dW_t^n\quad\forall n\in[N],\label{eq:gradcl2}\\
d\theta_t &= -\left[\sum_{n=1}^{N} \nabla_\theta^2\ell(\theta_t,X_t^n)\right]^{-1}\left[\sum_{n=1}^N \nabla_\theta \ell(\theta_t, X_t^n)\right]dt,\label{eq:newtoncl}\\
dX_t^n&=\nabla_x \ell(\theta_*(X_t^{1:N}), X_t^n)dt +\sqrt{2}dW_t^n\quad\forall n\in[N].\label{eq:margradcl}
\end{align}
As shown below, \eqref{eq:margradcl}'s law satisfies 
\begin{equation}\label{eq:marginalclfp}
\dot{q}_t(x^{1:N})=\nabla_{x^{1:N}}\cdot\left[ q_t(x^{1:N})\nabla_{x^{1:N}} \log\left(\frac{q_t(x^{1:N})}{\rho_N(x^{1:N})}\right)\right],
\end{equation}
where $\rho_N$ is the (unnormalized) distribution on $\cal{X}^N$ given by
\begin{equation}\label{eq:rhoN}\rho_N(x^{1:N}):=\prod_{n=1}^Np_{\theta_*(x^{1:N})}(x^n,y)\end{equation}
Clearly, the (unique) normalized fixed point of \eqref{eq:marginalclfp} (i.e.\ the stationary distribution of~\eqref{eq:margradcl}) is 
\begin{equation}\label{eq:piN}\pi_N(x^{1:N}):=\frac{\rho_N(x^{1:N})}{\cal{Z}_N}\quad\text{where}\quad\cal{Z}_N:=\int\rho_N(x^{1:N})dx^{1:N}.\end{equation}
As also shown below, (\ref{eq:gradcl1},\ref{eq:gradcl2})'s law satisfies
\begin{align}\label{eq:gradclfp}
\dot{q}_t(\theta,x^{1:N})=&\nabla_{x^{1:N}}\cdot\left[ q_t(\theta,x^{1:N})\nabla_{x^{1:N}} \log\left(\frac{q_t(\theta,x^{1:N})}{\prod_{n=1}^Np_{\theta}(x^{n},y)}\right)\right]\\
&-\nabla_\theta\cdot  \left[\frac{q_t(\theta,x^{1:N})}{N}\sum_{n=1}^N \nabla_\theta \ell(\theta,x^n)\right],\nonumber
\end{align}
where, if necessary, the above should be interpreted weakly. Because $\sum_{n=1}^N \nabla_\theta \ell(\theta_*(x^{1:N}),x^n)=0$ by $\theta_*(x^{1:N})$'s definition, it is easy to check that 
\[\pi_N(d\theta,dx^{1:N})=\delta_{\theta_*(x^{1:N})}(d\theta)\pi_N(x^{1:N})dx^{1:N}\]
is a fixed point of~\eqref{eq:gradclfp} (i.e.\ a stationary distribution of~(\ref{eq:gradcl1},\ref{eq:gradcl2})). Similar manipulations show that the above is also a   stationary distribution of~(\ref{eq:gradcl2},\ref{eq:newtoncl}). Our algorithms use the empirical distribution of the particles to approximate the posterior. Hence, for the estimates they produce to be `unbiased', it would have to be the case that
\[\int \left(\theta_*(x^{1:N}),\frac{1}{N}\sum_{n=1}^N\delta_{x^n}(dx)\right)\pi_N(dx^{1:N})=(\theta_*,p_{\theta_*}(dx|y)),\]
for some stationary point $\theta_*$ of $\theta\mapsto p_\theta(y)$. However, because $\theta_*(x^{1:N})$ is invariant to permutations of $x^{1:N}$'s components, $\pi_N$'s definition in~(\ref{eq:rhoN},\ref{eq:piN}) implies that its marginals $\pi_N(dx^1),\dots,\pi_N(dx^N)$ all equal the same distribution, $\mu_N$ on $\cal{X}$. Hence,
\begin{align*}\int \left(\frac{1}{N}\sum_{n=1}^N\delta_{x^n}(dx)\right)\pi_N(dx^{1:N})&=\frac{1}{N}\sum_{n=1}^N\int \delta_{x^n}(dx)\pi_N(dx^n)=\frac{1}{N}\sum_{n=1}^N\int \delta_{x^n}(dx)\mu_N(dx^{n})\\
&=\frac{1}{N}\sum_{n=1}^N\int \mu_N(dx)=\mu_N(dx).\end{align*}
In summary, for our algorithms to yield unbiased estimates, it would need to be the case that
\[\left(\int \theta_*(x^{1:N})\pi_N(dx^{1:N}),\mu_N(dx)\right)=(\theta_*,p_{\theta_*}(dx|y)).\]
It is easy to find examples in which the above fails to hold (see, for instance, App.~\ref{app:hierfinitepopbias}). 

\begin{proof}[Proof of \eqref{eq:marginalclfp}]\eqref{eq:margradcl}'s Fokker-Planck equation (e.g. see~\citet{Chaintron2022}) reads
\begin{align}
\dot{q}_t(x^{1:N})=&-\sum_{n=1}^N\nabla_{x^{n}} \cdot[q_t(x^{1:N})\nabla_{x^n}\ell(\theta_*(x^{1:N}),x^n)]+\sum_{n=1}^N\nabla_{x^{n}}\cdot\nabla_{x^{n}}q_t(x^{1:N}).\label{eq:ndsay8fbwna8f}
\end{align}
But
\begin{align}
-\nabla_{x^{1:N}}\cdot\left[ q_t(x^{1:N})\nabla_{x^{1:N}} \log\left(\frac{\rho_N(x^{1:N})}{q_t(x^{1:N})}\right)\right]=-\sum_{n=1}^N\nabla_{x^{n}}\cdot\left[ q_t(x^{1:N})\nabla_{x^{n}} \log\left(\frac{\rho_N(x^{1:N})}{q_t(x^{1:N})}\right)\right],\label{eq:ndsay8fbwna8f2}
\end{align}
and, using $\rho_N(x^{1:N})$'s definition in~\eqref{eq:rhoN},
\begin{align}
\nabla_{x^{n}} \log\left(\frac{\rho_N(x^{1:N})}{q_t(x^{1:N})}\right)=&\nabla_{x^{n}} \log(\rho_N(x^{1:N}))-\nabla_{x^{n}}\log(q_t(x^{1:N}))\nonumber\\
=&\nabla_{x^{n}} \ell(\theta_*(x^{1:N}),x^n) +\sum_{m=1}^N\nabla_\theta \ell(\theta_*(x^{1:N}),x^m)  \cdot \nabla_{x^m}\theta_*(x^{1:N})\nonumber\\
&-\frac{\nabla_{x^{n}}q_t(x^{1:N})}{q_t(x^{1:N})}.\label{eq:ndsay8fbwna8f3}
\end{align}
But, $\theta_*(x^{1:N})$'s definition implies that 
\[\nabla_{x^1}\theta_*(x^{1:N})=\dots=\nabla_{x^N}\theta_*(x^{1:N}),\qquad \sum_{m=1}^N\nabla_\theta \ell(\theta_*(x^{1:N}),x^m)=0.\]
Hence, the middle term in \eqref{eq:ndsay8fbwna8f3}'s RHS equals zero and \eqref{eq:marginalclfp} follows from (\ref{eq:ndsay8fbwna8f},\ref{eq:ndsay8fbwna8f2}).
\end{proof}

\begin{proof}[Proof of \eqref{eq:gradclfp}] This is straightforward: (\ref{eq:gradcl1},\ref{eq:gradcl2})'s Fokker-Planck equation (e.g. see~\cite{Chaintron2022}) reads
\begin{align*}
\dot{q}_t(\theta,x^{1:N})=&- \nabla_\theta \cdot \left[\frac{q_t(\theta,x^{1:N})}{N}\sum_{n=1}^N \nabla_\theta \ell(\theta,x^n)\right]-\sum_{n=1}^N\nabla_{x^{n}} \cdot[q_t(\theta,x^{1:N})\nabla_{x^n} \ell(\theta,x^n)]\\
&+\sum_{n=1}^N\nabla_{x^{n}}\cdot\nabla_{x^{n}}q_t(\theta,x^{1:N}).
\end{align*}
But,
\begin{align*}
&\nabla_{x^{1:N}}\cdot\left[ q_t(\theta,x^{1:N})\nabla_{x^{1:N}} \log\left(\frac{\prod_{n=1}^Np_{\theta}(x^{n},y)}{q_t(\theta,x^{1:N})}\right)\right]\\
&=\sum_{n=1}^N\nabla_{x^{n}}\cdot\left[ q_t(\theta,x^{1:N})\nabla_{x^{n}} \log\left(\frac{\prod_{n=1}^Np_{\theta}(x^{n},y)}{q_t(\theta,x^{1:N})}\right)\right]\\
&=\sum_{n=1}^N\nabla_{x^{n}}\cdot \left[ q_t(\theta,x^{1:N})\left(\nabla_{x^{n}} \ell(\theta,x^n)-\frac{\nabla_{x^{n}}q_t(\theta,x^{1:N})}{q_t(\theta,x^{1:N})}\right)\right]\\
&=\sum_{n=1}^N\nabla_{x^{n}}\cdot [ q_t(\theta,x^{1:N})\nabla_{x^{n}}\ell(\theta,x^n)]+\sum_{n=1}^N\nabla_{x^{n}}\cdot \nabla_{x^{n}}q_t(\theta,x^{1:N}).
\end{align*}
\end{proof}

\subsection{Continuum limits for Ex.~\ref{ex:hier}}\label{app:hierfinitepopbias}

By~\eqref{eq:hierthetaopt} and $\rho_N$'s definition in~\eqref{eq:rhoN}, 
\begin{align*}
\rho_N(x^{1:N})&= \prod_{n=1}^N\exp\left(-\frac{1}{2}\norm{y-x^n}^2-\frac{1}{2}\norm{x^n-\frac{\bm{1}^T_{ND_x}x^{1:N}}{ND_x}\bm{1}_{D_x}}^2\right)\\
&=\exp\left(-\frac{1}{2}\norm{y^{1:N}-x^{1:N}}^2-\frac{1}{2}\norm{x^{1:N}-\frac{\bm{1}^T_{ND_x}x^{1:N}}{ND_x}\bm{1}_{ND_x}}^2\right),
\end{align*}
where $y^{1:N}$ stacks $N$ copies of $y$. Applying the expressions in  \citet[p.\ 92]{Bishop2006} and the Sherman–Morrison formula, we find that
\[
\pi_N(x^{1:N})=\frac{\rho_N(x^{1:N})}{\int \rho_N(x^{1:N})dx^{1:N}}=\cal{N}\left(x^{1:N};\frac{1}{2}\left(y^N+\frac{\bm{1}^T_{D_x}y}{D_x}\bm{1}_{ND_x}\right),\frac{1}{2}\left(I_{ND_x}+\frac{\bm{1}_{ND_x}\bm{1}^T_{ND_x}}{ND_x}\right)\right);\]
whose marginals are equal to
\[\mu_N(x)=\cal{N}\left(x;\frac{1}{2}\left(y+\frac{\bm{1}^T_{D_x}y}{D_x}\bm{1}\right),\frac{1}{2}\left(I_{D_x}+\frac{\bm{1}_{D_x}\bm{1}^T_{D_x}}{ND_x}\right)\right).\]
Comparing with~\eqref{eq:hieroptimum}, we see that the covariance matrix is slightly off with a $\cal{O}(N^{-1})$ error. However, due to the linearity in the model, there is no bias in the $\theta$ estimates:
\begin{align*}
\int \theta_*(x^{1:N}) \pi_N(x^{1:N})dx^{1:N}&=\frac{\bm{1}^T_{ND_x}\int x^{1:N}\pi_N(x^{1:N})dx^{1:N}}{ND_x}\\
&=\frac{1}{2ND_x}\left(\bm{1}^T_{ND_x}y^N+\frac{\bm{1}^T_{ND_x}\bm{1}_{ND_x}\bm{1}^T_{D_x}y}{D_x}\right)=\frac{\bm{1}^T_{D_x}y}{D_x}=\theta_*.
\end{align*}

\section{METROPOLIS-HASTINGS METHODS}\label{app:MH}
As mentioned at the end of App.~\ref{app:biasdisc}, one fairly obvious way to try to remove the bias (B1, Sec.~\ref{sec:pgd}) from the estimates produced by PGD, PQN, and PMGD is to replace the ULA kernels with `exact' kernels whose stationary distributions coincide with the posteriors (e.g.\ by adding an accept-reject step to each individual particle update). Here, we consider other, slightly less obvious and (to the best of our knowledge) novel extensions of the Metropolis-Hastings algorithm (e.g.\ see \citet{Andrieu2003}) that also tackle (S1,2 in Sec.~\ref{sec:intro}). While these methods need not necessarily be associated with an optimization routine, their comprehension is also aided by viewing  (S1,2) as a joint problem over $\theta$ and $q$.   
The methods have one practical downside that limits their scalability: similar to standard Metropolis-Hastings algorithms (e.g.\ see \citet{Beskos2013,Vogrinc2022,Kuntz2018,Kuntz2019} and references therein), the acceptance probability degenerates with increasing latent variable dimensions $D_x$ and the particle numbers $N$. This, in turn, forces us to choose small step sizes $h$ for large $D_x$ and $N$, which leads to slow convergence. This is why we focused on the `unadjusted' methods in the main text rather than the Metropolized ones in this appendix.

\subsection{Marginal variants}\label{sec:marginalMH}

\begin{algorithm}[t]
\begin{algorithmic}[1]
\STATE{\textit{Initial conditions:} $X^{1:N}_0:=(X^1_0,\dots,X_0^N)$. }
\FOR{$k=0,\dots, K-1$}
\STATE{\textit{Propose:}  draw $Z^{1:N}=(Z^1,\dots,Z^N)$  from $K_N$.}
\STATE{\textit{Generate uniform R.V.:}  draw $U_k$  from the uniform distribution on $[0,1]$.}
\IF{If $U_k\leq a(X^{1:N}_k,Z^{1:N})$, with $a(\cdot,\cdot)$ as in~\eqref{eq:acceptprob}}
\STATE{\textit{Accept:} set $X_{k+1}^{1:N}:=Z^{1:N}$.}
\ELSE
\STATE{\textit{Reject:} set $X_{k+1}^{1:N}:=X_{k}^{1:N}$.}
\ENDIF
\ENDFOR
\end{algorithmic}
 \caption{The marginal MH method.}
 \label{alg:mmh}
\end{algorithm}

Suppose  that Assumpt.~\ref{ass:Mstep} holds and, for the sake of discussion, that the marginal likelihood $\theta\mapsto p_\theta(y)$ has a unique maximizer $\theta_*$. Notice that the entire particle system $(X_k^{1:N})_{k=0}^\infty:=(X_k^1,\dots,X_k^N)_{k=0}^\infty$ generated by PMGD~\eqref{eq:W2gradalg} is a Markov chain taking values in $\cal{X}^N$. Its kernel is given by
\[K_N(x^{1:N},z^{1:N}):=\prod_{n=1}^NK_{\theta_*(x^{1:N})}(x^n,z^n),\]
where $K_\theta$ denotes the ULA kernel in~\eqref{eq:ULAkernel}. (The precise form of $K_\theta$ is immaterial to the ensuing discussion as long as, for each $\theta$, $K_\theta$ is a Markov kernel on $\cal{X}$.) Ideally,   $K_N$'s stationary distribution would be  $\prod_{n=1}^Np_{\theta_*}(x^n|y)$ but   (B1,2 in Sec.~\ref{sec:pgd}) preclude it. Correcting for this using an accept-reject step requires evaluating $\prod_{n=1}^Np_{\theta_*}(x^n,y)$, which we cannot do because $\theta_*$ is unknown. We can, however, evaluate $\rho_N(x^{1:N})$ in~\eqref{eq:rhoN} and instead add an accept-reject step with acceptance probability
\begin{equation}\label{eq:acceptprob}a_N(x^{1:N},z^{1:N}):=1\wedge\left(\frac{\rho_N(z^{1:N})K_N(z^{1:N},x^{1:N})}{\rho_N(x^{1:N})K_N(x^{1:N},z^{1:N})}\right).\end{equation}
Running Alg.~\ref{alg:mmh}, we then obtain a chain  $(X_k^{1:N})_{k=0}^\infty$ whose stationary distribution is given by $\rho_N(x^{1:N})$'s normalization, $\pi_N(x^{1:N})$ in~\eqref{eq:piN}. Under Assumpt.~\ref{ass:Mstep}, $\pi_N(x^{1:N})$ is the unique fixed point of the continuum limits of PGD, PQN, and PMGD, see App.~\ref{app:biasfinitepop}. In other words, by imposing the accept-reject step, we have removed the (B1, Sec.~\ref{sec:pgd}) source of bias (see Fig.~\ref{fig:hierbias}c).

The other source, (B2, Sec.~\ref{sec:pgd}) due to the finite population size, remains, but it can be mitigated by growing $N$. In particular, by its definition, $\theta_*(x^{1:N})$ is invariant to permutations $x^{1:N}$'s components. Hence, the components of any vector $(X^1,\dots,X^N)$ drawn from $\pi_N$ are exchangeable  (similarly  for $(X_k^1,\dots,X_k^N)$ in Alg.~\ref{alg:mmh}). By chaos (e.g.~\cite{Chaintron2022,Hauray2014}), we expect that, under appropriate technical conditions, there exists a distribution $\pi$ in $\cal{P}(\cal{X})$ to which all of $\pi_N$'s marginals converge. It  should  follow that, under $\pi_N$,
\[
\frac{1}{N}\sum_{n=1}^N\delta_{x^n}\approx  \pi\quad\Rightarrow\quad \theta_*(x^{1:N})=\theta_*\left(\frac{1}{N}\sum_{n=1}^N\delta_{x^n}\right)\approx \theta_*(\pi),
\]
with the above holding exactly in the $N\to\infty$ limit. Marginalising (\ref{eq:rhoN},\ref{eq:piN}) and taking limits we would then find that $\pi(x) \propto p_{\theta_*(\pi)}(x,y)$. In other words, $\pi$ satisfies the first order optimality condition for $F_*$ in Thrm.~\ref{thrm:statpointsmar}:  $\nabla F_*(\pi)=0$. It then follows from $\theta_*(q)$'s definition and Thrm.~\ref{thrm:statpointsmar}, that $\theta_*(\pi)$ is a stationary point of $\theta\mapsto p_\theta(y)$ and $\pi$ is the corresponding posterior $p_{\theta_*(\pi)}(\cdot|y)$. While we yet lack rigorous statements formalizing this discussion, it is easy to verify that it holds true for analytically tractable models (e.g.\ App.~\ref{app:hierfinitepopbias}), and our numerical experiments seem to corroborate it further (e.g.~Fig.~\ref{fig:hierbias}c).
\subsection{Joint variants}\label{sec:MHjoint}

We can also mitigate (B1, Sec.~\ref{sec:pgd}) for PGD (Alg.~\ref{alg:pgd}) and PQN (Alg.~\ref{alg:pqn}) using a population-wide accept-reject step along the lines of that in App.~\ref{sec:marginalMH}. To begin, note that these algorithms are special cases of
\begin{align}
\theta_{k+1}&=u(\theta_k,q_k^N)\quad\text{with}\quad q_k^N:=\frac{1}{N}\sum_{n=1}^N\delta_{X_k^n},\label{eq:jmh1}\\
X_{k+1}^{1:N}&\sim K_N^{\theta_k}(X_{k}^{1:N},\cdot)\quad\text{with}\quad K_N^\theta(x^{1:N},z^{1:N}):=\prod_{n=1}^NK_{\theta}(x^n,z^n),\label{eq:jmh2}
\end{align}
where $K_\theta$ denotes a Markov kernel on $\cal{X}$ for each $\theta$ in $\Theta$ (in particular, the ULA kernel in~\eqref{eq:ULAkernel}, but this is once again unimportant) and $u$ denotes an update operator satisfying~\eqref{eq:update}.
Clearly, $(\theta_k,X_k^1,\dots,X_k^N)_{k=0}^\infty$ forms a Markov chain with transition kernel
\[K_N((\theta,x^{1:N}),(d\psi,dz^{1:N}))=\delta_{u\left(\theta,\bar{\delta}_{x^{1:N}}\right)}(d\psi)K_N^{\theta}(x^{1:N},dz^{1:N})\quad\text{where}\quad \bar{\delta}_{x^{1:N}}:=\frac{1}{N}\sum_{n=1}^N\delta_{x^n}.\]
Emulating our steps in App.~\ref{sec:marginalMH}, we impose an accept-reject step with acceptance probability
\begin{equation}\label{eq:acceptprobjoint}a_N((x^{1:N},\theta),(z^{1:N},\psi)):=1\wedge\left(\frac{K^\psi_N(z^{1:N},x^{1:N})}{K^\theta_N(x^{1:N},z^{1:N})}\prod_{n=1}^N\frac{p_{\psi}(z^n,y)}{p_{\theta}(x^n,y)}\right),\end{equation}
so obtaining Alg.~\ref{alg:jmh}.

\begin{algorithm}[t]
\begin{algorithmic}[1]
\STATE{\textit{Initial conditions:} $\theta_0$ and $X^{1:N}_0:=(X^1_0,\dots,X_0^N)$. }
\FOR{$k=0,\dots, K-1$}
\STATE{\textit{Propose:} set $\psi:=u\left(\theta_k,N^{-1}\sum_{n=1}^N\delta_{X_k^n}\right)$ and  draw $Z^{1:N}$  from $K_N^{\theta_k}$.}
\STATE{\textit{Generate uniform R.V.:} draw $U_k$ independently from the uniform distribution on $[0,1]$.}
\IF{$U_k\leq a((\theta_k,X^{1:N}_k,(\psi,Z^{1:N}))$, with $a(\cdot,\cdot)$ as in~\eqref{eq:acceptprobjoint}}
\STATE{\textit{Accept:} set $\theta_{k+1}:=\psi$ and $X_{k+1}^{1:N}:=Z^{1:N}$.}
\ELSE
\STATE{\textit{Reject:} set $\theta_{k+1}:=\theta_k$ and $X_{k+1}^{1:N}:=X_{k}^{1:N}$.}
\ENDIF
\ENDFOR
\end{algorithmic}
 \caption{The joint MH method.}
 \label{alg:jmh}
\end{algorithm}

We have so far failed to obtain analytical expressions for the resulting chain's stationary distributions. However, it is straightforward to find heuristic arguments suggesting that,  as $N\to\infty$, the $(\theta,x^n)$-marginal,  for any $n$, of these distributions approaches measures of the form $\delta_{\theta_*}(d\theta)p_{\theta_*}(x^n|y)dx^n$,  where $\theta_*$ is a stationary point of the marginal likelihood. In particular, note that the i.i.d.\ structure in $K_N^\theta$'s definition ensures that the  particle system is exchangeable (hence, $X_k^1,\dots,X_k^N$ all have the same law $q_k$). By propagation of chaos (e.g.\ see~\cite{Chaintron2022}), it is reasonable to expect that, for large $N$, the particle's empirical distribution, $q_k^N$ in~\eqref{eq:jmh1}, closely approximates  $q_k$:
\[q_k^N\approx q_k.\]
Suppose that the above holds exactly, and consider the resulting `idealized version' $(\tilde{\theta}_k,\tilde{X}^{1:N}_k)$ of the chain $(\theta_k,X^{1:N}_k)$ produced by Alg.~\ref{alg:jmh}: just as in~(\ref{eq:jmh1},\ref{eq:jmh2}), except that the empirical distribution $q_k^N$ in~\eqref{eq:jmh1} is replaced by the exact law $q_k$. The idealized chain's one-dimensional law $\mu_k(d\theta,dx^{1:N})$ satisfies
\begin{equation}\label{eq:ngeua9gneuagnuaw2}\mu_{k+1}(d\psi,dz^{1:N})=\int \mu_k(d\theta,dx^{1:N})P_{N}^{q_k}((\theta,x^{1:N}),(d\psi,dz^{1:N})),\end{equation}
where the `idealized kernel' $P_N^q$ is given by
\begin{align*}P_N^q((\theta,x^{1:N}),(d\psi,dz^{1:N}))=&a_N((\theta,x^{1:N}),(\psi,z^{1:N}))\delta_{u(\theta,q)}(d\psi)K^\theta_N(x^{1:N},z^{1:N})dz^{1:N}\\
&+[1-a((\theta,x^{1:N}),(\psi,z^{1:N}))]\delta_\theta(d\psi)\delta_{x^{1:N}}(dz^{1:N}).\end{align*}

It is then straightforward to show that, for any stationary point $\theta_*$ of the marginal likelihood, 
\[\pi_N^*(d\theta,dx):=\delta_{\theta_*}(d\theta)\prod_{n=1}^Np_{\theta_*}(x^n|y)dx^n\]
is a stationary distribution of the idealized chain in the sense that
\begin{equation}\label{eq:jmhstat}(\tilde{\theta}_0,\tilde{X}^{1:N}_0)\sim \pi_N^*\quad\Rightarrow \quad(\tilde{\theta}_k,\tilde{X}^{1:N}_k)\sim \pi_N^*\enskip\forall k=1,2,\dots\end{equation}
In particular, detailed balance holds: for any $(\theta,x^{1:N})\neq(\psi,z^{1:N})$ satisfying 
\begin{align}&\qquad\qquad K^\psi_N(z^{1:N},x^{1:N})\prod_{n=1}^Np_{\psi}(z^n,y)\leq K^\theta_N(x^{1:N},z^{1:N})\prod_{n=1}^Np_{\theta}(x^n,y),\label{eq:dna8dbnsa6y8bdsaydsadsa2}\\
&\Rightarrow a_N((\theta,x^{1:N}),(\psi,z^{1:N}))=\frac{K^\psi_N(z^{1:N},x^{1:N})}{K^\theta_N(x^{1:N},z^{1:N})}\prod_{n=1}^N\frac{p_{\psi}(z^n,y)}{p_{\theta}(x^n,y)},\quad a_N((\psi,z^{1:N}),(\theta,x^{1:N}))=1,\nonumber\end{align}
we have that, with $\pi_*(\cdot):=p_{\theta_*}(\cdot|y)$,
\begin{align*}
&\pi_N^*(d\theta,dx^{1:N})P^{\pi_*}_N((\theta,x^{1:N}),(d\psi,dz^{1:N}))\\
&=\delta_{\theta_*}(d\theta)\left(\prod_{n=1}^N p_{\theta_*}(x^n|y)\right)a_N((\theta,x^{1:N}),(\psi,z^{1:N}))\delta_{u(\theta,\pi_*)}(d\psi)K^\theta_N(x^{1:N},z^{1:N})dx^{1:N}dz^{1:N}\\
&=\delta_{\theta_*}(d\theta)\left(\prod_{n=1}^N p_{\theta_*}(x^n|y)\frac{p_\psi(z^n,y)}{p_\theta(x^n,y)}\right)\delta_{u(\theta,\pi_*)}(d\psi)K^{\psi}_N(z^{1:N},x^{1:N})dx^{1:N}dz^{1:N}\\
&=\delta_{\theta_*}(d\theta)\left(\prod_{n=1}^N p_{\theta_*}(x^n|y)\frac{p_\psi(z^n,y)}{p_{\theta_*}(x^n,y)}\right)\delta_{u(\theta_*,\pi_*)}(d\psi)K^{\psi}_N(z^{1:N},x^{1:N})dx^{1:N}dz^{1:N}\\
&=\delta_{u(\theta_*,\pi_*)}(d\theta)\left(\prod_{n=1}^N \frac{p_{\theta_*}(z^n,y)}{p_{\theta_*}(y)}\right)\delta_{\theta_*}(d\psi)K^{\psi}_N(z^{1:N},x^{1:N})dx^{1:N}dz^{1:N}\\
&=\delta_{u(\psi,\pi_*)}(d\theta)\left(\prod_{n=1}^Np_{\theta_*}(z^n|y)\right)\delta_{\theta_*}(d\psi)K^{\psi}_N(z^{1:N},x^{1:N})dx^{1:N}dz^{1:N}\\
&=\delta_{u(\psi,\pi_*)}(d\theta)\left(\prod_{n=1}^Np_{\theta_*}(z^n|y)\right)a_N((\psi,z^{1:N}),(\theta,x^{1:N}))\delta_{\theta_*}(d\psi)K^{\psi}_N(z^{1:N},x^{1:N})dx^{1:N}dz^{1:N}z\\
&=\pi_N^*(d\psi,dz^{1:N})P_N^{\pi_*}((\psi,z^{1:N}),(d\theta,dx^{1:N})).
\end{align*}
Reversing the roles of $(\theta,x^{1:N})$ and $(\psi,z^{1:N})$, we find that the above also holds should the inequality~\eqref{eq:dna8dbnsa6y8bdsaydsadsa2} be reversed. Because it holds trivially if $(\theta,x^{1:N})=(\psi,z^{1:N})$, integrating both sides over $(\theta,x^{1:N})$, we find that 
\[\int\pi_N^*(d\theta,dx^{1:N})P^{\pi_*}_N((\theta,x^{1:N}),(d\psi,dz^{1:N}))=\pi_N^*(d\psi,dz^{1:N});\]
and~\eqref{eq:jmhstat} follows from~\eqref{eq:ngeua9gneuagnuaw2}.
\end{document}